\documentclass[draftcls,onecolumn]{IEEEtran}
\ifCLASSINFOpdf
\else
\fi
\usepackage{microtype}
\usepackage{graphicx}
\usepackage{subfigure}
\usepackage{booktabs,pifont} % for professional tables

\usepackage{empheq}

\usepackage{hyperref}

\usepackage{pgfplots}
\pgfplotsset{compat=1.12}
\usetikzlibrary{shapes,arrows}
\usetikzlibrary{positioning}
\usepackage{tikz}
\usetikzlibrary{positioning,chains,fit,shapes,calc}
\usepackage{amsmath,bm,times}
\usetikzlibrary{calc}
\usepackage{dsfont}
\usepackage{cuted}
\usepackage{caption}

\usepackage{mathtools}
\DeclarePairedDelimiter{\ceil}{\lceil}{\rceil}
\DeclarePairedDelimiter{\floor}{\lfloor}{\rfloor}
\usepackage{amsthm}
\usepackage{comment}
\usepackage{enumitem}
\usepackage{arydshln}
\usepackage{multirow}
\usepackage{calc}
\usepackage{blkarray}
\usepackage{url}
\theoremstyle{definition}
\newtheorem{theorem}{Theorem}

\newtheorem{lemma}{Lemma}
\newtheorem{claim}{Claim}

\newtheorem{example}{Example}
\newtheorem{remark}{Remark}
\newtheorem{definition}{Definition}
\usepackage{graphicx}
\usepackage{dblfloatfix}
\usetikzlibrary{decorations.pathreplacing}
\usepackage{blindtext, graphicx, amsfonts,
	amssymb,multirow,epstopdf}
\def\BibTeX{{\rm B\kern-.05em{\sc i\kern-.025em b}\kern-.08em
    T\kern-.1667em\lower.7ex\hbox{E}\kern-.125emX}}

\makeatletter
\renewcommand*\env@matrix[1][*\c@MaxMatrixCols c]{%
  \hskip -\arraycolsep
  \let\@ifnextchar\new@ifnextchar
  \array{#1}}
\makeatother
\usepackage[linesnumbered,ruled]{algorithm2e}

\newcommand\bovermat[2]{%
    \makebox[0pt][l]{$\smash{\overbrace{\phantom{%
                    \begin{matrix}#2\end{matrix}}}^{\text{#1}}}$}#2}

\newcommand\undermat[2]{% http://tex.stackexchange.com/a/102468/5764
  \makebox[0pt][l]{$\smash{\underbrace{\phantom{%
    \begin{matrix}#2\end{matrix}}}_{\text{$#1$}}}$}#2}

\newcommand{\calS}{\mathcal{S}}
\newcommand{\calB}{\mathcal{B}}

\newcommand{\calA}{\mathcal{A}}
\newcommand{\calT}{\mathcal{T}}

\newcommand{\calW}{\mathcal{W}}
\newcommand{\calP}{\mathcal{P}}
\newcommand{\calX}{\mathcal{X}}
\newcommand{\calG}{\mathcal{G}}
\newcommand{\calH}{\mathcal{H}}
\newcommand{\calI}{\mathcal{I}}
\newcommand{\calJ}{\mathcal{J}}
\newcommand{\calU}{\mathcal{U}}
\newcommand{\calV}{\mathcal{V}}
\newcommand{\calR}{\mathcal{R}}
\newcommand{\calN}{\mathcal{N}}
\newcommand{\bfU}{\mathbf{U}}
\newcommand{\bfC}{\mathbf{C}}

\newcommand{\bfV}{\mathbf{V}}

\newcommand{\bfA}{\mathbf{A}}

\newcommand{\bfG}{\mathbf{G}}

\newcommand{\bfu}{\mathbf{u}}

\newcommand{\bfx}{\mathbf{x}}
\newcommand{\bfw}{\mathbf{w}}

\newcommand{\bfB}{\mathbf{B}}
\newcommand{\bfI}{\mathbf{I}}
\newcommand{\bfz}{\mathbf{z}}

\newcommand{\bfR}{\mathbf{R}}

\newcommand{\bfv}{\mathbf{v}}

\newcommand{\cmark}{\ding{51}}%
\newcommand{\xmark}{\ding{55}}

\newcommand{\aditya}[1]{\marginpar{+}{\bf Aditya's remark}: {\em #1}}

\begin{document}
%\title{Leveraging Stragglers in Distributed Matrix Computations}
\title{Coded sparse matrix computation schemes that leverage partial stragglers}

\definecolor{mygr}{rgb}{0.6,0.4,0.0}
\definecolor{my1color}{rgb}{0.6,0.4,0.0}
\definecolor{mycolor1}{rgb}{0.00000,0.44700,0.74100}%
\definecolor{mycolor2}{rgb}{0.85000,0.32500,0.09800}%
\tikzset{
block/.style    = {draw, thick, rectangle, minimum height = 2em, minimum width = 2em},
sum/.style      = {draw, circle, node distance = 1cm},
sum1/.style      = {draw, circle, minimum size = 1.1 cm},
input/.style    = {coordinate},
output/.style   = {coordinate},
}

\author{\IEEEauthorblockN{Anindya Bijoy Das and Aditya Ramamoorthy} \\
\IEEEauthorblockA{Department of Electrical and Computer Engineering,\\
Iowa State University, Ames, IA 50011 USA\\
\texttt{\{abd149,adityar\}@iastate.edu}
}
\thanks{This work was supported in part by the National Science Foundation (NSF) under grants CCF-1718470 and CCF-1910840. The material in this work has appeared in part at the 2021 IEEE International Symposium on Information Theory, Melbourne, Australia and at the 2018 IEEE Information Theory Workshop (ITW), Guangzhou, China.}}

%{Shell \MakeLowercase{\textit{et al.}}: Bare Demo of IEEEtran.cls for IEEE Transactions on Magnetics Journals}



% for Transactions on Magnetics papers, we must declare the abstract and
% index terms PRIOR to the title within the \IEEEtitleabstractindextext
% IEEEtran command as these need to go into the title area created by

\IEEEtitleabstractindextext{%
\begin{abstract}
Distributed matrix computations over large clusters can suffer from the problem of slow or failed worker nodes (called stragglers) which can dominate the overall job execution time. Coded computation utilizes concepts from erasure coding to mitigate the effect of stragglers by running “coded” copies of tasks comprising a job; stragglers are typically treated as erasures. While this is useful, there are issues with applying, e.g., MDS codes in a straightforward manner. Several practical matrix computation scenarios involve sparse matrices. MDS codes typically require dense linear combinations of submatrices of the original matrices which destroy their inherent sparsity. This is problematic as it results in significantly higher worker computation times. Moreover, treating slow nodes as erasures ignores the potentially useful partial computations performed by them. Furthermore, some MDS techniques also suffer from significant numerical stability issues.
In this work we present schemes that allow us to leverage partial computation by stragglers while imposing constraints on the level of coding that is required in generating the encoded submatrices. This significantly reduces the worker computation time as compared to previous approaches and results in improved numerical stability in the decoding process. Exhaustive numerical experiments on Amazon Web Services (AWS) clusters support our findings.
\end{abstract}

}

\maketitle

\IEEEdisplaynontitleabstractindextext

% \fi
\IEEEpeerreviewmaketitle

\section{Introduction}

\label{sec:intro}
% \IEEEPARstart{D}{istributed}
Distributed computation plays a major role in several problems in machine learning. For example, large scale matrix-vector multiplication is repeatedly used in gradient descent which in turn plays a key role in high dimensional machine learning problems. The size of the underlying matrices makes it impractical to perform the computation on a single computer (both from a speed and a storage perspective). Thus, the computation is typically subdivided into smaller tasks that are run in parallel across multiple worker nodes.

%different problems in signal processing and machine learning where large scale numerical calculations are required

%, for example, matrix-vector multiplication in the gradient descent problem in optimization algorithms. For a single machine, a large numerical process takes quite a long time to finish which can be done significantly faster if the same process is carried out in a distributed fashion \cite{zaharia2010spark}.

In these systems the overall execution time is typically dominated by the speed of the slowest worker. Thus, the presence of stragglers (as slow or failed workers are called) can negatively impact the performance of distributed computation. In recent years, techniques from coding theory (especially maximum-distance-separable (MDS) codes) \cite{lee2018speeding,dutta2016short,yu2017polynomial,tandon2017gradient} have been used to mitigate the effect of stragglers for problems such as matrix-vector and matrix-matrix multiplication. For instance, the work of \cite{lee2018speeding} proposes to partition the computation of $\bfA^T \bfx$ by first splitting $\bfA = [\bfA_0 ~|~ \bfA_1]$ into two block-columns (with an equal number of column vectors) and assigning three workers, the task of computing $\textbf{A}^T_0\textbf{x}$, $\textbf{A}^T_1\textbf{x}$ and $\left(\textbf{A}_0+\textbf{A}_1\right)^T\textbf{x}$, respectively. Evidently, the computational load on each node is half of the original job. Furthermore, it is easy to see that $\bfA^T \bfx$ can be recovered as soon as any two workers complete their tasks (with some minimal post-processing). Thus, this system is resilient to one straggler. The work of \cite{yu2017polynomial}, poses the multiplication of two matrices in a form that is roughly equivalent to a Reed-Solomon code. In particular, each worker node's task (which is multiplying smaller submatrices) can be imagined as a coded symbol. As long as enough tasks are complete, the master node can recover the matrix product by polynomial interpolation.

For such coded computing systems we can define a so-called recovery threshold. It is the minimum value of $\tau$, such that the master node can recover the result as long as {\it any} $\tau$ workers complete their tasks. Thus, at the top level, in these systems stragglers are treated as the equivalent of erasures in coding theory, i.e., the assumption is that no useful information can be obtained from the stragglers.

While these are interesting ideas, there are certain issues that are ignored in the majority of prior work (see \cite{kiani2018exploitation, mallick2018rateless, wang2018coded, 9252114, 8683267} for some exceptions). Firstly, several practical cases of matrix-vector or matrix-matrix multiplication involve sparse matrices. Using MDS coding strategies in a straightforward manner will often destroy the sparsity of the matrices being processed by the worker nodes. In fact, as noted in \cite{wang2018coded}, this can cause the overall job execution time to actually go up rather than down. Secondly, in the distributed computation setting, we make the observation that it is possible to leverage partial computations that are performed by the stragglers. Thus, a slow worker may not necessarily be a useless worker. Fig. \ref{strtime} (which also appears in \cite{das2019random}) shows the variation of speed of different {\tt t2.micro} machines in AWS (Amazon Web Services) cluster, and it can be seen that for a particular job, even the slowest worker node may have approximately $60\% - 70\%$ of the speed of the fastest worker.

\begin{figure}[t]
\centering
\definecolor{mycolor1}{rgb}{0.80000,0.50000,0.60000}%
\definecolor{mycolor2}{rgb}{0.40000,0.40000,0.80000}%
\centering
\captionsetup{justification=centering}
%\resizebox{0.99\linewidth}{!}{
\resizebox{0.69\linewidth}{!}{
\begin{tikzpicture}

\begin{axis}[%
width=5.1in,
height=3.203in,
at={(2.6in,0.756in)},
scale only axis,
xmin=0.5,
xmax=40.5,
xlabel style={font=\color{white!15!black}, font=\Large},
xlabel={Worker Index},
ymin=5.75,
ymax=8.75,
ylabel style={font=\color{white!15!black}, font=\Large},
ylabel={Time Required},
ytick={6, 6.5,7,7.5,8,8.5},
xtick={5,10,15,20,25,30,35,40},
tick label style={font=\Large},
axis background/.style={fill=white},
%axis x line=bottom,
%axis y line=left,
ymajorgrids,
legend style={nodes={scale=1.7}, at={(0.96,0.95)}, legend cell align=left, align=left, draw=white!15!black}
]
\addplot [color=gray, draw=none, mark=o, mark options={solid, gray}]
 plot [error bars/.cd, y dir = both, y explicit]
 table[row sep=crcr, y error plus index=2, y error minus index=3]{%
1	5.98581658840179	0.959149250984192	0.0465415620803835\\
2	6.00359764814377	0.839434473514557	0.0741278243064878\\
3	6.00716423034668	1.24595690727234	0.0672843360900881\\
4	5.97045059204102	0.631181526184082	0.045343589782715\\
5	5.97794937849045	0.582720758914948	0.039126393795013\\
6	5.97529013156891	0.923119711875915	0.0508352041244509\\
7	5.98533693313599	0.743480052947998	0.072195920944214\\
8	6.58506104230881	1.42492607355118	0.572207181453705\\
9	5.97961959838867	0.851500415802002	0.0456934928894039\\
10	5.980953540802	0.220931444168091	0.0493646526336669\\
11	5.98451259851456	0.546078412532807	0.0332155537605283\\
12	5.92907848358154	0.136641597747802	0.0791985034942631\\
13	6.0029022192955	1.23922281503677	0.0634812808036802\\
14	6.17343792200089	0.344864995479583	0.10947188615799\\
15	5.98008903503418	1.77345195770264	0.0548890495300292\\
16	5.98003795862198	0.804897038936615	0.0608618569374082\\
17	5.96088913917542	0.907340965270996	0.0543842697143555\\
18	6.01432245492935	1.43969160795212	0.0612913918495179\\
19	5.96912997484207	0.49222094297409	0.0366241288185121\\
20	5.9702387046814	0.4424183177948	0.0319636058807369\\
21	5.96883571147919	0.422864317893982	0.0247136354446411\\
22	5.99740345239639	1.02502148389816	0.0413964486122129\\
23	5.97340108156204	0.767411935329437	0.0416791558265688\\
24	5.98424435377121	0.63846277475357	0.0466013026237491\\
25	5.96148981332779	0.432592256069183	0.0372907090187073\\
26	5.94988328456879	0.330657658576965	0.0523052835464481\\
27	6.00984258890152	2.70222930669785	0.0753967308998105\\
28	5.9754453420639	0.991303610801697	0.0492685556411745\\
29	5.97305377244949	0.527073304653168	0.0423777890205379\\
30	5.96291655302048	0.345107614994049	0.039221465587616\\
31	5.98049085855484	0.163976018428802	0.0347806763648988\\
32	6.01313875675201	1.89541222095489	0.0715987539291385\\
33	5.98265320062637	0.366771876811981	0.0487893223762512\\
34	6.01112633943558	1.5211336016655	0.056864321231842\\
35	6.00164761781692	0.753265473842621	0.0405026936531065\\
36	5.95397742509842	0.773488442897797	0.0572324585914608\\
37	5.97822363615036	0.491628334522248	0.0391076350212094\\
38	5.96405220508575	0.256363863945007	0.0437552976608275\\
39	5.95139958620071	0.876449325084686	0.0582823967933654\\
40	5.96683704376221	0.959806089401245	0.0413140010833741\\
};
\addlegendentry{Average}

\addplot [color=mycolor1, line width=6.0pt]
  table[row sep=crcr]{%
1	5.93927502632141\\
1	6.94496583938599\\
};

\addplot [color=mycolor1, line width=6.0pt]
  table[row sep=crcr]{%
2	5.92946982383728\\
2	6.84303212165833\\
};

\addplot [color=mycolor1, line width=6.0pt]
  table[row sep=crcr]{%
3	5.93987989425659\\
3	7.25312113761902\\
};

\addplot [color=mycolor1, line width=6.0pt]
  table[row sep=crcr]{%
4	5.9251070022583\\
4	6.6016321182251\\
};

\addplot [color=mycolor1, line width=6.0pt]
  table[row sep=crcr]{%
5	5.93882298469543\\
5	6.5606701374054\\
};

\addplot [color=mycolor1, line width=6.0pt]
  table[row sep=crcr]{%
6	5.92445492744446\\
6	6.89840984344482\\
};

\addplot [color=mycolor1, line width=6.0pt]
  table[row sep=crcr]{%
7	5.91314101219177\\
7	6.72881698608398\\
};

\addplot [color=mycolor1, line width=6.0pt]
  table[row sep=crcr]{%
8	6.0128538608551\\
8	8.00998711585999\\
};

\addplot [color=mycolor1, line width=6.0pt]
  table[row sep=crcr]{%
9	5.93392610549927\\
9	6.83112001419067\\
};

\addplot [color=mycolor1, line width=6.0pt]
  table[row sep=crcr]{%
10	5.93158888816833\\
10	6.20188498497009\\
};

\addplot [color=mycolor1, line width=6.0pt]
  table[row sep=crcr]{%
11	5.95129704475403\\
11	6.53059101104736\\
};

\addplot [color=mycolor1, line width=6.0pt]
  table[row sep=crcr]{%
12	5.84987998008728\\
12	6.06572008132935\\
};

\addplot [color=mycolor1, line width=6.0pt]
  table[row sep=crcr]{%
13	5.93942093849182\\
13	7.24212503433228\\
};

\addplot [color=mycolor1, line width=6.0pt]
  table[row sep=crcr]{%
14	6.0639660358429\\
14	6.51830291748047\\
};

\addplot [color=mycolor1, line width=6.0pt]
  table[row sep=crcr]{%
15	5.92519998550415\\
15	7.75354099273682\\
};

\addplot [color=mycolor1, line width=6.0pt]
  table[row sep=crcr]{%
16	5.91917610168457\\
16	6.78493499755859\\
};

\addplot [color=mycolor1, line width=6.0pt]
  table[row sep=crcr]{%
17	5.90650486946106\\
17	6.86823010444641\\
};

\addplot [color=mycolor1, line width=6.0pt]
  table[row sep=crcr]{%
18	5.95303106307983\\
18	7.45401406288147\\
};

\addplot [color=mycolor1, line width=6.0pt]
  table[row sep=crcr]{%
19	5.93250584602356\\
19	6.46135091781616\\
};

\addplot [color=mycolor1, line width=6.0pt]
  table[row sep=crcr]{%
20	5.93827509880066\\
20	6.4126570224762\\
};

\addplot [color=mycolor1, line width=6.0pt]
  table[row sep=crcr]{%
21	5.94412207603455\\
21	6.39170002937317\\
};

\addplot [color=mycolor1, line width=6.0pt]
  table[row sep=crcr]{%
22	5.95600700378418\\
22	7.02242493629456\\
};

\addplot [color=mycolor1, line width=6.0pt]
  table[row sep=crcr]{%
23	5.93172192573547\\
23	6.74081301689148\\
};

\addplot [color=mycolor1, line width=6.0pt]
  table[row sep=crcr]{%
24	5.93764305114746\\
24	6.62270712852478\\
};

\addplot [color=mycolor1, line width=6.0pt]
  table[row sep=crcr]{%
25	5.92419910430908\\
25	6.39408206939697\\
};

\addplot [color=mycolor1, line width=6.0pt]
  table[row sep=crcr]{%
26	5.89757800102234\\
26	6.28054094314575\\
};

\addplot [color=mycolor1, line width=6.0pt]
  table[row sep=crcr]{%
27	5.93444585800171\\
27	8.71207189559937\\
};

\addplot [color=mycolor1, line width=6.0pt]
  table[row sep=crcr]{%
28	5.92617678642273\\
28	6.9667489528656\\
};

\addplot [color=mycolor1, line width=6.0pt]
  table[row sep=crcr]{%
29	5.93067598342896\\
29	6.50012707710266\\
};

\addplot [color=mycolor1, line width=6.0pt]
  table[row sep=crcr]{%
30	5.92369508743286\\
30	6.30802416801453\\
};

\addplot [color=mycolor1, line width=6.0pt]
  table[row sep=crcr]{%
31	5.94571018218994\\
31	6.14446687698364\\
};

\addplot [color=mycolor1, line width=6.0pt]
  table[row sep=crcr]{%
32	5.94154000282288\\
32	7.90855097770691\\
};

\addplot [color=mycolor1, line width=6.0pt]
  table[row sep=crcr]{%
33	5.93386387825012\\
33	6.34942507743835\\
};

\addplot [color=mycolor1, line width=6.0pt]
  table[row sep=crcr]{%
34	5.95426201820374\\
34	7.53225994110107\\
};

\addplot [color=mycolor1, line width=6.0pt]
  table[row sep=crcr]{%
35	5.96114492416382\\
35	6.75491309165955\\
};

\addplot [color=mycolor1, line width=6.0pt]
  table[row sep=crcr]{%
36	5.89674496650696\\
36	6.72746586799622\\
};

\addplot [color=mycolor1, line width=6.0pt]
  table[row sep=crcr]{%
37	5.93911600112915\\
37	6.46985197067261\\
};

\addplot [color=mycolor1, line width=6.0pt]
  table[row sep=crcr]{%
38	5.92029690742493\\
38	6.22041606903076\\
};

\addplot [color=mycolor1, line width=6.0pt]
  table[row sep=crcr]{%
39	5.89311718940735\\
39	6.8278489112854\\
};

\addplot [color=mycolor1, line width=6.0pt]
  table[row sep=crcr]{%
40	5.92552304267883\\
40	6.92664313316345\\
};

\addplot [color=black, draw=none, mark=*, mark options={solid, fill=mycolor2, black}]
  table[row sep=crcr]{%
1	5.98581658840179\\
2	6.00359764814377\\
3	6.00716423034668\\
4	5.97045059204102\\
5	5.97794937849045\\
6	5.97529013156891\\
7	5.98533693313599\\
8	6.58506104230881\\
9	5.97961959838867\\
10	5.980953540802\\
11	5.98451259851456\\
12	5.92907848358154\\
13	6.0029022192955\\
14	6.17343792200089\\
15	5.98008903503418\\
16	5.98003795862198\\
17	5.96088913917542\\
18	6.01432245492935\\
19	5.96912997484207\\
20	5.9702387046814\\
21	5.96883571147919\\
22	5.99740345239639\\
23	5.97340108156204\\
24	5.98424435377121\\
25	5.96148981332779\\
26	5.94988328456879\\
27	6.00984258890152\\
28	5.9754453420639\\
29	5.97305377244949\\
30	5.96291655302048\\
31	5.98049085855484\\
32	6.01313875675201\\
33	5.98265320062637\\
34	6.01112633943558\\
35	6.00164761781692\\
36	5.95397742509842\\
37	5.97822363615036\\
38	5.96405220508575\\
39	5.95139958620071\\
40	5.96683704376221\\
};
\addlegendentry{Bounds}

\end{axis}

\end{tikzpicture}
}
\caption{\small Variation of worker speeds for the same job over 100 runs across $40$ workers within AWS; the job involves multiplying two random matrices of size $4000 \times 4000$ twice. The average time is shown by the small circle for each worker. The upper and lower edges indicate the maximum and minimum time over the 100 runs. The required time exhibits a wide variation from $5.85$ seconds to $8.71$ seconds.}
\label{strtime}
\end{figure}
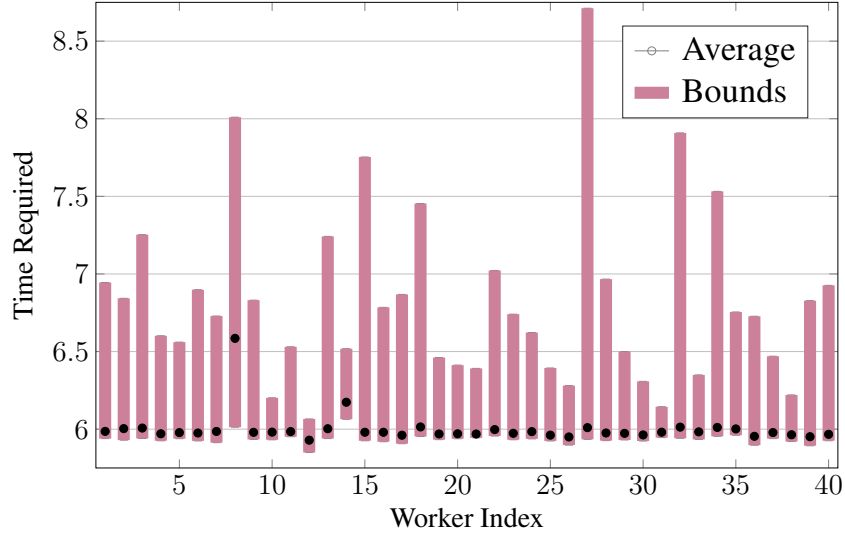 

In this work we propose schemes which are not only resilient to full stragglers, but can also exploit slow workers by utilizing their partially finished tasks. The works in \cite{c3les} and \cite{8849451} also address this issue but they are applicable only for matrix-vector multiplication whereas in this work, we propose schemes for matrix-matrix multiplication too. Furthermore, in several of our schemes we can specify the number of block-columns of the individual $\bfA$ and $\bfB$ matrices that are linearly combined to arrive at the encoded matrices. This is especially useful in the case of sparse matrices ($\bfA$ and $\bfB$) that often appear in practical settings. Thus, in short, our proposed approaches can leverage the partial computations of the stragglers and exploit the sparsity of the input matrices, both of which can enhance the overall speed of the whole system. 
%In this work, we investigate the trade-offs between the coded and uncoded job assignments about how we can utilize the stragglers. 
%Finally we propose a scheme which is optimal in terms of resilience to full stragglers and can be very close to optimal in terms of the utilization of the partial stragglers.

This paper is organized as follows. Section \ref{sec:back_rel_work} describes the background and related work and summarizes the contributions of our work. Section \ref{sec:prel} outlines some basic definitions and observations which are required for the subsequent presentation. Section \ref{sec:beta_level} discusses our proposed $\beta$-level coding schemes which constrain the level of coding in the encoded submatrices while leveraging partial computations. Following this, Section \ref{sec:optimal} proposes schemes for both matrix-vector and matrix-matrix multiplication which can be optimal in terms of resilience to full stragglers and can improve the utilization of the partial stragglers. Section \ref{sec:numerical_exp} discusses the experimental performance of our proposed methods and shows the comparison with other available approaches. We conclude the paper with a discussion about future work in Section \ref{sec:conclusion}.

\section{Background and Related Work}
\label{sec:back_rel_work}

Consider the case where a master node has a matrix $\bfA$ and either a matrix $\bfB$ or a vector $\bfx$ and needs to compute either $\bfA^T \bfB$ or $\bfA^T \bfx$. The computation needs to be carried out in a distributed fashion over $n$ worker nodes. Each worker receives the equivalent of a certain fraction (denoted by $\gamma_A$ and $\gamma_B$, respectively) of the columns of $\bfA$ and $\bfB$ or the whole vector $\bfx$. The node is responsible for computing its assigned submatrix-submatrix or submatrix-vector products.

We discuss the matrix-matrix scenario below where each worker node receives coded versions of submatrices of $\bfA$ and $\bfB$ respectively\footnote{A general formulation need not restrict the assignment to coded submatrices of $\bfA$ and $\bfB$. Nevertheless, all known schemes thus far and our proposed schemes work with equal-sized submatrices, so we present the formulation in this way.}. The corresponding matrix-vector case can be obtained as a special case. %It computes pairwise products (either all or some subset thereof) of these and sends them to the master node which needs to decode to recover $\bfA^T \bfB$. 
Consider a $p \times u$ and $p \times v$ block decomposition of $\bfA$ and $\bfB$ respectively as shown below.
\begin{align*}
\bfA &= \begin{bmatrix}
\bfA_{0,0} &\dots& \bfA_{0,u-1}\\
\vdots & \ddots & \vdots \\
\bfA_{p-1,0} & \dots & \bfA_{p-1,u-1}
\end{bmatrix},  \; \; \;  \textrm{and} \; \; \; \bfB = \begin{bmatrix}
\bfB_{0,0} &\dots& \bfB_{0,v-1}\\
\vdots & \ddots & \vdots \\
\bfB_{p-1,0} & \dots & \bfB_{p-1,v-1}
\end{bmatrix}. %\label{eq:block_decomp_A_B}
\end{align*} The master node creates coded submatrices by computing appropriate scalar linear combinations of the $\bfA_{i,j}$ submatrices and respectively the $\bfB_{i,j}$ submatrices. This implies that the master node only performs scalar multiplications and additions. It is not responsible for any of the computationally intensive matrix operations.  Following this, it sends the corresponding coded submatrices to each of the workers who perform the matrix operations.

%\anindya{In the previous paragraph, linear combinations of the $\bfA_{i,j}$ and $\bfB_{i,j}$ submatrices may be misleading.}

In this work we only consider a decomposition of $\bfA$ and $\bfB$ into block-columns, i.e., $p=1$. We assume that the storage fraction $\gamma_A$ (or $\gamma_B$) can be expressed as $\ell_A/\Delta_A$ (likewise $\ell_B/\Delta_B$) where both $\ell_A$ and $\Delta_A$ (and $\ell_B$ and $\Delta_B$) are integers. We assume that $\bfA$ and $\bfB$ are large enough and satisfy divisibility constraints so that we can choose any large enough value of $\Delta_A$ and $\Delta_B$ to partition the columns of $\bfA$ and $\bfB$ into $\Delta_A$ and $\Delta_B$ block-columns. These are denoted as $\bfA_0, \bfA_1, \dots, \bfA_{\Delta_A-1}$ and $\bfB_0, \bfB_1, \dots, \bfB_{\Delta_B-1}$. Each node is assigned the equivalent of $\ell_A$ block-columns of $\bfA$ and $\ell_B$ block-columns of $\bfB$. Each of those $\ell_A$ block-columns from $\bfA$ will be multiplied with each of the $\ell_B$ block-columns from $\bfB$, so a particular worker node will compute, in total, $\ell = \ell_A \ell_B$ block-products for matrix-matrix multiplication. In case of matrix-vector multiplication, the worker node will compute $\ell = \ell_A$ block products, where each of $\ell_A$ blocks from $\bfA$ will be  multiplied with $\bfx$.

The assignment can simply be subsets of $\{\bfA_0, \bfA_1, \dots, \bfA_{\Delta_A-1}\}$ or $\{\bfB_0, \bfB_1, \dots, \bfB_{\Delta_B-1}\}$; in this case we call the solution ``uncoded". Alternatively, the assignment can be suitably chosen functions of $\{\bfA_0, \bfA_1, \dots, \bfA_{\Delta_A-1}\}$ or $\{\bfB_0, \bfB_1, \dots, \bfB_{\Delta_B-1}\}$; in this case we call the solution ``coded". The assignment also specifies a sequential order from top to bottom in which each worker node needs to process its tasks. This implies that if a node is currently processing the $i$-th assignment ($0 \leq i \leq \ell-1$), then it has already processed assignments $0$ through $i-1$. In this work, we  assume that each time a node computes a product, it transmits the result to the master node. As we shall show, the processing order matters in this problem.

There are two requirements that our system needs to have. The master node should be able to decode the intended result ($\bfA^T \bfB$ or $\bfA^T \bfx$) from any $n-s$ workers for $s$ as large as possible. i.e., $n-s$ is the recovery threshold of the scheme \cite{yu2017polynomial}. The second requirement is that the master node should be able to recover $\bfA^T \bfB$ or $\bfA^T \bfx$ as long as it receives {\it any} $Q$ products from the worker nodes. This formulation subsumes treating stragglers as non-working nodes. %Indeed, suppose that we want a system of $n$ workers that is resilient to $s$ failures. Then, a sufficient condition would be that $Q \leq (n-s) \ell$ in our system. 
To our best knowledge, this second requirement has not been examined systematically within the coded computation literature, even though it is a natural constraint that allows for succinct treatment of recovery in distributed computing clusters where the workers have differing speeds.

\begin{example}
\label{eg:initial_example}
Consider a system with $n=3$ worker nodes with $\gamma_A = 2/3$. We partition $\bfA$ into $\Delta_A=3$ block-columns and the assignment of block-columns to each node is shown in Fig. \ref{uncoded_toy} (this is an uncoded solution). We emphasize that the order of the computation also matters here, i.e., worker node $W_0$ (for example) computes $\bfA_0^T \bfx$ first and then $\bfA_1^T \bfx$. For the specific assignment it is clear that the computation is successful as long as any four block products are returned by the workers. Thus, for this system $Q = 4$.

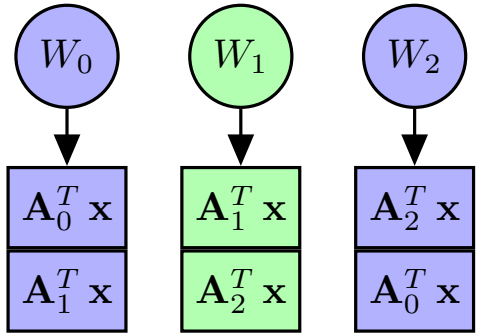
\begin{figure}[t]
\centering
\captionsetup{justification=centering}
\resizebox{0.4\linewidth}{!}{
\begin{tikzpicture}[auto, thick, node distance=2cm, >=triangle 45]

\draw
	node at (0,0)[right=-3mm]{}
	node [sum1, minimum size = 0.8cm, fill=blue!30] (blk1){$W_0$}
    node [sum1, minimum size = 0.8cm, fill=green!30,right = 0.6cm of blk1] (blk2) {$W_1$}
    node [sum1, minimum size = 0.8cm, fill=blue!30,right = 0.6cm of blk2] (blk3) {$W_2$}
    node [block, fill=blue!30,below = 0.5 cm of blk1] (blk11) {$\bfA_{0}^T \, \bfx$}
    node [block, fill=blue!30,below = 0.0005 cm of blk11] (blk12) {$\bfA_{1}^T\, \bfx$}
    node [block, fill=green!30,below = 0.5 cm of blk2] (blk21) {$\bfA_{1}^T\, \bfx$}
    node [block, fill=green!30,below = 0.0005 cm of blk21] (blk22) {$\bfA_{2}^T\, \bfx$}
    node [block, fill=blue!30,below = 0.5 cm of blk3] (blk31) {$\bfA_{2}^T\, \bfx$}
    node [block, fill=blue!30,below = 0.0005 cm of blk31] (blk32) {$\bfA_{0}^T\, \bfx$}
    ; 
\draw[->](blk1) -- node{} (blk11);
\draw[->](blk2) -- node{} (blk21);
\draw[->](blk3) -- node{} (blk31);

\end{tikzpicture}
}
\caption{\small Matrix $\bfA$ is partitioned into three submatrices. Each worker is assigned two of those uncoded submatrices. Here $Q = 4$.}
\label{uncoded_toy}
\end{figure} 

On the other hand, Fig. \ref{coded_toy} demonstrates a coded solution, where the bottom assignment in the workers are some suitably chosen functions of the elements of $\{\bfA_0^T \bfx, \bfA_1^T \bfx, \bfA_2^T \bfx\}$. For this assignment, it is obvious that the master node can recover $\bfA^T \bfx$ as long as any three block products are returned by the workers, so in this system $Q = 3$.
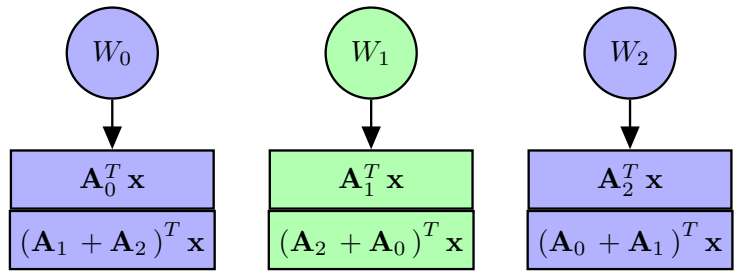
\begin{figure}[t]
\centering
\captionsetup{justification=centering}
\resizebox{0.6\linewidth}{!}{
\begin{tikzpicture}[auto, thick, node distance=2cm, >=triangle 45]

\draw
	node at (0,0)[right=-3mm]{}
	node [sum1, minimum size = 1.1cm,fill=blue!30] (blk1){$W_0$}
    node [sum1, minimum size = 1.1cm,fill=green!30, right = 2cm of blk1] (blk2) {$W_1$}
    node [sum1, minimum size = 1.1cm,fill=blue!30, right = 2cm of blk2] (blk3) {$W_2$}
    node [block, fill=blue!30, minimum width = 2.45cm, below = 0.6 cm of blk1] (blk11) {$\bfA_{0}^T \, \bfx$}
    node [block, fill=blue!30, below = 0.0005 cm of blk11] (blk12) {$ \left( \bfA_{1}\, + \bfA_{2}\, \right)^T \bfx$}
    node [block, fill=green!30, minimum width = 2.45cm,below = 0.6 cm of blk2] (blk21) {$\bfA_{1}^T\, \bfx$}
    node [block, fill=green!30, below = 0.0005 cm of blk21] (blk22) {$\left( \bfA_{2}\, + \bfA_{0}\, \right)^T \bfx$}
    node [block, fill=blue!30, minimum width = 2.45cm,below = 0.6 cm of blk3] (blk31) {$\bfA_{2}^T\, \bfx$}
    node [block, fill=blue!30, below = 0.0005 cm of blk31] (blk32) {$\left( \bfA_{0}\, + \bfA_{1}\, \right)^T \bfx$}
    ; 

\draw[->](blk1) -- node{} (blk11);
\draw[->](blk2) -- node{} (blk21);
\draw[->](blk3) -- node{} (blk31);

\end{tikzpicture}
}
\caption{\small Matrix $\bfA$ is partitioned into three submatrices. Each worker is assigned one uncoded and one coded task. Here $Q = 3$.}
\label{coded_toy}
\end{figure} 
\end{example}

For any time $t$, we let $w_i(t)$ represent the state of computation of the $i$-th worker node, i.e., $w_i(t)$ is a non-negative integer such that $0 \leq w_i(t) \leq \ell$  which represents the number of tasks that have been processed by worker node $i$. Thus, our system requirement states as long as $\sum_{i=0}^{n-1} w_i(t) \geq Q$, the master node should be able to determine $\bfA^T \bfB$ or $\bfA^T \bfx$. As $\Delta$, the number of unknowns to be recovered, is a parameter that can be chosen, our objective is to minimize the value of $Q/\Delta$ for such a system. For matrix-vector multiplication, $\Delta = \Delta_A$, whereas for matrix-matrix multiplication, $\Delta = \Delta_A \Delta_B$. This formulation minimizes the worst case overall computation performed by the worker nodes.

%\section{Background, Related Work and Summary of Contributions}
%\label{sec:lit}
\subsection{Related Work}
%\aditya{still needs more edits to make it comprehensive}
Several coded computation schemes have been proposed for matrix multiplication \cite{lee2018speeding, yu2020straggler, yu2017polynomial, dutta2016short, dutta2019optimal, mallick2018rateless, wang2018coded, c3les, TangKR19,  kiani2018exploitation}, most of which are designed to mitigate the full stragglers; see \cite{ramamoorthyDTMag20} for a tutorial overview. We illustrate the basic idea below using the polynomial code approach of \cite{yu2017polynomial} for a system with $n=5$ workers where each of these worker nodes can store $\gamma_A = \frac{1}{2}$ fraction of matrix $\bfA$ and $\gamma_B = \frac{1}{2}$ fraction of matrix $\bfB$. Consider $u = v = 2$ and $p = 1$, thus we partition both $\bfA$ and $\bfB$ into two block-columns $\bfA_0, \bfA_1$ and $\bfB_0, \bfB_1$ respectively. Next, we define two matrix polynomials as
\begin{align*}
\bfA(z) &= \bfA_0 + \bfA_1 z  \; \; \; \textrm{and} \; \; \; \bfB(z) = \bfB_0 + \bfB_1 z^2  ;\\
\textrm{so} \; \bfA^T(z) \bfB(z) & = \bfA^T_0 \bfB_0 + \bfA_1^T \bfB_0 z + \bfA_0^T \bfB_1 z^2 + \bfA_1^T \bfB_1 z^3 .
\end{align*} The master node evaluates these polynomial $\bfA(z)$ and  $\bfB(z)$ at distinct real values $z_0, z_1, \dots, z_{n-1}$, and sends the corresponding matrices to worker node $W_i$. Each worker node computes the product of its assigned submatrices. %shows the job assignment for five workers where $z_i = i + 1$, where the workers will compute the product of their respective submatrices. (see Fig. \ref{matmatfig} where $z_i = i+1$)
%\anindya{Don't we need to mention that we are using $z_i = i+1$ ?}
It follows that decoding at the master node is equivalent to decoding a degree-3 real-valued polynomial. Thus, the master node can recover $\bfA^T \bfB$ as soon as it receives the results from {\it any} four workers, i.e., in this example, the recovery threshold is, $\tau = 4$. When $\gamma_A = 1/k_A$ and $\gamma_B = 1/k_B$ and $p=1$, the work of \cite{yu2020straggler} shows that their scheme has a threshold $\tau = k_A k_B$ which is optimal. Random coding solutions for this problem were investigated in \cite{8919859}. Approaches based on convolutional coding  were presented in \cite{das2019random, 8849395}. In these schemes (analogous to linear block codes) there are systematic workers that only contain uncoded assignments and parity workers that contain coded assignments.

The case when $p > 1$ was considered in the work of \cite{dutta2016short,dutta2019optimal,yu2020straggler,TangKR19}. Structuring the computation in this manner increases the computational load on the workers and the communication load from the workers to the master node but can reduce the recovery threshold as compared to the case of $p=1$.

It is well-recognized that in several practical situations the underlying matrices $\bfA$ and $\bfB$ are sparse. Computing the inner product $\mathbf{a}^T \bfx$ of $n$-length vectors $\mathbf{a}$ and $\bfx$ where $\mathbf{a}$ has around $\delta n$ ($0 < \delta \ll 1$) non-zero entries takes $\approx 2 \delta n$ floating point operations (flops) as compared to $\approx 2 n$ flops in the dense case. In general, the encoding process within coded computation increases the number of non-zero entries in the resultant encoded matrices. For instance, polynomial evaluations of degree $d$ will increase the number of non-zero entries by approximately $d$ times. This results in a $d$-fold increase in the worker computation times which can be unacceptably high. Thus, it is important to consider schemes where the encoding only combines a limited number of submatrices.

\begin{example}
Consider an example with two large sparse matrices $\bfA$ and $\bfB$ both of whose sizes are $10,000 \times 10,000$. Both of them have sparsity $\sigma = 3\%$, i.e., randomly chosen approximately $3\%$ entries of $\bfA$ and $\bfB$ are non-zero (we have used {\tt MATLAB} command {\tt sprand} for this example). We partition matrices $\bfA$ and $\bfB$ into $4$ and $5$ block-columns, respectively. First we choose a block-column $\bfA_i$ and a block-column $\bfB_j$, and next we obtain two coded submatrices $\tilde{\bfA}_i$ and $\tilde{\bfB}_j$ which are random linear combinations of the uncoded block-columns of $\bfA$ and $\bfB$, respectively. Table \ref{compspar} shows that it is around $4$ times more expensive to compute the coded product than the uncoded product, although the sizes of the corresponding matrices are exactly the same. The reason is that the number of non-zero entries in the coded submatrices have gone up significantly. %Theoretically, the computation time for the small subtask should be $20$ times smaller than the original one, but because of the failure to leverage the sparsity, it is not possible in the dense coded approaches.
\end{example}

\begin{table}[t]
\caption{{\small Computation time for sparse matrix multiplication. We choose matrices $\bfA$ and $\bfB$ both of size $10,000 \times 10,000$. Both of them have sparsity $\sigma = 3\%$, thus randomly chosen $3\%$ entries of $\bfA$ and $\bfB$ are non-zero.}}
\label{compspar}
\begin{center}
\begin{small}
\begin{sc}
\begin{tabular}{c c}
\hline
\toprule
Job & Required time \\
\midrule
To compute $\bfA^T \bfB$ & $9.41$ seconds  \\ 
To compute (uncoded) $\bfA_i^T \bfB_j$ &  $0.58$ seconds \\
To compute (coded) $\tilde{\bfA}_i^T \tilde{\bfB}_j$ & $2.13$ seconds \\

\bottomrule
\end{tabular}
\end{sc}
\end{small}
\end{center}
\end{table}%

An important aspect of coded computation is ``numerical stability'' of the recovered result. Indeed, while coded computation borrows techniques from classical coding theory (over finite fields), it differs in the sense that the coded submatrices and the decoding operates over the reals. Over finite fields, the invertibility of a matrix is sufficient to solve a system of equations. In contrast, over the reals if the corresponding matrix is ill-conditioned, then the recovery will in general be inaccurate. It is well-recognized that real Vandermonde matrices corresponding to polynomial interpolation have condition numbers that grow exponentially in the matrix sizes. This is a serious issue with the polynomial-based approaches of \cite{yu2017polynomial}, \cite{yu2018straggler}. There have been some works that have addressed these issues \cite{das2019random, 8849451, 8919859, ramamoorthy2019numerically, 8849468, 9174440, 9513242} in part. %\aditya{double check if rateless codes including LT code citations have been included}

%Several works \cite{aliasgari2020private}, \cite{hasircioglu2021speeding} have addressed the issue on privacy breach along with the straggler mitigation issue. Here the goal is that no information about the matrices $\bfA$ or $\bfB$ can be obtained from any set $m$ workers, $m \geq 1$, in order to get rid of leakage of any sensitive information. Another class of codes \cite{keshtkarjahromi2018dynamic}, \cite{vedadi2021adaptive} assumes the nature of the cloud of workers to be heterogeneous and time-varying, so that the system may have access to different number of workers at different moments where the workers may have different speeds and/or different storage capacities. 

Yet another feature of the coded computation problem that distinguishes it from classical codes is the processing order. The worker nodes process the assigned tasks in a specific order, such that if a worker node is processing a given task, it has already completed the previously assigned tasks. Thus, at any given time the pattern of tasks that have been completed is restricted. Interestingly, codes for such systems have been investigated in \cite{rosenbloomT97,ganesan2007existence}. These ideas were adapted for the distributed matrix-vector multiplication problem in \cite{8849451}.

We note here that in principle using polynomial approaches can allow us to address both the optimal threshold and the optimal $Q/\Delta = 1$ by simply placing multiple evaluations of the polynomials at distinct points within each worker node. However, this approach is not practical, firstly because of numerical stability issues. Secondly, as discussed above when considering sparse $\bfA$ and $\bfB$ matrices, the polynomial approaches result in dense coded submatrices which can cause an unacceptable increase in the worker node computation times. Numerical experiments supporting these conclusions can be found in Section \ref{sec:numerical_exp}. 
%
%
%\aditya{discuss UDMs, works that deal with sparse matrices. Need to discuss random encoding strategies}
\subsection{Summary of Contributions}

The contributions of our work can be summarized as follows. 
\begin{itemize}
    \item We present a fine-grained model of the distributed matrix-vector and matrix-matrix multiplication that allows us to (i) leverage the slower workers using their partial computations and (ii) impose constraints on what extent coding is allowed in the solution. This allows us to capture a scenario where workers have differing speeds and the intended result can be recovered as long as the workers together complete a minimum number ($Q$) of the assigned tasks. This applies to the practically important case where the underlying matrices are sparse. The formulation leads to new questions within coded computing that to our best knowledge have not been investigated before systematically within the coded computing literature.
    \item We present systematic methods for both matrix-vector and matrix-matrix multiplication that address both the recovery threshold and the $Q/\Delta$ metric. For the uncoded assignment case, we present a lower bound on the performance of any scheme that our constructions are able to match.
    
    \item We have proposed two different schemes for distributed computations, first of which is named as $\beta$-level coding. In this approach, we have used resolvable combinatorial designs \cite{stinson2007combinatorial} to improve the recovery threshold and the $Q/\Delta$ metric over the uncoded approach. We have shown that the metrics can be further improved if we utilize certain relations among the blocks of different parallel classes within the resolvable designs.
    
    \item Prior work has demonstrated schemes with the optimal recovery threshold for certain storage fractions. In this work we present novel schemes that retain the optimal recovery threshold and also have low $Q/\Delta$ values.
    
    \item Finally, we present exhaustive experimental comparisons that demonstrate the benefit of our schemes while considering sparse matrices in terms of worker node computation times and numerical stability.
    
\end{itemize} In Table \ref{compboth}, we present a summary of the properties and the advantages of both of our proposed approaches, $\beta$-level coding and sparsely coded straggler (SCS) optimal scheme. Moreover, a detailed comparison of the properties of our methods with other available schemes is demonstrated in Table \ref{compare}. 

It should be noted that there are other issues within coded matrix computations. Several works \cite{aliasgari2020private}, \cite{hasircioglu2021speeding} have considered the issue of private computation along with the straggler mitigation issue. Here the goal is that no information about the matrices $\bfA$ or $\bfB$ can be obtained from any set of at most $m$ workers. Another class of codes \cite{keshtkarjahromi2018dynamic}, \cite{vedadi2021adaptive} assumes the workers to be heterogeneous and time-varying, so that the system may have access to different number of workers at different moments where the workers may have different speeds and/or different storage capacities. These issues are out of the scope of this paper.

\begin{table}[t]
\caption{{\small Comparison between our proposed approaches. Here $\gamma_A$ and $\gamma_B$ indicate the storage fractions for matrices $\bfA$ and $\bfB$, respectively, and $n$ denotes the total number of workers.}}
\label{compboth}
\begin{center}
\begin{small}
\begin{sc}
\begin{tabular}{c c c c}
\hline
\toprule
Approach & Properties & Parameter Regime & Advantages  \\ 
 \midrule
$\beta$-level & Combine $\beta$ submatrices  & For $\gamma_A = \frac{a_1}{a_2}$ and $\gamma_B = \frac{b_1}{b_2}$, & Assigned submatrices \\
Coding & of $\bfA$ (or $\bfB$) & we need $\; n = c a_2 b_2$ & are sparse \\
 \midrule
SCS Optimal  & Majority of assigned & For $\gamma_A = \frac{1}{k_A}$ and $\gamma_B = \frac{1}{k_B}$, & Optimal Recovery  \\
Scheme & submatrices are uncoded & where $k_A$ and $k_B$ are integers & Threshold ($k_A k_B$) \\

\bottomrule
\end{tabular}
\end{sc}
\end{small}
\end{center}
\end{table}%

\begin{table*}[t]
\caption{{\small Comparison with existing works.}} 
\vspace{-0.3 cm}
\label{compare}
\begin{center}
\begin{small}
\begin{sc}
\begin{tabular}{c c c c c c}
\hline
\toprule
\multirow{2}{1 cm}{Codes} & Mat-Mat & Optimal & Numerical & Partial & Sparsely\\
 & Mult? & Threshold? & Stability? & Comput? & Coded?\\
 \midrule
Repetition Codes & \cmark & \xmark & \cmark &  \xmark & \cmark\\ \hline
Rateless Codes \cite{mallick2018rateless} & \xmark & \xmark & \cmark &  \xmark & \xmark \\ \hline 
Prod. Codes \cite{lee2017high}, Factored Codes \cite{9513242}  & \cmark &\xmark & \cmark  & \xmark & \xmark\\ \hline
Polynomial Codes \cite{yu2017polynomial} & \cmark  & \cmark & \xmark & \xmark & \xmark\\ \hline
Biv. Hermitian Poly. Code \cite{hasirciouglu2020bivariate} & \cmark & \cmark & \xmark & \cmark & \xmark\\ \hline
Dynamic Hetero.-Aware Code \cite{keshtkarjahromi2018dynamic} & \xmark &\xmark & \cmark  & \cmark & \xmark\\ \hline
OrthoPoly \cite{8849468}, RKRP code\cite{8919859} & \cmark &\cmark & \cmark  & \xmark & \xmark\\ \hline
Conv. Code \cite{das2019random}, Circ. \& Rot. Mat. \cite{ramamoorthy2019numerically} & \cmark &\cmark & \cmark  & \xmark & \xmark\\ \hline
C$^3$LES \cite{c3les} & \xmark & \xmark & \cmark & \cmark & \cmark\\ \hline
{\bf $\beta$-level Coding (proposed)} & \cmark & \xmark & \cmark & \cmark & \cmark\\ \hline
{\bf SCS Optimal Scheme (proposed)} & \cmark & \cmark & \cmark & \cmark & \cmark\\

\bottomrule
\end{tabular}
\end{sc}
\end{small}
\end{center}
\end{table*}%

\section{Preliminaries}
\label{sec:prel}
In this section we discuss some basic facts and observations that serve to explain our proposed distributed matrix computation schemes.
%\subsection{Order of Matrix Computations}
Suppose that a given worker node is assigned encoded block-columns $\tilde{\bfA}_i, i = 0, 1, \dots, \ell_A-1$ and $\tilde{\bfB}_j, j = 0, 1, \dots, \ell_B-1$. The assignment also specifies a top to bottom order. For the matrix-vector problem, the node processes them simply in the order $\tilde{\bfA}^T_0 \bfx, \tilde{\bfA}^T_1 \bfx, \dots, \tilde{\bfA}^T_{\ell_A-1} \bfx$. On the other hand for the matrix-matrix problem the node computes in the order $\tilde{\bfA}_0^T \tilde{\bfB}_0, \tilde{\bfA}_0^T \tilde{\bfB}_1, \dots, \tilde{\bfA}_0^T \tilde{\bfB}_{\ell_B-1}, \tilde{\bfA}_1^T \tilde{\bfB}_0, \dots, \tilde{\bfA}_1^T \tilde{\bfB}_{\ell_B-1},$ $\dots, \tilde{\bfA}_{\ell_A-1}^T \tilde{\bfB}_{0}, \dots, \tilde{\bfA}_{\ell_A-1}^T \tilde{\bfB}_{\ell_B-1}$.
%In this work we examine schemes where the coding across the block-columns of $\bfA$ and $\bfB$ are {\it sparse}. This is motivated by important practical considerations ({\it cf.} Section \ref{sec:lit}). 
\begin{definition}
A coding scheme for distributed matrix computation is said to be a $\beta$-level coding scheme if the assigned block-columns are a linear combination of exactly $\beta$ block-columns of $\bfA$ and $\bfB$. The case of $\beta = 1$ represents an uncoded scheme.
\end{definition}
Our constructions leverage the properties of combinatorial structures known as resolvable designs \cite{stinson2007combinatorial}. 
\begin{definition}
A resolvable design is a pair $(\calX, \calA)$ where $\calX$ is a set of elements (called points) and $\calA$ is a family of non-empty subsets of $\calX$ (called blocks) that have the same cardinality. A subset $\calP \subset \calA$ in a design $(\calX, \calA)$ is called a parallel class if $\cup_{\{i: \calA_i \in \calP \}} \calA_i = \calX$ and  if $\calA_i \cap \calA_j = \emptyset$ for $\calA_i, \calA_j \in \calP$ when $i \neq j$. A partition of $\calA$ into several parallel classes is called a resolution and $(\calX, \calA)$ is said to be a resolvable design if $\calA$ has at least one resolution \cite{stinson2007combinatorial}. 
\end{definition}
A  resolvable design always exists if the cardinality of a block divides $|\calX|$.
\begin{example}
\label{ex:resol}
Let $\calX = \{0, 1, 2, 3 \}$ and $\calA = \{ \{0, 1\},\{0, 2\},\{0, 3\},\{1, 2\},\{1, 3\},\{2, 3\}\}$. Now $(\calX, \calA)$ forms a resolvable design with parallel classes, $\calP_0 = \{ \{0, 1\},\{2, 3\}\}$, $\calP_1 = \{ \{0, 2\},\{1, 3\}\}$ and $\calP_2 = \{ \{0, 3\},\{1, 2\}\}$.
\end{example}
We note that the specification of the ``incidence relations'' between the points and blocks of a design can also be shown by means of an incidence matrix.
\begin{definition}
The incidence matrix $\calN$ of a design $(\calX, \calA)$ is a $|\calX| \times |\calA|$ binary matrix such that the $(i,j)$-th entry is a $1$ if the $i$-th point is a member of the $j$-th block and zero, otherwise.
\end{definition}
%\aditya{Put in example of incidence matrix} Consider the same $\calX$ and $\calA$ as mentioned in Example \ref{ex:resol}. 
For example, the incidence matrix for the resolvable design in Example \ref{ex:resol} is given by %In this case, the incidence matrix can be written as
\begin{align*}
    \calN = \begin{bmatrix}
1 & 1 & 1 & 0 & 0 & 0 \\
1 & 0 & 0 & 1 & 1 & 0 \\
0 & 1 & 0 & 1 & 0 & 1 \\
0 & 0 & 1 & 0 & 1 & 1 \\
\end{bmatrix}.
\end{align*}

%\subsection{Cyclic Assignment}
We will use a cyclic assignment of tasks extensively in our constructions. We illustrate this by means of the following matrix-vector multiplication example.
\begin{example}
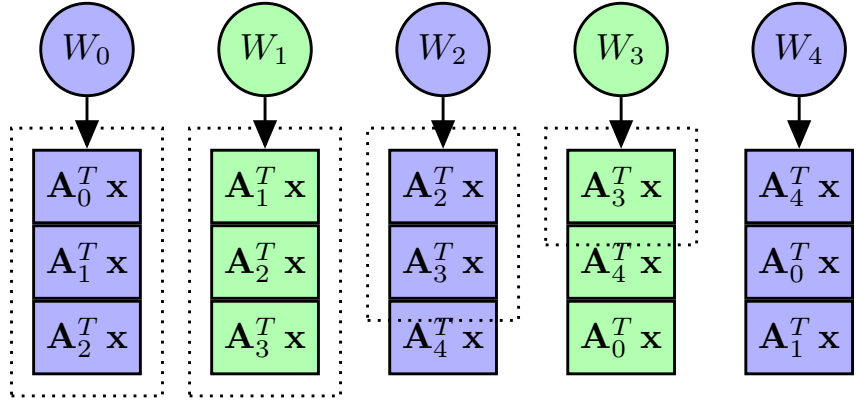
\begin{figure}[t]
\centering
\captionsetup{justification=centering}
\resizebox{0.7\linewidth}{!}{
\begin{tikzpicture}[auto, thick, node distance=2cm, >=triangle 45]

\draw
	node at (0,0)[right=-3mm]{}
	node [sum, fill=blue!30] (blk1){$W_0$}
    node [sum, fill=green!30,right = 0.8cm of blk1] (blk2) {$W_1$}
    node [sum, fill=blue!30,right = 0.8cm of blk2] (blk3) {$W_2$}
    node [sum, fill=green!30,right = 0.8cm of blk3] (blk4) {$W_3$}
    node [sum, fill=blue!30,right = 0.8cm of blk4] (blk5) {$W_4$}
    
    node [block, fill=blue!30,below = 0.5 cm of blk1] (blk11) {$\bfA_{0}^T \, \bfx$}
    node [block, fill=blue!30,below = 0.0005 cm of blk11] (blk12) {$\bfA_{1}^T\, \bfx$}
    node [block, fill=blue!30,below = 0.0005 cm of blk12] (blk13) {$\bfA_{2}^T\, \bfx$}

    node [block, fill=green!30,below = 0.5 cm of blk2] (blk21) {$\bfA_{1}^T\, \bfx$}
    node [block, fill=green!30,below = 0.0005 cm of blk21] (blk22) {$\bfA_{2}^T\, \bfx$}
    node [block, fill=green!30,below = 0.0005 cm of blk22] (blk23) {$\bfA_{3}^T\, \bfx$}

    node [block, fill=blue!30,below = 0.5 cm of blk3] (blk31) {$\bfA_{2}^T\, \bfx$}
    node [block, fill=blue!30,below = 0.0005 cm of blk31] (blk32) {$\bfA_{3}^T\, \bfx$}
    node [block, fill=blue!30,below = 0.0005 cm of blk32] (blk33) {$\bfA_{4}^T\, \bfx$}

    node [block, fill=green!30,below = 0.5 cm of blk4] (blk41) {$\bfA_{3}^T\, \bfx$}
    node [block, fill=green!30,below = 0.0005 cm of blk41] (blk42) {$\bfA_{4}^T\, \bfx$}
    node [block, fill=green!30,below = 0.0005 cm of blk42] (blk43) {$\bfA_{0}^T\, \bfx$}

    node [block, fill=blue!30,below = 0.5 cm of blk5] (blk51) {$\bfA_{4}^T\, \bfx$}
    node [block, fill=blue!30,below = 0.0005 cm of blk51] (blk52) {$\bfA_{0}^T\, \bfx$}
    node [block, fill=blue!30,below = 0.0005 cm of blk52] (blk53) {$\bfA_{1}^T\, \bfx$}
    ; 
\draw[->](blk1) -- node{} (blk11);
\draw[->](blk2) -- node{} (blk21);
\draw[->](blk3) -- node{} (blk31);
\draw[->](blk4) -- node{} (blk41);
\draw[->](blk5) -- node{} (blk51);

\draw[thick,dotted] ($(blk11.north west)+(-0.2,0.2)$)  rectangle ($(blk13.south east)+(0.2,-0.2)$);
\draw[thick,dotted] ($(blk21.north west)+(-0.2,0.2)$)  rectangle ($(blk23.south east)+(0.2,-0.2)$);
\draw[thick,dotted] ($(blk31.north west)+(-0.2,0.2)$)  rectangle ($(blk32.south east)+(0.2,-0.2)$);
\draw[thick,dotted] ($(blk41.north west)+(-0.2,0.2)$)  rectangle ($(blk41.south east)+(0.2,-0.2)$);

\end{tikzpicture}
}
\caption{\small Partitioning matrix $\bfA$ into five submatrices and assigning three uncoded tasks in a cyclic fashion to the workers. The system is resilient to two stragglers and $Q=10$. The tasks enclosed in dots can be processed without processing any copy of $\bfA_4^T \bfx$.}
\label{uncoded_matvec_prop}
\end{figure} 
Consider an example of computing $\bfA^T \bfx$, where we have $n = 5$ workers and each worker can process $\gamma = 3/5$ fraction of the total job. We partition $\bfA$ into $\Delta = 5$ block-columns: $\bfA_0, \bfA_1 , \dots, \bfA_4$. Let $\calX = \{0,1,2,3,4\}$. If we do not incorporate any coding among the block-columns, then for $\beta = 1$, we have the trivial parallel class $\calP = \{ \{0\}, \{1\}, \dots, \{4\}\}$. Fig. \ref{uncoded_matvec_prop} shows a cyclic assignment of jobs where three uncoded submatrices are allocated to each of the workers in a cyclic fashion according to the indices of three elements of $\calP$. It can be easily verified that the system is resilient to $s = 2$ stragglers. In the sequel, our assignment can be coded as well. 
\end{example}

%For the sake of generality, the discussion below proceeds in terms of symbols.
More generally, suppose that we have $\Delta$ symbols denoted $0, \dots, \Delta -1$, $n = \Delta$ worker nodes and $\ell$ symbols to be placed in each worker node where $\ell \leq \Delta$. The symbols can be encoded block-columns of $\bfA$ or the product of encoded block-columns of $\bfA$ and $\bfB$. A cyclic assignment in this case assigns the set $\{j, j+1, \dots, j + \ell-1\} \pmod \Delta$ to worker $W_j$; symbol $j$ appears at the top and sequentially symbol $(j + \ell -1)$ (the values are reduced modulo $\Delta$) at the bottom. The node $W_j$ processes the tasks specified by the symbols from top to bottom. Within a node, the position of a symbol is denoted by an integer between $0$ and $\ell-1$, where $0$ denotes the top and $\ell-1$ denotes the bottom. 
\begin{lemma}
\label{lem:cyclicQ}
The cyclic assignment satisfies the following properties.
\begin{itemize}
    \item Each symbol appears $\ell$ times across $n$ worker nodes. Furthermore, it appears in each position $0, \dots, \ell-1$ exactly once, across all $n$ workers.
    \item Let $\alpha_c$ be the maximum number of symbols that can be processed across all worker nodes such that a specific symbol $j$ is processed exactly $c$ times (where $0 \leq c \leq \ell$). Then, $\alpha_c = \Delta \ell - \frac{\ell(\ell + 1)}{2} + \sum_{i=0}^{c-1} (\ell - i)$, independent of $j$.
\end{itemize}
\end{lemma}
\begin{proof}
The first claim follows since $\ell \leq \Delta=n$ and symbol $j$, where $0 \leq j \leq \Delta -1$, appears in workers $j, j-1, j-2, \dots,  j- \ell + 1$ (indices reduced modulo-$\Delta$).

For the second claim we proceed by contradiction. Suppose that there is a symbol $j$ for which the condition is violated. From part $(a)$, symbol $j$ appears once in positions $0, \dots, \ell-1$ across the workers. Thus, one can process at most $(\Delta-\ell)\ell + \sum_{i=0}^{\ell-1} i = \Delta \ell - \frac{\ell(\ell+1)}{2}$ symbols without processing any copy of $j$. Following this, any symbol processed will necessarily process symbol $j$. If we process the copy of $j$ at position $i$, we can process another $\ell-1 - i$ symbols without processing another copy of $j$. Therefore, the maximum number of symbols that can be processed such that $c$ copies of $j$ are processed are
$\Delta \ell - \frac{\ell(\ell + 1)}{2} + \sum_{i=0}^{c-1} (\ell -i)$.
\end{proof}

\section{$\beta$-level Coding for Distributed Computations}
\label{sec:beta_level}
%We now explore coded schemes in our setting. As demonstrated in Example \ref{eg:initial_example}, the value of $Q$ in the coded scenario can be strictly lesser than in the uncoded case. A simple solution to this problem is to use MDS (maximum distance separable) codes as was done in some of the initial papers \cite{lee2018speeding} in this area. Namely, one could use an $(n\ell,\Delta)$-MDS code for some value of $\Delta$ and $n \ell \geq \Delta$. The unknowns will be combined using the corresponding generator matrix (with an appropriate mapping from the finite field to the real field). It is clear that in this case, the master node can compute $\bfA^T \bfx$ or $\bfA^T \bfB$ as long as any $\Delta$ blocks are processed. However, such codes may require rather dense linear combinations of the columns of $\bfA$ and $\bfB$ and this may incur a significant overhead in the computation performed by the worker nodes, especially in the practical case when $\bfA$ and/or $\bfB$ is sparse to begin with. Thus in order to develop schemes for distributed matrix computations where we can preserve the sparsity up to a level, in this section we incorporate $\beta$-level coding leveraging the idea of resolvable designs. \aditya{Discussion in preceding para moves to earlier section to motivation for the work}

We begin our discussion of $\beta$-level coding by considering the uncoded $\beta = 1$ case. In this scenario, the assignments are simply elements such as $\bfA_i^T \bfx$ (in the matrix-vector case) or elements such as $\bfA_i^T \bfB_j$ (in the matrix-matrix case). In the discussion below we refer the assignment of ``symbols'' to treat both cases together, where a symbol can either be of the form $\bfA_i^T \bfx$ or $\bfA_i^T \bfB_j$. Note that we can disregard the case when multiple copies of a symbol appear within the same worker node.
Consider a  $\langle n,\ell,\Delta, r \rangle$-uncoded system with $n$ workers each of which can process $\ell \geq 1$ symbols out of a total of $\Delta$ symbols. We assume that each symbol appears $r$ times across the different worker nodes, so $n \ell = \Delta r$. Now we show a lower bound on the value of $Q$ for such a system.

%and/or $\{\bfB_0, \bfB_1, \dots, \bfB_{\Delta_B-1}\}$.  Here we define an $\langle n,\ell,\Delta, r \rangle$-uncoded system which have $n$ workers, in total, each of which can process $\ell \geq 1$ unknowns out of $\Delta$. Thus the unknowns are assigned to be processed $r$ times accross all the workers and we can say $n \ell = \Delta r$. In this case we show a lower bound on the value of $Q$ that is met with equality by a simple cyclic assignment scheme.

 %In this case we show a lower bound on the value of $Q$ that is met with equality by a simple cyclic assignment scheme. 
 %\aditya{lower bound argument goes here. Does it also work for uncoded matrix-matrix?}

\begin{theorem}
\label{thm:Q_by_Delta_bound}
For a $\langle n,\ell,\Delta, r \rangle$-uncoded system we have $Q \geq \Delta r - \frac{r}{2}\, (\ell + 1) + 1$.
\end{theorem}
\begin{proof}%[\textbf{Proof of Theorem 2}]
For the system under consideration, let $Q_j$ represent the maximum number of symbols that are processed in the worst case without processing symbol $j$ (see Fig. \ref{uncoded_matvec_prop} for an example). It is evident in this case that $Q = \max_{j=0, \dots, \Delta-1} Q_j + 1$.

Our strategy is to calculate the average $\overline{Q} = \frac{1}{\Delta}\sum_{j=0}^{\Delta-1} Q_j$ and use the simple bound $Q \geq \overline{Q} + 1$. Toward this end, note that for any uncoded solution, we can calculate $\sum_{j=0}^{\Delta-1} Q_j$ in a different way. For any worker $i$, there are $\ell$ assigned block-columns and the other $\Delta - \ell$ do not appear in it. Thus, in the calculation of $\sum_{j=0}^{\Delta-1} Q_j$, worker node $i$ contributes
\begin{align*}
 (\Delta - \ell) \ell + \sum\limits_{k=1}^{\ell} (k-1) \; \; \textrm{symbols},
\end{align*}which is clearly independent of $i$. Therefore,
\begin{align*}
\overline{Q} = n \; \frac{\left[\sum\limits_{k=1}^{\ell} (k-1) + (\Delta - \ell) \ell \right]}{\Delta} =  n \ell - \frac{n \ell}{2 \Delta} (\ell+1).
\end{align*}Thus, we have the lower bound as
\begin{align}
\label{eq:Qlb}
Q \geq \Delta r - \frac{r}{2}\, (\ell + 1) + 1 
\end{align} since $n \ell = \Delta r$.
\end{proof}
\begin{remark}
\label{remark:Delta_value}
In general, we are given the number of workers $n$ and the storage fraction $\gamma$. The parameters $\Delta$ and $\ell$ can be treated as design parameters. In this setting, from \eqref{eq:Qlb}, we have
\begin{align}
    \frac{Q}{\Delta} \geq r - \frac{r}{2} \frac{\ell}{\Delta} - \frac{r}{2 \Delta} + \frac{1}{\Delta}; \;\; \; \textrm{but $r = n \gamma$, thus} \; \; \;   \frac{Q}{\Delta} \geq n\gamma \left(1 - \frac{\gamma}{2}\right) + \left(1 - \frac{n \gamma}{2}\right) \frac{1}{\Delta}. \label{eq:uncoded_Q_by_Delta_analysis}
\end{align}
If $r = n\gamma > 2$, then the second term in the RHS above is negative and has an inverse dependence on $\Delta$.
\end{remark}
The lower bound in \eqref{eq:Qlb} is met with equality when we consider the cyclic assignment scheme. For instance, Fig. \ref{uncoded_matvec_prop} shows an example where $\Delta = n = 5$, and it can be verified that $Q = 10$ and meets the lower bound in \eqref{eq:Qlb}. A similar result holds for the matrix-matrix case. These results are discussed in the relevant parts of the remainder of this section.

%On particular, we partion matrix $\bfA$ into $\Delta = n$ block-columns. In the uncoded scheme, the master node can recover any $j$-indexed unknown, $\bfA_j^T \bfx$ if any of the workers computes that. Thus according to Lemma \ref{lem:cyclicQ}, all worker nodes can process $\alpha_0$ symbols such that any specific $\bfA_j^T \bfx$ is not computed by any of those $n$ workers. Thus setting $c = 0$ in Lemma \ref{lem:cyclicQ}, we obtain $Q = \alpha_0 + 1 = \Delta \ell -  \frac{\ell(\ell+1)}{2} + 1$. Since we set $\Delta = n$, we can say that $\ell = r$ and the lower bound in \eqref{eq:Qlb} is met by the uncoded cyclic scheme. Fig. \ref{uncoded_matvec_prop} shows an example where $\Delta = n = 5$, and it can be verified that $Q = 10$ and meets the lower bound in \eqref{eq:Qlb}. 

%It should be noted that similar arguments are also applicable for the uncoded matrix-matrix multiplication case where the assigned block-columns are simply elements of $\{\bfA_0, \bfA_1 \dots, \bfA_{\Delta_A-1}\}$  and $\{\bfB_0, \bfB_1, \dots, \bfB_{\Delta_B-1}\}$. The cyclic uncoded scheme can meet the same lower bound for $Q$ to recover all $\Delta = \Delta_A \Delta_B$ unknowns.

%\anindya{need to discuss, should comment om matrix-matrix ? move fig. \ref{uncoded_matmat_prop} ?}

%The remainder of the discussion is applicable in the case when $\beta > 1$.
\subsection{Matrix-vector Multiplication}

We consider a $\beta$-level coding matrix-vector scenario where the storage fraction $\gamma = a_1/a_2$ for positive integers $a_1$ and $a_2$ with $a_1 \leq a_2$ such that $\gamma \leq \frac{1}{\beta}$. We assume that the number of worker nodes $n = c a_2$ where $c$ is a positive integer. 
%\aditya{putting in the constraint explicitly}.

We partition $\bfA$ into $\Delta$ block-columns where $\Delta$ is divisible by $\beta$. Next, we pick a resolvable design $(\calX, \calA)$ where $\calX = \{0, \dots, \Delta -1\}$. The size of the blocks in $\calA$ is $\beta$. Let $\calP_1, \calP_2, \dots$ denote distinct parallel classes of this design. We will refer to the blocks of the design as meta-symbols (to avoid potential confusion with the term block-columns which we also have used extensively). Thus, the elements of a parallel class are meta-symbols.

The overall idea is to partition the set of worker nodes into $c$ groups denoted $\calG_0, \dots, \calG_{c-1}$. For each group we pick a parallel class and place meta-symbols from the parallel class in a cyclic fashion. The parallel classes for the different groups can be the same as well. For each meta-symbol, we generate a coded block-column by choosing a random linear combination of the $\beta$ block-columns within it. In the discussion below we refer to the block-columns as ``unknowns'' as they need to be decoded by the master nodes. A precise description appears in Algorithm \ref{Alg:betalevelcoding}. We illustrate it by means of an example below.

\begin{algorithm}[h]
	\caption{$\beta$-level coding scheme for distributed matrix-vector multiplication}
	\label{Alg:betalevelcoding}
   \SetKwInOut{Input}{Input}
   \SetKwInOut{Output}{Output}
   \Input{Matrix $\bfA$ and vector $\bfx$. Storage fraction $\gamma = \frac{a_1}{a_2} \leq \frac{1}{\beta}$, $\beta$-allowed coding level, and number of workers $n = c a_2$ where $c$ is a positive integer.}
   Set $\Delta = \beta a_2$. Partition $\bfA$ into $\Delta$ block-columns\;
   Number of assigned blocks per worker, $\ell = \Delta \gamma$\;
   Assume $\calX = \{0, 1, 2, \dots, \Delta-1 \}$ and find $c$ parallel classes $\calP_i$ having a block size $\beta$, $i = 0, 1, \dots, c - 1$\;

   \For{$i\gets 0$ \KwTo $c - 1$}{
   Let the blocks of $\calP_i$ be denoted as $p_0, p_1, \dots, p_{\frac{\Delta}{\beta} - 1}$ \;
   \For{$j\gets 0$ \KwTo $\frac{\Delta}{\beta} - 1$}{
   Assign meta-symbols $p_j, p_{j+1}, \dots, p_{j + \ell - 1}$ from top to bottom (indices reduced modulo $a_2$) and vector $\bfx$ to worker $\frac{\Delta}{\beta} i + j$\;
   %Choose a random vector of length $\beta$ to obtain the random linear combination of the submatrices of the assigned elements of $\calP_i$.
   For each meta-symbol choose a random linear combination of length-$\beta$ of the constituent block-columns\;
   }
   }
   \Output{Distributed matrix-vector multiplication scheme having $\beta$-level coding.}
\end{algorithm}

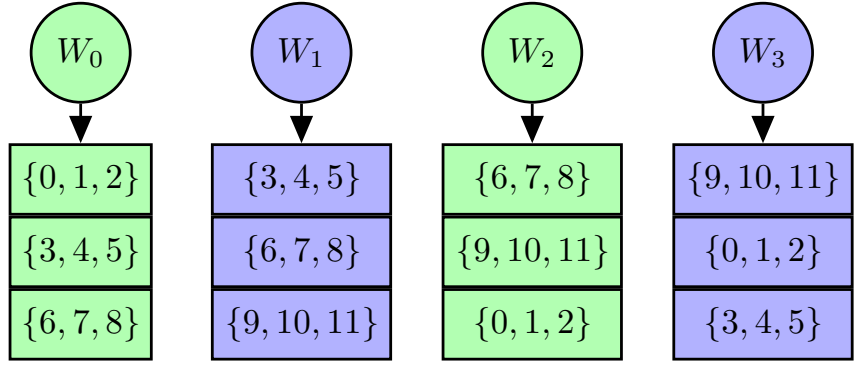
\begin{figure}[t]
\centering
\captionsetup{justification=centering}
\resizebox{0.7\linewidth}{!}{
\begin{tikzpicture}[auto, thick, node distance=2cm, >=triangle 45]
\draw
    node [sum1, minimum size = 1cm, fill=green!30] (blk1) {$W_0$}
    node [sum1, minimum size = 1cm, fill=blue!30,right = 1.2 cm of blk1] (blk2) {$W_1$}
    node [sum1, minimum size = 1cm, fill=green!30,right = 1.3 cm of blk2] (blk3) {$W_2$}
    node [sum1, minimum size = 1cm, fill=blue!30,right = 1.3 cm of blk3] (blk4) {$W_3$}

    node [block, fill=green!30,below = 0.4 cm of blk1] (blk011) {$\{ 0, 1, 2\}$}
    node [block, fill=green!30,below = 0.0005 cm of blk011] (blk012) {$\{3, 4, 5\}$}
    node [block, fill=green!30,below = 0.0005 cm of blk012] (blk013) {$\{6, 7, 8\}$}
    
    node [block, minimum width = 1.8 cm, fill=blue!30,below = 0.4 cm of blk2] (blk21) {$\{3, 4, 5\}$}
    node [block, minimum width = 1.8 cm, fill=blue!30,below = 0.0005 cm of blk21] (blk22) {$\{6, 7, 8\}$}
    node [block, minimum width = 1.8 cm,fill=blue!30,below = 0.0005 cm of blk22] (blk23) {$\{9,10,11\}$}
    
    node [block, minimum width = 1.8 cm,fill=green!30,below = 0.4 cm of blk3] (blk31) {$\{6, 7, 8\}$}
    node [block, minimum width = 1.8 cm,fill=green!30,below = 0.0005 cm of blk31] (blk32) {$\{9,10,11\}$}
    node [block, minimum width = 1.8 cm,fill=green!30,below = 0.0005 cm of blk32] (blk33) {$\{0, 1, 2\}$}
    
    node [block, minimum width = 1.8 cm,fill=blue!30,below = 0.4 cm of blk4] (blk41) {$\{9,10,11\}$}
    node [block, minimum width = 1.8 cm,fill=blue!30,below = 0.0005 cm of blk41] (blk42) {$\{0, 1, 2\}$}
    node [block, minimum width = 1.8 cm,fill=blue!30,below = 0.0005 cm of blk42] (blk43) {$\{3, 4, 5\}$}
    
    ;

\draw[->](blk1) -- node{} (blk011);
\draw[->](blk2) -- node{} (blk21);
\draw[->](blk3) -- node{} (blk31);
\draw[->](blk4) -- node{} (blk41);

\end{tikzpicture}
}
\centering
\caption{\small Job assignment for worker group $\calG_0$ for $\beta$-level matrix-vector multiplication scheme for $n = 12$ with $\gamma_{A} = \frac{1}{4}$ and $\Delta_A = 12$ using a single parallel class with $\beta = 3$. The indices $\{i,j,k\}$ indicates a random linear combination of the submatrices $\bfA_i, \bfA_j$ and $\bfA_k$. $\calG_1$ and $\calG_2$ are assigned the same symbols as workers $0-3$ but with different random coefficients.}

\label{beta_matvec}
\end{figure} 
\begin{example}
Consider a scenario with $n = 12$, $\gamma = 1/4$ and $\beta = 3$, and set $\Delta = 12$. We let $\calX = \{0, 1, \dots, 11\}$ and pick $\calP = \{\{0,1,2\}, \{3,4,5\}, \{6,7,8\}, \{9,10,11\}\}$. In this example, all three groups use the same parallel class $\calP$. As shown in Fig. \ref{beta_matvec}, in each group the meta-symbols are arranged in a cyclic fashion. For each meta-symbol a random linear combination is chosen, e.g. in worker $W_0$ the meta symbol $\{0,1,2\}$ will be replaced by $\tilde{\bfA}_0 = z_0 \bfA_0 + z_1 \bfA_1 + z_2 \bfA_2$ where the $z_i$'s are chosen at random. This implies that $W_0$ is responsible for computing $\tilde{\bfA}_0^T \bfx$, and the unknowns $\bfA_0^T \bfx, \bfA_1^T \bfx$ and $\bfA_2^T \bfx$ can be decoded if three copies of the meta-symbol $\{0,1,2\}$  are obtained from the workers as the corresponding equations are linearly independent with probability 1.
\end{example}

\begin{theorem}
\label{thm:beta_matvecstrQ}
Consider a distributed matrix-vector multiplication scheme for $n = c a_2$ workers where each worker can store $\gamma = \frac{a_1}{a_2}$ fraction of matrix $\bfA$. Suppose that $c \geq \beta$ and we use the same parallel class $\calP$ over all the worker groups. Then, the scheme described in Alg. \ref{Alg:betalevelcoding} will be resilient to $s = c \ell - \beta$ stragglers, and $Q = n \ell - \frac{c \ell (\ell + 1)}{2} + \ell (\beta - 1) + 1$.
\end{theorem}

\begin{proof}
Based on our construction we know that any meta-symbol $\in \calP$ will appear in $\ell$ distinct workers in each worker group consisting of $\Delta/\beta = a_2$ workers ({\it cf.} Lemma \ref{lem:cyclicQ}). Thus there are $\frac{n}{a_2} = c$ such worker groups and it follows that there are a total of $c \ell$ appearances of that meta-symbol across all the worker nodes. Furthermore, each meta-symbol corresponds to a random linear combination of the corresponding unknowns (block-columns). As the choice of these random coefficients is made from a continuous distribution, as long as {\it any} $\beta$ meta-symbols are processed across all the worker nodes, the constituent unknowns will be decodable with probability $1$. Thus, the scheme is resilient to the failure of any $c \ell - \beta$ stragglers.

%The corresponding coefficients of these $\ell \beta$ symbols are chosen from random distribution, and we know any $\beta \times \beta$ submatrix of a $\beta \times \ell \beta$ random matrix is full rank, which indicates that $\beta$ symbols related to that element are enough to recover the corresponding submatrix products. Thus we can say that $\beta \ell - \beta = \beta (\Delta \gamma - 1)$ workers are redundant. This is true for any arbitrary element of $\calP$ having cardinality $\beta$, so we can say that the scheme is resilient to $s = \beta (n \gamma - 1)$ stragglers.

For the second claim, suppose that there exists a meta-symbol $\star \in \calP$ that is processed at most $\beta-1$ times when $n \ell - \frac{c \ell (\ell + 1)}{2} + \ell (\beta - 1) + 1$ meta-symbols have been processed. For each worker group, the meta-symbol $\star$ appears in all the positions $0, \dots, \ell-1$. Suppose that $\star$ appears $i$ times in $\eta_i$ worker groups for $i= 1, \dots, y$. Thus, $\sum_{i=1}^y i \eta_i \leq \beta -1$ and the maximum number of meta-symbols that can be processed is
\begin{align*}
Q{'} =  \sum_{i=1}^y \eta_i \alpha_i + (c- \sum_{i=1}^y \eta_i) \alpha_0
\end{align*} where $\alpha_0 = \frac{\Delta}{\beta} \ell - \frac{\ell(\ell+1)}{2}$ and $\alpha_i = \alpha_0 + \sum_{j=0}^{i-1} (\ell - i) = \alpha_0 + i \ell - \frac{i(i-1)}{2}$ as specified in Lemma \ref{lem:cyclicQ} (by setting the number of symbols to $\Delta/\beta$). Thus,
\begin{align}
    \label{eq:Qbar}
    Q^{'} = c \alpha_0 + \ell \sum\limits_{i=1}^y i \eta_i  - \sum\limits_{i=1}^y \eta_i \frac{i(i-1)}{2} \leq c \alpha_0 + \ell (\beta - 1)
\end{align} since we have $\sum_{i=1}^y i \eta_i \leq \beta -1$. Equality holds in \eqref{eq:Qbar} if we have $y=1$ and $\eta_1 = \beta - 1$.

\begin{comment}
Thus, $\sum_{i=1}^y i \eta_i = \beta -1$ and the maximum number of meta-symbols that can be processed is
\begin{align*}
    \sum_{i=1}^y \eta_i \alpha_i + (c- \sum_{i=1}^y \eta_i) \alpha_0
\end{align*}
where the $\alpha_i$'s are specified in Lemma \ref{lem:cyclicQ} (with number of symbols $=\Delta/\beta$). It can be shown that the above expression is maximized when we set $\eta_1 = \beta-1$. To see this note that
\begin{align*}
    \sum_{i=1}^y \eta_i \alpha_i + (c- \sum_{i=1}^y \eta_i) \alpha_0 &< (\beta-1) \alpha_1 + (c-(\beta-1)) \alpha_0 \\
    \iff \sum_{i=1}^y \eta_i \alpha_i + \sum_{i=1}^y (i-1)\eta_i \alpha_0 &< (\beta-1) \alpha_1 \text{(using $\sum_{i=1}^y i \eta_i = \beta -1)$}\\
    \iff \sum_{i=1}^y \eta_i (\alpha_i-\alpha_0) + (\beta-1)\alpha_0 &< (\beta-1) \alpha_1\\
    \iff \sum_{i=1}^y \eta_i (\alpha_i-\alpha_0) &< \ell (\beta-1)\\
    \iff 0 &< \sum_{i=1}^y \eta_i (i \ell - \alpha_i + \alpha_0).
\end{align*}
It can be verified that the last inequality is true by substituting the values of the $\alpha_i$'s.
\end{comment}

In the worst case therefore, we can process $\alpha_1$ symbols from $\beta-1$ groups and $\alpha_0$ symbols from the remaining groups. This gives a total of
$$ (\beta-1)\alpha_1 + (c - \beta + 1)\alpha_0 = n \ell - \frac{c \ell (\ell + 1)}{2} + \ell (\beta - 1) $$
symbols, which is the same as the upper bound in \eqref{eq:Qbar}. Thus if $Q \geq n \ell - \frac{c \ell (\ell + 1)}{2} + \ell (\beta - 1) + 1$ then we are guaranteed that every meta-symbol is processed at least $\beta$ times. This concludes the proof. 
\end{proof}

It can be verified that the distributed matrix-vector multiplication scheme shown in Fig. \ref{beta_matvec} is resilient to $s = c \ell - \beta = 3\times 3 - 3 = 6$ stragglers and has  $Q = 25$. Theorem \ref{thm:beta_matvecstrQ} provides the value for $s$ and $Q$ for distributed matrix-vector multiplication when $\beta \leq c$. In Appendix \ref{app:diffc}, we show the calculation for $s$ and $Q$  for the case when $\beta > c$.

\begin{remark}
The proposed $\beta$-level coding scheme leads to an algorithm for uncoded matrix-vector multiplication when we set $\beta = 1$ (see Fig. \ref{uncoded_matvec_prop} for an example). The ratio $Q/\Delta$ for the construction in Alg. \ref{Alg:betalevelcoding} is lower in general as compared to the scheme in \cite{c3les}. For instance,  with $n = 10$ and $\gamma = 2/5$, Alg. \ref{Alg:betalevelcoding} results in a scheme with $Q/\Delta = 3.0$, whereas the \cite{c3les} scheme has $Q/\Delta=3.1$. The reduction is due to the lower value of $\Delta$ ({\it cf.} Remark \ref{remark:Delta_value}). 

\end{remark}

\begin{remark}
For $\beta > 1$ the $Q/\Delta$ ratio can be reduced significantly as compared to the uncoded ($\beta=1$) case. To see this consider 
%Now to show the difference between $Q/\Delta$ values of the uncoded case and $\beta \geq 2$ case, we consider an example with 
$n = c a_2$ and $\gamma = \frac{a_1}{a_2}$, where $c \geq \beta$. For the uncoded case, we set $\Delta_{unc} = a_2$, and we have $Q_{unc} = n \ell - \frac{c \ell (\ell + 1)}{2} + 1$ where $\ell = a_1$. On the other hand for $\beta$-level coding, we set $\Delta_{\beta} = \beta a_2$, and we have $Q_{\beta} = n \ell - \frac{c \ell (\ell + 1)}{2} + \ell (\beta - 1) + 1$ where $\ell = \beta a_1$. This implies that 
\begin{align*}
    \frac{Q_{unc}}{\Delta_{unc}} - \frac{Q_{\beta}}{\Delta_{\beta}} 
    & = (\beta - 1) \left(\gamma \left( \frac{c a_1}{2} - 1\right) + \frac{1}{\beta a_2}\right) > 0.
\end{align*}

\end{remark}

It turns out that the recovery threshold can be further reduced if we judiciously choose different parallel classes for the different worker groups in Alg. \ref{Alg:betalevelcoding}. Utilizing these parallel classes, we present a method that improves on Theorem \ref{thm:beta_matvecstrQ} if we assume the property that the blocks among different parallel classes have intersection size to be at most one. Before stating the theorem, we discuss the decodabilty of the approach since this is not as straightforward as the single parallel class $\beta$-level coding.

To understand the decoding in this setting we consider a bipartite graph $\bfG_{dec} = \calU \cup \calV$ whose vertex set consists of the unknowns ($\calU$) on the left and the processed meta-symbols ($\calV$) on the right; an example is shown in Fig. \ref{hall_bipartite}. A meta-symbol is connected to its constituent unknowns. Note that $\bfG_{dec}$ specifies a system of equations in $\Delta$ unknowns and we need to argue that this system is invertible. In the argument below, suppose that the random linear coefficients of each meta-symbol are indeterminates and we argue that there exists a matching in $\bfG_{dec}$ where all the unknowns in $\calU$ are matched.

\begin{figure}[t]
\definecolor{myblue}{RGB}{80,80,160}
\definecolor{mygreen}{RGB}{80,160,80}
\centering
\captionsetup{justification=centering}
\resizebox{0.4\linewidth}{!}{

\begin{tikzpicture}[thick,
  every node/.style={draw,circle},
  fsnode/.style={fill=myblue},
  ssnode/.style={fill=mygreen},
  every fit/.style={ellipse,draw,inner sep=-2pt,text width=2cm},
  ->,shorten >= 3pt,shorten <= 3pt
]
% the vertices of U
\begin{scope}[start chain=going below,node distance=7mm]
\foreach \i in {1,2,...,4}
{
  \pgfmathtruncatemacro{\j}{\i - 1}
  \node[fsnode,on chain] (f\i) [label=left: $u_{\j}$] {};
}
\end{scope}

% the vertices of V
\begin{scope}[xshift=4cm,yshift=0.5cm,start chain=going below,node distance=7mm]
\foreach \i in {5,6,...,9}
{ 
  \pgfmathtruncatemacro{\j}{\i - 5}
  \node[ssnode,on chain] (s\i) [label=right: $v_{\j}$] {};
}
\end{scope}

% the set U
\node [myblue,fit=(f1) (f4),label=above:$\mathcal{U}$] {};
\node [myblue,fit=(f1) (f4),label=below:$\textrm{Unknowns}$] {};
% the set V
\node [mygreen,fit=(s5) (s9),label=above:$\mathcal{V}$] {};
\node [mygreen,fit=(s5) (s9),label=below:$\textrm{Symbols}$] {};

% the edges
\draw (f1) -- (s5);
\draw (f1) -- (s6);
\draw (f2) -- (s5);
\draw (f2) -- (s6);
\draw (f2) -- (s7);
\draw (f3) -- (s7);
\draw (f3) -- (s8);
\draw (f3) -- (s9);
\draw (f4) -- (s8);
\draw (f4) -- (s9);

\end{tikzpicture}
}
\caption{\small For the case $\beta = 2$, every symbol is a random linear combination of two unknowns, thus there is a bipartite graph between the unknowns and symbols.}
\label{hall_bipartite}
\end{figure}
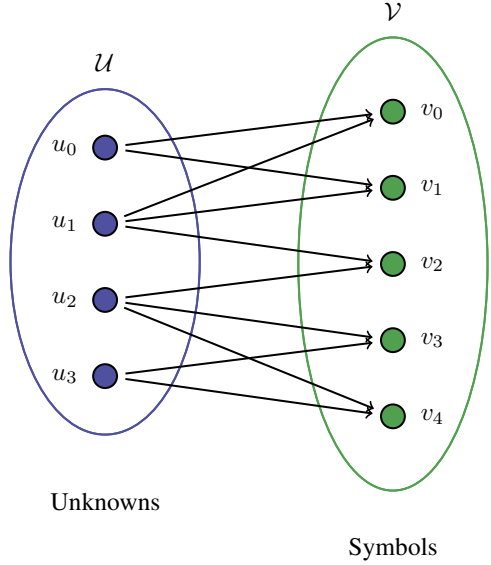 

Consider a set $\tilde{\calU}$ of $d$ unknowns from $\bfG_{dec}$ and the corresponding neighborhood $\tilde{\calV} = \calN(\tilde{\calU})$. Suppose we have a set of equations where these $d$ unknowns, namely $u_0, u_1, \dots, u_{d-1}$, participate in $\ell_0, \ell_1, \dots, \ell_{d-1}$ equations. Thus the number of outgoing edges from $\tilde{\calU}$ is $\sum_{i=0}^{d-1} \ell_i$. On the other hand, because of the structure of $\beta$-level coding approach, any symbol in $\tilde{\calV}$ has a degree $\beta$, thus the number of incoming edges in $\tilde{\calV}$ is $\beta |\tilde{\calV}|$. 

Now, suppose that a matching where all the elements of $\calU$ are matched does not exist. Hall's marriage theorem \cite{marshall1986combinatorial} gives a necessary and sufficient condition for the existence of the matching. Suppose that Hall's condition is violated for the set $\tilde{\calU}$, i.e., $|\tilde{\calV}| \leq d-1$. This means that
%Now, we assume that we cannot have the matching between $\tilde{\calU}$ and $\tilde{\calV}$. 
%We know that $|\tilde{\calU}| < |\tilde{\calV}|$, thus $|\tilde{\calV}|$ can be at most $d-1$. Furthermore, this means that
%Thus, if there is no matching between any set of $d$ unknowns and its neighborhood, we will have
\begin{align}
\label{eq:hallstheorem}
    \sum\limits_{i=0}^{d-1} \ell_i \; \leq \; \beta (d - 1) .
\end{align} %In other words, if \eqref{eq:hallstheorem} cannot be satisfied for any set of $d$ unknowns out of $\Delta$ unknowns, then there exists a matching in $\bfG_{dec}$ where all the unknowns are matched.

\begin{lemma}
\label{lem:allbeta}
Suppose that $\bfG_{dec}$ is such that each of $\Delta$ unknowns has at least degree $1$ and at least $\Delta - 1$ unknowns have degree at least $\beta$ each. Then, the master node can decode all the unknowns.
\end{lemma}
\begin{proof}
First, consider $d = 1$, so $|\tilde{\calU}| = 1$. Since each unknown has at least degree $1$, thus $\sum_{i=0}^{d-1} \ell_i \geq 1 > \beta (d - 1)$. Next we consider any set of $d \geq 2$ unknowns, where we know that at least $(d-1)$ unknowns have degree at least $\beta$. In that case, $\sum_{i=0}^{d-1} \ell_i \geq 1 + \beta(d-1) > \beta (d-1)$. Thus \eqref{eq:hallstheorem} cannot be satisfied for any $d \geq 2$. So, there exists a matching in $\bfG_{dec}$ where all the unknowns are matched, hence the master node can decode all the unknowns.
\end{proof}

Now we state the result when different parallel classes are utilized. %theorem which generalizes the scheme for resilience to $c \ell - \beta + \lambda$ stragglers, $\lambda < \beta$. It should be noted that every unknown appears exactly $c \ell$ times over all the workers, thus the scheme cannot be resilient to more than $c \ell - 1$ stragglers, hence $\lambda$ has to be smaller than $\beta$.

\begin{theorem}
\label{thm:diff_par}
Consider $c$ distinct parallel classes (with block size $\beta$) such that the size of the intersection between any two blocks from different parallel classes is at most $1$. Using Alg. \ref{Alg:betalevelcoding}, the distributed matrix-vector multiplication scheme can be resilient to at least $c \ell - \beta + \lambda$ stragglers, when $c \geq \beta + \lambda$ and $\lambda < \beta$.
\end{theorem}
\begin{proof}
The main idea is to find the scenario where Lemma \ref{lem:allbeta} can be directly applicable. To establish that we consider two unknowns $u_0$ and $u_1$. The event that a pair of unknowns belong to the same meta-symbol can happen within at most one parallel class (in other words, within only one worker group) according to our choices of parallel classes. Thus in the remaining $(c - 1)$ worker groups, those two unknowns exist in different meta-symbols. If they appear in the same meta-symbol, then there are $\ell$ workers within the worker group where they appear. On the other hand, if they appear in different meta-symbols, then there are at least $\ell + 1$ workers within the worker group where either $u_0$ or $u_1$ or both appear as part of a meta-symbol.

So, the unknowns $u_0$ or $u_1$ or both participate in different meta-symbols in at least in $\ell + (\ell + 1) (c - 1)$ workers. Now since we have $c \ell - \beta + \lambda$ stragglers, using $c \geq \beta + \lambda$, we have still
\begin{align*}
    \ell + (\ell + 1) (c - 1) - \left( c \ell - \beta + \lambda \right) \geq 2 \beta -1 \; 
\end{align*} workers left. This means that either $u_0$ or $u_1$ or both exist in at least $2\beta-1$ workers after the stragglers are removed. This in turn implies that the corresponding $\bfG_{dec}$ has at least $2\beta-1$ edges emanating from the pair of unknowns $u_0$ and $u_1$, so that at least one of them has degree $\geq \beta$. Thus Lemma \ref{lem:allbeta} is satisfied and we can decode all the unknowns.
\end{proof}

\begin{example}
Consider a scenario with $n = 20$ workers with $\gamma = \frac{1}{5}$, thus $c = 4$, and we apply $\beta$-level coding approach with $\beta = 3$. In this regard, we incorporate $four$ different parallel classes of block size $\beta = 3$ obtained from the solution of the famous Kirkman's Schoolgirl problem \cite{kirkman1850lady}, where any two blocks from any two different parallel classes have an intersection size at most $one$. It can be verified that the distributed matrix-vector multiplication scheme will be resilient to at least $s = c \ell - \beta  + 1 = 10$ stragglers whereas the number of stragglers if we used the single parallel class would have been $9$.
\end{example}

%Note that \eqref{eq:hallstheorem} leads to a sufficiency condition for decoding Moreover, in the proof of Lemma \ref{lem:allbeta}, we have ignored the cases where both of the unknowns participate in coded symbols within any particular worker. Thus for any $c \geq \beta + \lambda$ groups of workers, $c \ell - \beta + \lambda$ is the lower bound of the number of stragglers; in practical, the number of stragglers that the scheme is resilient to can be more than that. 
Table \ref{compbeta} compares experimental results for different matrix-vector multiplication approaches in terms of number of stragglers and $Q/\Delta$ values. For every case, we observe a significant improvement of the metrics if we incorporate multiple parallel classes instead of a single parallel class. We note here that the $Q/\Delta$ was computed via computer experiments. 

\begin{table}[t]
\caption{{\small Comparison of different metrics for different approaches.}}
\label{compbeta}
\begin{center}
\begin{small}
\begin{sc}
\begin{tabular}{c c c c c}
\hline
\toprule
\multirow{2}{*}{System} & \multirow{2}{*}{Metrics} & Dense Codes & $\beta$-level Coding & $\beta$-level Coding \\ 
 & & \cite{yu2017polynomial}, \cite{8919859}, \cite{8849468} & (Single parallel class) & (Multiple parallel classes) \\ 
 \midrule
$n = 8$, $\gamma_A = \frac{1}{4}$ & $s$ & $4$ & $2$ & $3$ \\
\cline{2-5}
and $\; \beta = 2$ & $Q/\Delta$ & $-$ & $13/8$ & $11/8$ \\
 \midrule
$n = 8$, $\gamma_A = \frac{1}{4}$  & $s$ & $4$ & $3$ & $4$ \\
\cline{2-5}
and $\; \beta = 3$ & $Q/\Delta$ & $-$ & $19/12$ & $14/12$ \\
 \midrule
$n = 10$, $\gamma_A = \frac{1}{5}$  & $s$ & $5$ & $3$ & $4$ \\
\cline{2-5}
and $\; \beta = 3$ & $Q/\Delta$ & $-$ & $25/15$ & $20/15$ \\
\bottomrule
\end{tabular}
\end{sc}
\end{small}
\end{center}
\end{table}%
The analysis in Theorem \ref{thm:diff_par} above is somewhat loose as we only assume that the intersection sizes between blocks from different parallel classes is at most 1. Indeed, exploiting more structure in the choice of the parallel classes can yield better results, though the analysis becomes significantly harder.
%We can further improve the metrices for different specific judicious choices for the parallel classes. 
Here we present a method that improves on Theorem \ref{thm:beta_matvecstrQ} when $c = \beta = 2, \ell \leq \Delta/2-2$ and $\Delta \geq 8$.  Let $\calX = \{0, 1, \dots, \Delta-1\}$ where $\Delta = n = 2a_2$. The block size of the design is two and the parallel classes are given as follows.
\begin{align}
    \label{paracls}
    \calP_0 & = \left\lbrace \{0, 1 \},  \{2, 3 \}, \dots, \{\Delta-2, \Delta-1 \} \right\rbrace \nonumber \\
\textrm{and}\; \;\calP_1 &= \left\lbrace \{0, 5 \},  \{2, 7 \}, \dots, \{\Delta-2, 3\} \right\rbrace .
\end{align}
Thus, the $i$-th block in $\calP_0$ and $\calP_1$, for $0 \leq i \leq \Delta/2 -1$ is given by $\{2i, 2i+1\}$ and $\{2i, 2i+5\} \pmod \Delta$, respectively. We follow the Alg. \ref{Alg:betalevelcoding} for the specification of the coding scheme.

\begin{theorem}
\label{thm:beta2_matvecstrQ}
Let $c=\beta =2$, $\ell \leq \Delta/2-2$ and $\Delta \geq 8$. If we use the parallel classes in \eqref{paracls}, then the matrix-vector scheme described in Alg. \ref{Alg:betalevelcoding} will be resilient to $s = 2\ell  - 1$ stragglers, and $Q = n \ell - \ell (\ell + 1) + 1$.
\end{theorem}

\begin{proof}
The detailed proof is discussed in Appendix \ref{App:diff_beta_2}
\end{proof}

\subsection{Matrix-matrix Multiplication}
Now we consider the case of matrix-matrix multiplication, where we assume that each of the $n$ worker nodes can store $\gamma_A = \frac{a_1}{a_2}$ and $\gamma_B = \frac{b_1}{b_2}$ fractions of matrices $\bfA$ and $\bfB$. In this case, we consider $\beta_A$ and $\beta_B$-level coding for $\bfA$ and $\bfB$, respectively so that $\gamma_A \leq \frac{1}{\beta_A}$ and $\gamma_B \leq \frac{1}{\beta_B}$. We partition matrices $\bfA$ and $\bfB$ into $\Delta_A$ and $\Delta_B$ block-columns, respectively, and so, we have, in total, $\Delta = \Delta_A \Delta_B$ unknowns. Next we assign $\ell_A = \Delta_A \gamma_A$ block-columns from $\bfA$ and $\ell_B = \Delta_B \gamma_B$ block-columns of $\bfB$ to each of the workers. Thus each worker computes $\ell = \ell_A \ell_B$ submatrix products according to the natural order discussed in Section \ref{sec:prel}.% which is, in fact, $\gamma = \frac{\ell}{\Delta} = \gamma_A \gamma_B$ fraction of the whole result, $\bfA^T \bfB$.

Once the matrices are decomposed into block-columns, we allow $\beta_A$-level and $\beta_B$-level coding for matrices $\bfA$ and $\bfB$, respectively.  In this case we choose two separate resolvable designs with block sizes $\beta_A$ and $\beta_B$ supported on point sets $\{0, 1, \dots, \Delta_A -1\}$ and $\{0, 1, \dots, \Delta_B -1\}$ respectively. Furthermore, we assume that  the number of worker nodes $n = c \times a_2 b_2$ where $c$ is a positive integer.

\begin{algorithm}[h]
	\caption{$\beta$-level coding scheme for matrix-matrix multiplication}
	\label{Alg:betalevelcoding_matmat}
   \SetKwInOut{Input}{Input}
   \SetKwInOut{Output}{Output}
   \Input{Matrices $\bfA$ and $\bfB$, storage fractions of the workers $\gamma_A = \frac{a_1}{a_2}$ $\gamma_B = \frac{b_1}{b_2}$, $\beta_A, \beta_B$-coding level for $\bfA$ and $\bfB$, respectively, and number of worker nodes, $n = c \times a_2 b_2$, where $c$ is a positive integer.}
   Partition $\bfA$ into $\Delta_A = \beta_A a_2$ block-columns and partition $\bfB$ into $\Delta_B = \beta_B b_2$ block-columns\; 
   
  $\Delta = \Delta_A \Delta_B$, $\ell_A = \Delta_A \gamma_A$, $\ell_B = \Delta_B \gamma_B$, $\beta = \beta_A \beta_B$\;
  Assume $\calX_A = \{0, 1, 2, \dots, \Delta_A-1 \}$ and find parallel classes $\calP^A_i$ having block size $\beta_A$, $i = 0, 1, \dots, c - 1$\; 
  Assume $\calX_B = \{0, 1, 2, \dots, \Delta_B-1 \}$ and find parallel classes $\calP^B_i$ having block size $\beta_B$, $i = 0, 1, \dots, c - 1$\; 
  
  \For{$i\gets 0$ \KwTo $c - 1$}{
   Let the blocks of $\calP^A_i$ be denoted as $p_{A_0}, p_{A_1}, \dots, p_{A_{\alpha_A - 1}}$ \;
   Let the blocks of $\calP^B_i$ be denoted as $p_{B_0}, p_{B_1},\dots, p_{B_{\alpha_B - 1}}$ \;
   
   \For{$j\gets 0$ \KwTo $\frac{\Delta}{\beta} - 1$}{
     Assign sets $p_{A_j}, p_{A_{j+1}}, \dots, p_{A_{j + \ell_A - 1}}$ from top to bottom (indices reduced modulo $a_2$) to worker $\frac{\Delta}{\beta} i + j$\;
     
     $k \gets \floor{\frac{j}{a_2}}$, and assign sets $p_{B_k}, p_{B_{k+1}}$, $\dots$, $p_{B_{k + \ell_B - 1}}$ from top to bottom (indices reduced modulo $b_2$) to worker $\frac{\Delta}{\beta} i + j$\;
     
     Choose random linear combinations of the constituent block-columns of the meta-symbols of $\calP^A_i$ and $\calP^B_i$ of length $\beta_A$ and $\beta_B$ respectively\;
   }
   }
   \Output{Distributed matrix-matrix multiplication scheme having $\beta$-level coding.}
\end{algorithm}

Let $\calP^A$ and $\calP^B$ denote parallel classes for the matrices $\bfA$ and $\bfB$ respectively. As in the matrix-vector scheme, the coding scheme is specified by the meta-symbols (blocks) of $\calP^A$ and $\calP^B$. 
Let $\calN_A$ and $\calN_B$ denote the corresponding incidence matrices of these parallel classes. Recall that each meta-symbol is in one-to-one correspondence with the columns of the incidence matrices. 
Consider the matrix $\calN_{AB}$ formed by considering pair-wise Kronecker products of columns from $\calN_A$ and $\calN_B$. Then the rows of $\calN_{AB}$ correspond to unknowns of the form $\bfA_i^T \bfB_j$ and the columns correspond to the support of the random linear equations that are formed by considering the pairwise products. We will refer to the meta-symbols of $\calN_{AB}$ as product meta-symbols and denote it by $ \calP^{AB}$. For example, suppose that $\beta_A = \beta_B = 2$ and consider two meta-symbols $\{0,1\} \in \calP^{A} $ and $\{0,1\} \in \calP^B$. If these symbols are placed in a worker, the corresponding product would be $(x_0 \bfA_0^T + x_1 \bfA_1^T)(y_0 \bfB_0 + y_1 \bfB_1) = x_0 y_0 \bfA_0^T \bfB_0 + x_0 y_1 \bfA_0^T \bfB_1 + x_1 y_0 \bfA_1^T \bfB_0 + x_1 y_1 \bfA_1^T \bfB_1$ where $x_0, x_1, y_0, y_1$ are chosen i.i.d. at random from a continuous distribution. Thus, the coefficients of the corresponding equation can be expressed as 
\begin{align}
[x_0 ~x_1] \otimes [y_0 ~y_1] \label{eq:kron_prod_matmat}    
\end{align}
 where $\otimes$ denotes the Kronecker product.
\begin{claim}
If $\calN_A$ (of size $\Delta_A \times \Delta_A/\beta_A$) and $\calN_B$ (of size $\Delta_B \times \Delta_B/\beta_B$) correspond to incidence matrices of parallel classes, then $\calN_{AB}$ also forms a parallel class of size $\Delta_A\Delta_B \times \Delta_A \Delta_B/\beta_A \beta_B$.
\end{claim}
\begin{proof}
Let $\bfu_i \otimes \bfv_i$ for $i = 0, 1$ denote two distinct columns of $\calN_{AB}$ such that $\bfu_i$ and $\bfv_i$ are columns in $\calN_A$ and $\calN_B$ respectively. Then, 
\begin{align*}
(\bfu_0 \otimes \bfv_0)^T (\bfu_1 \otimes \bfv_1) &= \bfu_0^T \bfu_1 \times \bfv_0^T \bfv_1\\
&=0
\end{align*}
since either $\bfu_0 \neq \bfu_1$ or $\bfv_0 \neq \bfv_1$. Moreover, there are $\frac{\Delta_A\Delta_B}{\beta_A\beta_B}$ distinct columns in $\calN_{AB}$ each with a support of size $\beta_A \beta_B$. This implies that together all the product meta-symbols in $\calN_{AB}$ cover all the $\Delta_A\Delta_B$ points.
\end{proof}

As in the matrix-vector case, the scheme operates by placing cyclically shifted meta-symbols from $\calP^A$ with $\ell_A$ meta-symbols in each worker for the first $\Delta_A/\beta_A$ workers. For these workers, the assignment of meta-symbols from $\calP^B$ is the same. For the next set of $\Delta_A/\beta_A$ workers the assignment of meta-symbols from $\calP^A$ repeats; however, we now employ a cyclic shift for the assignment of meta-symbols from $\calP^B$. The complete algorithm is specified in Alg. \ref{Alg:betalevelcoding_matmat} and an example is depicted in Fig. \ref{beta_matmat}. As before, a group in this setting contains $\Delta/\beta$ workers and there a total of $\frac{n}{\Delta/\beta} = c$ groups denoted $\calG_i, i = 0, 1, \dots, c-1$. Let $\calX_{AB} = \{\bfA_0^T\bfB_0, \bfA_0^T\bfB_1, \bfA_0^T\bfB_2, \dots, \bfA_{\Delta_A - 1}^T \bfB_{\Delta_B - 1} \}$ denote the set of unknowns. The product of two assigned coded block-columns consists of a random linear combination of $\beta = \beta_A \beta_B$ unknowns from $\calX_{AB}$.

\begin{example}
We consider an example with $n = 36$ workers in Fig. \ref{beta_matmat}, each of which can store $\gamma_A = \gamma_B = \frac{1}{3}$ of each of matrices $\bfA$ and $\bfB$, and $\beta_A = \beta_B = 2$. We set $\Delta_A = \Delta_B = 6$, thus the cardinality of $\calX_{AB}$ is $36$. In terms of indices, we use the same parallel class, $\{\{0,1\},\{2,3\},\{4,5\}\}$ for both $\bfA$ and $\bfB$. Finally we use random vectors of length $\beta_A = \beta_B = 2$ to obtain the symbols from the submatrices of the elements of the parallel classes, $\calP^A_i$ and $\calP^B_i$ in any worker group $\calG_i$, for $i = 0, 1, 2, 3$, as $c = 36/9 = 4$.
\begin{figure*}[t]
\centering
\captionsetup{justification=centering}
\resizebox{0.95\linewidth}{!}{
\begin{tikzpicture}[auto, thick, node distance=2cm, >=triangle 45]
\draw
    node [sum1, minimum size = 1 cm, fill=green!30] (blk1) {$W_0$}
    node [sum1, minimum size = 1 cm, fill=blue!30,right = 0.4 cm of blk1] (blk2) {$W_1$}
    node [sum1, minimum size = 1 cm, fill=green!30,right = 0.4 cm of blk2] (blk3) {$W_2$}
    node [sum1, minimum size = 1 cm, fill=blue!30,right = 0.4 cm of blk3] (blk4) {$W_3$}
    node [sum1, minimum size = 1 cm, fill=green!30,right = 0.4 cm of blk4] (blk5) {$W_4$}
	node [sum1, minimum size = 1 cm, fill=blue!30,right = 0.4 cm of blk5] (blk6) {$W_5$}
	node [sum1, minimum size = 1 cm, fill=green!30,right = 0.4 cm of blk6] (blk7) {$W_6$}
	node [sum1, minimum size = 1 cm, fill=blue!30,right = 0.4 cm of blk7] (blk8) {$W_7$}
	node [sum1, minimum size = 1 cm, fill=green!30,right = 0.4 cm of blk8] (blk9) {$W_8$}

    node [block, fill=green!30,below = 0.4 cm of blk1] (blk011) {$\{0, 1\}$}
    node [block, fill=green!30,below = 0.0005 cm of blk011] (blk012) {$\{2, 3\}$}
%    node [block, fill=green!30,below = 0.0005 cm of blk012] (blk013) {$6, 8$}
    node [block, fill=green!30,below = 0.4 cm of blk012] (blk013) {$\{0, 1\}$}
    node [block, fill=green!30,below = 0.0005 cm of blk013] (blk014) {$\{2, 3\}$}
    
    node [block, fill=blue!30,below = 0.4 cm of blk2] (blk21) {$\{2, 3\}$}
    node [block, fill=blue!30,below = 0.0005 cm of blk21] (blk22) {$\{4, 5\}$}
 %   node [block, fill=green!30,below = 0.0005 cm of blk22] (blk23) {$9,0$}
     node [block, fill=blue!30,below = 0.4 cm of blk22] (blk23) {$\{0, 1\}$}
    node [block, fill=blue!30,below = 0.0005 cm of blk23] (blk24) {$\{2, 3\}$}

    node [block, fill=green!30,below = 0.4 cm of blk3] (blk31) {$\{4,5\}$}
    node [block, fill=green!30,below = 0.0005 cm of blk31] (blk32) {$\{0, 1\}$}
%    node [block, fill=green!30,below = 0.0005 cm of blk32] (blk33) {$0, 2$}
    node [block, fill=green!30,below = 0.4 cm of blk32] (blk33) {$\{0, 1\}$}
    node [block, fill=green!30,below = 0.0005 cm of blk33] (blk34) {$\{2, 3\}$}
    
    node [block, fill=blue!30,below = 0.4 cm of blk4] (blk41) {$\{0,1\}$}
    node [block, fill=blue!30,below = 0.0005 cm of blk41] (blk42) {$\{2,3\}$}
%    node [block, fill=green!30,below = 0.0005 cm of blk42] (blk43) {$3, 4$}
    node [block, fill=blue!30,below = 0.4 cm of blk42] (blk43) {$\{2, 3\}$}
    node [block, fill=blue!30,below = 0.0005 cm of blk43] (blk44) {$\{4, 5\}$}
    
    node [block, fill=green!30,below = 0.4 cm of blk5] (blk51) {$\{2,3\}$}
    node [block, fill=green!30,below = 0.0005 cm of blk51] (blk52) {$\{4,5\}$}
%    node [block, fill=blue!30,below = 0.0005 cm of blk52] (blk53) {$6, 7$}
    node [block, fill=green!30,below = 0.4 cm of blk52] (blk53) {$\{2, 3\}$}
    node [block, fill=green!30,below = 0.0005 cm of blk53] (blk54) {$\{4, 5\}$}
    
    node [block, fill=blue!30,below = 0.4 cm of blk6] (blk61) {$\{4,5\}$}
    node [block, fill=blue!30,below = 0.0005 cm of blk61] (blk62) {$\{0,1\}$}
%    node [block, fill=blue!30,below = 0.0005 cm of blk62] (blk63) {$9,1$}
    node [block, fill=blue!30,below = 0.4 cm of blk62] (blk63) {$\{2, 3\}$}
    node [block, fill=blue!30,below = 0.0005 cm of blk63] (blk64) {$\{4, 5\}$}
    
    node [block, fill=green!30,below = 0.4 cm of blk7] (blk71) {$\{0,1\}$}
    node [block, fill=green!30,below = 0.0005 cm of blk71] (blk72) {$\{2,3\}$}
%    node [block, fill=blue!30,below = 0.0005 cm of blk72] (blk73) {$0, 1$}
    node [block, fill=green!30,below = 0.4 cm of blk72] (blk73) {$\{4, 5\}$}
    node [block, fill=green!30,below = 0.0005 cm of blk73] (blk74) {$\{0, 1\}$}
    
    node [block, fill=blue!30,below = 0.4 cm of blk8] (blk81) {$\{2,3\}$}
    node [block, fill=blue!30,below = 0.0005 cm of blk81] (blk82) {$\{4,5\}$}
%    node [block, fill=blue!30,below = 0.0005 cm of blk82] (blk83) {$3, 5$}
      node [block, fill=blue!30,below = 0.4 cm of blk82] (blk83) {$\{4, 5\}$}
    node [block, fill=blue!30,below = 0.0005 cm of blk83] (blk84) {$\{0, 1\}$}  
    
    node [block, fill=green!30,below = 0.4 cm of blk9] (blk91) {$\{4,5\}$}
    node [block, fill=green!30,below = 0.0005 cm of blk91] (blk92) {$\{0,1\}$}
%    node [block, fill=orange!30,below = 0.0005 cm of blk92] (blk93) {$7, 8$}
    node [block, fill=green!30,below = 0.4 cm of blk92] (blk93) {$\{4, 5\}$}
    node [block, fill=green!30,below = 0.0005 cm of blk93] (blk94) {$\{0, 1\}$}

    ;

\draw[->](blk1) -- node{} (blk011);
\draw[->](blk2) -- node{} (blk21);
\draw[->](blk3) -- node{} (blk31);
\draw[->](blk4) -- node{} (blk41);
\draw[->](blk5) -- node{} (blk51);
\draw[->](blk6) -- node{} (blk61);
\draw[->](blk7) -- node{} (blk71);
\draw[->](blk8) -- node{} (blk81);
\draw[->](blk9) -- node{} (blk91);

\end{tikzpicture}
}
\centering
\caption{\small Job assignment for worker group $\calG_0$ for $\beta$-level matrix-matrix multiplication scheme with $n = 36$ with $\gamma_{A} = \gamma_B = \frac{1}{3}$ and $\Delta_A = \Delta_B = 6$ using a single parallel class with $\beta_A = \beta_B = 2$. The indices $\{i,j\}$ on top and bottom parts indicate random linear combinations of the submatrices of $\bfA$ and $\bfB$, respectively. $\calG_1$, $\calG_2$ and $\calG_3$ are assigned the same symbols as workers $W_0 - W_8$, but with different random coefficients.}

\label{beta_matmat}
\end{figure*}
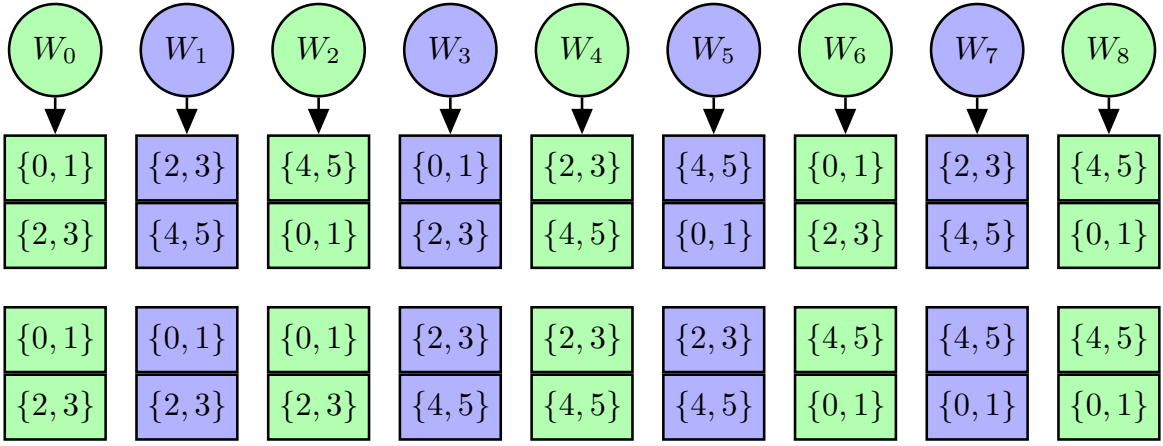
\end{example}

\begin{lemma}
\label{lem:beta_index}
The matrix-matrix multiplication scheme in Alg. \ref{Alg:betalevelcoding_matmat} is such that there are $\ell = \ell_A \ell_B$ symbols corresponding to any product meta-symbol $\in \calP_i^{AB}$ in a group $\calG_i$. Furthermore, this product meta-symbol appears in all locations $0, 1, 2, \dots, \ell - 1$ within $\calG_i$.
\end{lemma}

\begin{proof}
Any group $\calG_i$ can be partitioned into $\alpha_B = \frac{\Delta_B}{\beta_B}$ disjoint subgroups each of which consists of $\alpha_A = \frac{\Delta_A}{\beta_A}$ workers. These subgroups are denoted as $\calH_j$ where in terms of group worker indices, $\calH_j = \{j \alpha_A , j \alpha_A + 1, \dots, (j+1)\alpha_A -1\}$, for $j = 0, 1, \dots, \alpha_B - 1$. %Within each subgroup, the assignment of the meta-symbols of $\bfA$ follows a cyclic-shift scheme; this assignment repeats in the same manner for different $\calH_j$. On the other hand, the assignment of the meta-symbols of $\bfB$ is such that each worker within $\calH_j$ gets the same assignment and the cyclic shift occurs across the subgroups.

If meta-symbols $x \in \calP_i^A$ and $y \in \calP_i^B$ appear at locations $i_1$ and $j_1$, respectively, $0 \leq i_1 \leq \ell_A - 1$ and $0 \leq j_1 \leq \ell_B - 1$, then the product meta-symbol $x \otimes y$ appears at location $i_1 \ell_B + j_1$ in the ordering. In our case, meta-symbol $x$ appears $\ell_A$ times within subgroup $\calH_j$ at distinct locations $0, \dots, \ell_A -1$. Thus, if meta-symbol $y \in \calP_i^B$ appears in $\calH_j$ at location $j_1$ then the product meta-symbol $x \otimes y$ appears $\ell_A$ times at locations $j_1, \ell_B + j_1, 2 \ell_B + j_1, \dots, (\ell_A - 1)\ell_B + j_1$. %\aditya{remind - define location appropriately}
The result follows by realizing that there are $\ell_B$ subgroups where meta-symbol $y$ appears. Moreover, $y \in \calP_i^B$ appears at all locations $0, \dots, \ell_B -1$ across these subgroups.
\end{proof}

\begin{theorem}
\label{beta_matmatstrQ}
If we use a single parallel class $\calP_A$ for $\bfA$ and a single parallel class $\calP_B$ for $\bfB$ across all the worker groups, then the scheme described in Alg. \ref{Alg:betalevelcoding_matmat} will be resilient to $s = c \ell - \beta$ stragglers and will have, $Q = n \ell - \frac{c \ell (\ell + 1)}{2} + \ell (\beta - 1) + 1$, where $\ell = \ell_A \ell_B$ and $\beta = \beta_A \beta_B \leq c$.
\end{theorem}

\begin{proof}
The proof is very similar to the proof of Theorem \ref{thm:beta_matvecstrQ} once we use the fact that each product meta-symbol appears in all locations $0, \dots, \ell-1$ within the group in which it appears ({\it cf.} Lemma \ref{lem:beta_index}).
\end{proof}

It can be verified that the distributed scheme shown in Fig. \ref{beta_matmat} is resilient to $s = c \ell - \beta = 4 \times 4 - 4 = 12$ stragglers and has  $Q = 117$. Theorem \ref{beta_matmatstrQ} provides the value for $s$ and $Q$ for distributed matrix-matrix multiplication when $\beta \leq c$. In Appendix \ref{app:diffc}, we explicitly calculate the values for $s$ and $Q$ for the case when $\beta > c$. 
\begin{remark}
Similar to the matrix-vector case, the uncoded matrix-matrix multiplication scheme can also be thought as a special case of $\beta$-level coding scheme with $\beta = 1$. The lower bound given in \eqref{eq:Qlb} is matched by the proposed  scheme here with $\beta_A = \beta_B = 1$ (i.e., the uncoded scheme). An example appears in Fig. \ref{uncoded_matmat_prop} where we have $n = 12$ workers and the master node can recover the final product as soon as it receives $Q = 52$ symbols across all the workers.
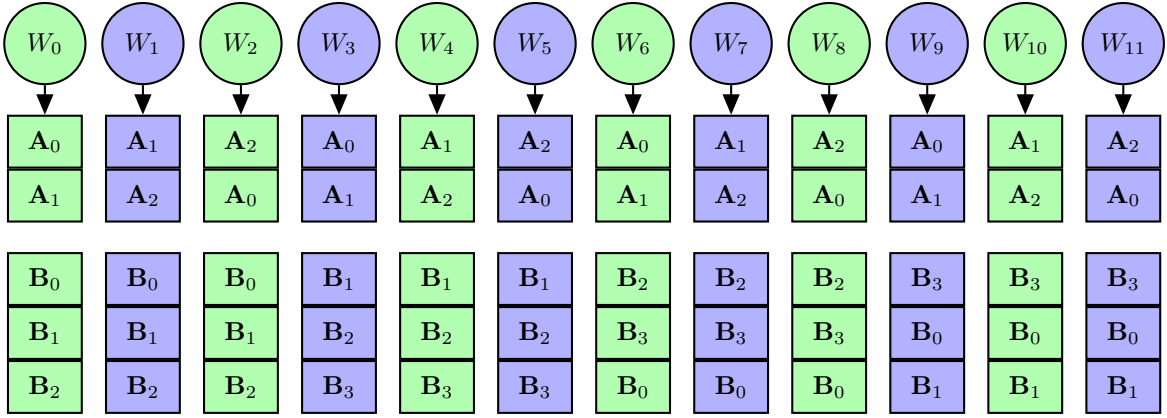
\begin{figure*}[t]
\centering
\captionsetup{justification=centering}
\resizebox{0.95\linewidth}{!}{
\begin{tikzpicture}[auto, thick, node distance=2cm, >=triangle 45]
\draw
    node [sum1, fill=green!30] (blk1) {$W_0$}
    node [sum1, fill=blue!30,right = 0.2 cm of blk1] (blk2) {$W_1$}
    node [sum1, fill=green!30,right = 0.2 cm of blk2] (blk3) {$W_2$}
    node [sum1, fill=blue!30,right = 0.2 cm of blk3] (blk4) {$W_3$}
    node [sum1, fill=green!30,right = 0.2 cm of blk4] (blk5) {$W_4$}
	node [sum1, fill=blue!30,right = 0.2 cm of blk5] (blk6) {$W_5$}
	node [sum1, fill=green!30,right = 0.2 cm of blk6] (blk7) {$W_6$}
	node [sum1, fill=blue!30,right = 0.2 cm of blk7] (blk8) {$W_7$}
	node [sum1, fill=green!30,right = 0.2 cm of blk8] (blk9) {$W_8$}
	node [sum1, fill=blue!30,right = 0.2 cm of blk9] (blk10) {$W_9$}
	node [sum1, fill=green!30,right = 0.2 cm of blk10] (blk11) {$W_{10}$}
	node [sum1, fill=blue!30,right = 0.2 cm of blk11] (blk12) {$W_{11}$}

    node [block, minimum width = 1 cm, fill=green!30,below = 0.4 cm of blk1] (blk011) {$\bfA_0$}
    node [block, minimum width = 1 cm,fill=green!30,below = 0.0005 cm of blk011] (blk012) {$\bfA_1$}
    
    node [block, minimum width = 1 cm,fill=blue!30,below = 0.4 cm of blk2] (blk21) {$\bfA_1$}
    node [block, minimum width = 1 cm,fill=blue!30,below = 0.0005 cm of blk21] (blk22) {$\bfA_2$}
    
    node [block, minimum width = 1 cm,fill=green!30,below = 0.4 cm of blk3] (blk31) {$\bfA_2$}
    node [block, minimum width = 1 cm,fill=green!30,below = 0.0005 cm of blk31] (blk32) {$\bfA_0$}
    
    node [block, minimum width = 1 cm,fill=blue!30,below = 0.4 cm of blk4] (blk41) {$\bfA_0$}
    node [block, minimum width = 1 cm,fill=blue!30,below = 0.0005 cm of blk41] (blk42) {$\bfA_1$}
    
    node [block, minimum width = 1 cm,fill=green!30,below = 0.4 cm of blk5] (blk51) {$\bfA_1$}
    node [block, minimum width = 1 cm,fill=green!30,below = 0.0005 cm of blk51] (blk52) {$\bfA_2$}
    
    node [block, minimum width = 1 cm,fill=blue!30,below = 0.4 cm of blk6] (blk61) {$\bfA_2$}
    node [block, minimum width = 1 cm,fill=blue!30,below = 0.0005 cm of blk61] (blk62) {$\bfA_0$}
    
    node [block, minimum width = 1 cm,fill=green!30,below = 0.4 cm of blk7] (blk71) {$\bfA_0$}
    node [block, minimum width = 1 cm,fill=green!30,below = 0.0005 cm of blk71] (blk72) {$\bfA_1$}
    
    node [block, minimum width = 1 cm,fill=blue!30,below = 0.4 cm of blk8] (blk81) {$\bfA_1$}
    node [block, minimum width = 1 cm,fill=blue!30,below = 0.0005 cm of blk81] (blk82) {$\bfA_2$}
    
    node [block, minimum width = 1 cm,fill=green!30,below = 0.4 cm of blk9] (blk91) {$\bfA_2$}
    node [block, minimum width = 1 cm,fill=green!30,below = 0.0005 cm of blk91] (blk92) {$\bfA_0$}
    
    node [block, minimum width = 1 cm,fill=blue!30,below = 0.4 cm of blk10] (blk101) {$\bfA_0$}
    node [block, minimum width = 1 cm,fill=blue!30,below = 0.0005 cm of blk101] (blk102) {$\bfA_1$}
    
    node [block, minimum width = 1 cm,fill=green!30,below = 0.4 cm of blk11] (blk111) {$\bfA_1$}
    node [block, minimum width = 1 cm,fill=green!30,below = 0.0005 cm of blk111] (blk112) {$\bfA_2$}
    
    node [block, minimum width = 1 cm,fill=blue!30,below = 0.4 cm of blk12] (blk121) {$\bfA_2$}
    node [block, minimum width = 1 cm,fill=blue!30,below = 0.0005 cm of blk121] (blk122) {$\bfA_0$}
    
    node [block, minimum width = 1 cm,fill=green!30,below = 0.4 cm of blk012] (blk0111) {$\bfB_0$}
    node [block, minimum width = 1 cm,fill=green!30,below = 0.0005 cm of blk0111] (blk0112) {$\bfB_1$}
    node [block, minimum width = 1 cm,fill=green!30,below = 0.0005 cm of blk0112] (blk0113) {$\bfB_2$}
       
    node [block, minimum width = 1 cm,fill=blue!30,below = 0.4 cm of blk22] (blk211) {$\bfB_0$}
    node [block, minimum width = 1 cm,fill=blue!30,below = 0.0005 cm of blk211] (blk212) {$\bfB_1$}
    node [block, minimum width = 1 cm,fill=blue!30,below = 0.0005 cm of blk212] (blk0213) {$\bfB_2$}
    
    node [block, minimum width = 1 cm,fill=green!30,below = 0.4 cm of blk32] (blk0111) {$\bfB_0$}
    node [block, minimum width = 1 cm,fill=green!30,below = 0.0005 cm of blk0111] (blk0112) {$\bfB_1$}
    node [block, minimum width = 1 cm,fill=green!30,below = 0.0005 cm of blk0112] (blk0113) {$\bfB_2$}
    
    node [block, minimum width = 1 cm,fill=blue!30,below = 0.4 cm of blk42] (blk211) {$\bfB_1$}
    node [block, minimum width = 1 cm,fill=blue!30,below = 0.0005 cm of blk211] (blk212) {$\bfB_2$}
    node [block, minimum width = 1 cm,fill=blue!30,below = 0.0005 cm of blk212] (blk0213) {$\bfB_3$}
    
    node [block, minimum width = 1 cm,fill=green!30,below = 0.4 cm of blk52] (blk0111) {$\bfB_1$}
    node [block, minimum width = 1 cm,fill=green!30,below = 0.0005 cm of blk0111] (blk0112) {$\bfB_2$}
    node [block, minimum width = 1 cm,fill=green!30,below = 0.0005 cm of blk0112] (blk0113) {$\bfB_3$}
    
    node [block, minimum width = 1 cm,fill=blue!30,below = 0.4 cm of blk62] (blk211) {$\bfB_1$}
    node [block, minimum width = 1 cm,fill=blue!30,below = 0.0005 cm of blk211] (blk212) {$\bfB_2$}
    node [block, minimum width = 1 cm,fill=blue!30,below = 0.0005 cm of blk212] (blk0213) {$\bfB_3$}
    
    node [block, minimum width = 1 cm,fill=green!30,below = 0.4 cm of blk72] (blk0111) {$\bfB_2$}
    node [block, minimum width = 1 cm,fill=green!30,below = 0.0005 cm of blk0111] (blk0112) {$\bfB_3$}
    node [block, minimum width = 1 cm,fill=green!30,below = 0.0005 cm of blk0112] (blk0113) {$\bfB_0$}
    
    node [block, minimum width = 1 cm,fill=blue!30,below = 0.4 cm of blk82] (blk211) {$\bfB_2$}
    node [block, minimum width = 1 cm,fill=blue!30,below = 0.0005 cm of blk211] (blk212) {$\bfB_3$}
    node [block, minimum width = 1 cm,fill=blue!30,below = 0.0005 cm of blk212] (blk0213) {$\bfB_0$}
    
    node [block, minimum width = 1 cm,fill=green!30,below = 0.4 cm of blk92] (blk0111) {$\bfB_2$}
    node [block, minimum width = 1 cm,fill=green!30,below = 0.0005 cm of blk0111] (blk0112) {$\bfB_3$}
    node [block, minimum width = 1 cm,fill=green!30,below = 0.0005 cm of blk0112] (blk0113) {$\bfB_0$}
    
    node [block, minimum width = 1 cm,fill=blue!30,below = 0.4 cm of blk102] (blk211) {$\bfB_3$}
    node [block, minimum width = 1 cm,fill=blue!30,below = 0.0005 cm of blk211] (blk212) {$\bfB_0$}
    node [block, minimum width = 1 cm,fill=blue!30,below = 0.0005 cm of blk212] (blk0213) {$\bfB_1$}
    
    node [block, minimum width = 1 cm,fill=green!30,below = 0.4 cm of blk112] (blk0111) {$\bfB_3$}
    node [block, minimum width = 1 cm,fill=green!30,below = 0.0005 cm of blk0111] (blk0112) {$\bfB_0$}
    node [block, minimum width = 1 cm,fill=green!30,below = 0.0005 cm of blk0112] (blk0113) {$\bfB_1$}
    
    node [block, minimum width = 1 cm,fill=blue!30,below = 0.4 cm of blk122] (blk211) {$\bfB_3$}
    node [block, minimum width = 1 cm,fill=blue!30,below = 0.0005 cm of blk211] (blk212) {$\bfB_0$}
    node [block, minimum width = 1 cm,fill=blue!30,below = 0.0005 cm of blk212] (blk0213) {$\bfB_1$}

    ;

\draw[->](blk1) -- node{} (blk011);
\draw[->](blk2) -- node{} (blk21);
\draw[->](blk3) -- node{} (blk31);
\draw[->](blk4) -- node{} (blk41);
\draw[->](blk5) -- node{} (blk51);
\draw[->](blk6) -- node{} (blk61);
\draw[->](blk7) -- node{} (blk71);
\draw[->](blk8) -- node{} (blk81);
\draw[->](blk9) -- node{} (blk91);
\draw[->](blk10) -- node{} (blk101);
\draw[->](blk11) -- node{} (blk111);
\draw[->](blk12) -- node{} (blk121);

\end{tikzpicture}
}
\centering
\caption{\small Uncoded matrix-matrix multiplication with $n = 12$ and $s = 5$ with $\gamma_A = \frac{2}{3}$ and $\gamma_B = \frac{3}{4}$ where $\Delta_A = 3$ and $\Delta_B = 4$.}

\label{uncoded_matmat_prop}
\end{figure*} 
\end{remark}

In the matrix-matrix case for $\beta_A =2, \beta_B = 1$ we can show that using different parallel classes can improve the straggler resilience of the system. The corresponding $Q$ analysis is harder to do and is part of future work.
\begin{theorem}
\label{thm:beta2_matmatstrQ}
Let $\ell_A \leq \frac{\Delta_A}{2}-2$ and $\Delta_A \geq 8$. If we use the parallel classes in \eqref{paracls} for encoding $\bfA$, then the matrix-matrix multiplication scheme described in Alg. \ref{Alg:betalevelcoding_matmat} will be resilient to $s = 2\ell  - 1$ stragglers, when $\beta_A = 2$ and $\beta_B = 1$ such that $c = \beta = 2$.
\end{theorem}
\begin{proof}
Consider the set $\calB_m = \{\bfA_0^T \bfB_m, \bfA_1^T \bfB_m, \dots, \bfA_{\Delta_A-1}^T \bfB_m \}$, i.e., the set of all unknowns corresponding to $\bfB_m$, for $m = 0, 1, \dots, \Delta_B - 1$, so $|\calB_m| = \Delta_A$. As $\bfB$ is uncoded, the equations consisting of the unknowns in $\calB_m$ are disjoint of the equations consisting of the unknowns of $\calB_p$, ($m \neq p$). Thus, we can form $\calG_{dec}^m$ using the unknowns corresponding to the set $\calB_m$ and analyze the decoding using it. The rest of the argument follows analogous to the proof of Theorem \ref{thm:beta2_matvecstrQ}.
%
%AR: put complete argument in Anindya dissertation
\begin{comment}
Since any product meta-symbol appears $\ell = \ell_A \ell_B$ times within any worker group if there are at $2\ell-1$ stragglers, it is evident that $\bfG_{dec}^m$ is such that each unknown has degree at least one. Let $X_i$ and $X_j$ denote the subset of worker nodes where unknowns $\bfA_i^T \bfB_m$ and $\bfA_j^T \bfB_m$ appear within a meta-symbol, so that $|X_i|=|X_j| = 2\ell$. Now suppose by way of contradiction that within any $\calG_{dec}^m$, we have two unknowns $\bfA_i^T \bfB_m$ and $\bfA_j^T \bfB_m$ where $i < j$ that appear exactly once across the remaining $n- 2\ell + 1$ workers.  

Furthermore, $|X_i \cap X_j| \leq 2 \ell -2$. To see this we note that if $\{i,j\}$ appear together w.l.o.g. in $\calP_0$, then they appear together in exactly $(\ell_A - 2)\ell_B = \ell - 2\ell_B$ workers in $\calG_1$, thus $|X_i \cap X_j| = \ell + \ell - 2\ell_B$, where $\ell_B \geq 1$. On the other hand if $i$ and $j$ do not appear together in either $\calP_0$ or $\calP_1$ then they appear together in the workers of each group at most $(\ell_A-1) \ell_B \leq \ell - 1$ times (since $\ell_B \geq 1$), so the claim holds. Thus, 
\begin{align*}
|X_i \cup X_j| &= |X_i| + |X_j| - |X_i \cap X_j|\\
&\geq 2\ell + 2.
\end{align*}
This means that if $2\ell-1$ workers are stragglers then unknowns $\bfA_i^T \bfB_m$ or $\bfA_j^T \bfB_m$ or both appear in at least three nodes, i.e., at least one of them appears at least twice. This contradicts our original assumption. 
\end{comment}
\end{proof}

\subsection{Coded at bottom scheme}
\label{sec:coded_blocks_bottom}
Intuitively, the $\beta$-level coding schemes can be improved if we allow for the inclusion of some densely coded block-columns. We now consider a variant of the uncoded scheme where such densely coded block-columns are added at the end of uncoded computations. This improves both the straggler resilience and the $Q$ value of the scheme. %We discuss the matrix-vector scheme in the discussion below. 

We now assume that each node receives $\gamma = \gamma_{u} + \gamma_{c}$ fraction of the columns of $\bfA$ and the vector $\bfx$. Here $\gamma_{u}$ corresponds to the storage fraction of the uncoded parts of $\bfA$, whereas $\gamma_{c}$ corresponds to the coded portion. The coded blocks appear at the bottom of each node. Thus, under normal operating circumstances (no slow or failed nodes), the master node can simply decode the intended result from the uncoded computations. If some nodes are operating slower than normal, then the coded computations can be leveraged. 

As in the uncoded setup let  $\ell_u = \Delta \gamma_u$ be the number of uncoded block-columns and $r_u$ be the replication factor. Likewise $\ell_c = \Delta \gamma_c$ represents the number of coded blocks in each worker. In this construction we set $\Delta = n$ so that $r_u = \ell_u$. In this case, the results from Theorem \ref{thm:beta_matvecstrQ} immediately imply that $Q \geq \max (\Delta, \Delta r_u - \frac{r_u}{2} (\ell_u + 1) +1)$.  This follows by applying $\beta = 1$ to the uncoded part of the solution where $r_u = \ell_u$. A construction that meets these bounds is outlined in Algorithm \ref{Alg:Cyclic_Partially_Uncoded_matvec}. The algorithm uses a random matrix of dimension $n \ell_c \times \Delta$.% in Algorithm \ref{Alg:Cyclic_Partially_Uncoded}.

\begin{algorithm}[t]
	\caption{Cyclic coded at the bottom scheme for distributed matrix-vector multiplication}
	\label{Alg:Cyclic_Partially_Uncoded_matvec}
   \SetKwInOut{Input}{Input}
   \SetKwInOut{Output}{Output}
   \Input{Matrix $\bfA$ and vector $\bfx$, $n$-number of worker nodes, total storage capacity fraction $\gamma$, replication factor for uncoded portion $r_u$.}
   Set $\Delta = n$, $\ell_u = r_u$, $\ell = \gamma \Delta$, $\ell_c = \ell - \ell_u$\; 
   %Determine a random matrix $\calR$ of dimension $n \ell_c \times \Delta$\;
   Partition $\bfA$ into $\Delta$ block-columns $\bfA_0, \bfA_1, \dots, \bfA_{\Delta-1}$\;
   \For{$i\gets 0$ \KwTo $n-1$}{
   Define $T = \left\lbrace i, i+1, \dots, i + \ell_u - 1 \right\rbrace$ (mod $\Delta$)\;
   Assign all $\bfA_{m}$'s sequentially from top to bottom to worker node $i$, where $m \in T$\;
   %Assign $\bfA_{i}, \bfA_{i+1}, \dots, \bfA_{i + \ell_u -1}$ from top to bottom (subscripts reduced modulo $\Delta$) to worker node $i$\;
   Assign $\ell_c$ different random linear combinations of $\bfA_m$'s for $m \notin T$\;
   
    %\For{$j \gets 1$  \KwTo $\ell_c$}{
     %   Define $T = \{i, i+1, \dots, i + \ell_u - 1\} \mod \Delta$\;
      %  Pick a row $\bfr$ of $\calR$  and assign coded block $\sum\limits_{k=0}^{\Delta-1} \bfr_k \mathds{1}_{k \notin T} \bfA_k$\;
      %  Remove $\bfr$ from $\calR$, so $\mathcal{R} \gets \mathcal{R} \backslash \{\bfr\}$ \;
%    }
   }
   \Output{Cyclic coded at the bottom scheme for matrix-vector multiplication.}
\end{algorithm}

\begin{theorem}
\label{thm:Q_partially_coded_end}
The scheme in Alg.\ref{Alg:Cyclic_Partially_Uncoded_matvec} satisfies $Q = \max (\Delta, \Delta r_u - \frac{r_u}{2} (\ell_u + 1) +1)$. Furthermore, it is resilient to  $\floor[\bigg]{\frac{n^2 \gamma_c + n \gamma_u - 1}{n \gamma_c + 1}}$ stragglers.
\end{theorem}
\begin{proof}
The detailed proof is discussed in Appendix \ref{App:coded_bot}
\end{proof}

\begin{example}
\label{eg:coded_bottom_example}
Consider the setting where we have $n = 5$ workers with $\gamma = \frac{3}{5}$ where we set $\Delta = n = 5$. Fig. \ref{uncoded_matvec_prop} shows the job assignments according to the uncoded scheme ($\beta = 1$). According to Theorem \ref{thm:beta_matvecstrQ} in Section \ref{sec:beta_level}, the system is resilient to $\beta (n\gamma - 1) = 2$ stragglers and $Q = 5 \times 3 - \frac{3 \times 4}{2} + 1 = 10$ which can be verified from Fig. \ref{uncoded_matvec_prop}.

Now we assume that the whole storage fraction can be distributed into an uncoded storage fraction $\gamma_{u} = \frac{2}{5}$ and a coded storage fraction $\gamma_{c} = \frac{1}{5}$. Using the coded scheme, we get the job assignments shown in Fig. \ref{coded_bottom_prop}. This scheme is resilient to $\floor[\bigg]{\frac{n^2 \gamma_c + n \gamma_u - 1}{n \gamma_c + 1}} = 3$ stragglers and it can be verified from  that $\bfA^T \bfx$ can be computed once any $ Q = \Delta r_u - \frac{r_u}{2}( \ell_u + 1 ) + 1 = 8$ block-columns have been processed. Thus, we can conclude that introducing a single coded block in each worker (at the bottom), helps to improve both $Q$ and the straggler resilience of the system as compared to an uncoded system.
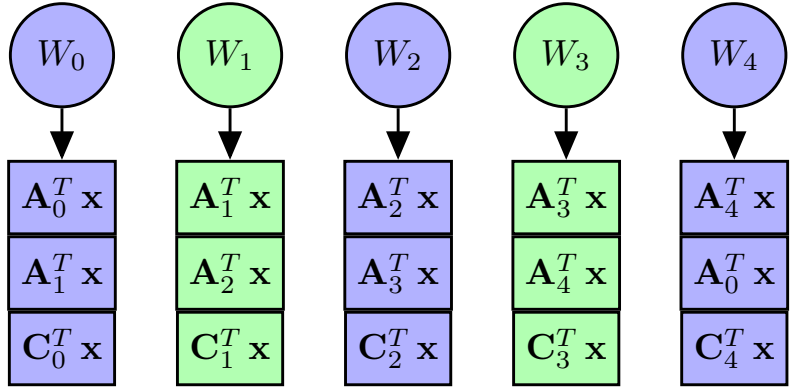
\begin{figure}[t]
\centering
\captionsetup{justification=centering}
\resizebox{0.65\linewidth}{!}{
\begin{tikzpicture}[auto, thick, node distance=2cm, >=triangle 45]

\draw
	node at (0,0)[right=-3mm]{}
	node [sum1, minimum size = 1cm, fill=blue!30] (blk1){$W_0$}
    node [sum1, minimum size = 1cm, fill=green!30,right = 0.6cm of blk1] (blk2) {$W_1$}
    node [sum1, minimum size = 1cm, fill=blue!30,right = 0.6cm of blk2] (blk3) {$W_2$}
    node [sum1, minimum size = 1cm, fill=green!30,right = 0.6cm of blk3] (blk4) {$W_3$}
    node [sum1, minimum size = 1cm, fill=blue!30,right = 0.6cm of blk4] (blk5) {$W_4$}
    
    node [block, fill=blue!30,below = 0.5 cm of blk1] (blk11) {$\bfA_{0}^T \, \bfx$}
    node [block, fill=blue!30,below = 0.0005 cm of blk11] (blk12) {$\bfA_{1}^T\, \bfx$}
    node [block, fill=blue!30,below = 0.0005 cm of blk12] (blk13) {$\bfC_{0}^T\, \bfx$}

    node [block, fill=green!30,below = 0.5 cm of blk2] (blk21) {$\bfA_{1}^T\, \bfx$}
    node [block, fill=green!30,below = 0.0005 cm of blk21] (blk22) {$\bfA_{2}^T\, \bfx$}
    node [block, fill=green!30,below = 0.0005 cm of blk22] (blk23) {$\bfC_{1}^T\, \bfx$}

    node [block, fill=blue!30,below = 0.5 cm of blk3] (blk31) {$\bfA_{2}^T\, \bfx$}
    node [block, fill=blue!30,below = 0.0005 cm of blk31] (blk32) {$\bfA_{3}^T\, \bfx$}
    node [block, fill=blue!30,below = 0.0005 cm of blk32] (blk33) {$\bfC_{2}^T\, \bfx$}

    node [block, fill=green!30,below = 0.5 cm of blk4] (blk41) {$\bfA_{3}^T\, \bfx$}
    node [block, fill=green!30,below = 0.0005 cm of blk41] (blk42) {$\bfA_{4}^T\, \bfx$}
    node [block, fill=green!30,below = 0.0005 cm of blk42] (blk43) {$\bfC_{3}^T\, \bfx$}

    node [block, fill=blue!30,below = 0.5 cm of blk5] (blk51) {$\bfA_{4}^T\, \bfx$}
    node [block, fill=blue!30,below = 0.0005 cm of blk51] (blk52) {$\bfA_{0}^T\, \bfx$}
    node [block, fill=blue!30,below = 0.0005 cm of blk52] (blk53) {$\bfC_{4}^T\, \bfx$}
    ; 
\draw[->](blk1) -- node{} (blk11);
\draw[->](blk2) -- node{} (blk21);
\draw[->](blk3) -- node{} (blk31);
\draw[->](blk4) -- node{} (blk41);
\draw[->](blk5) -- node{} (blk51);

%\draw[thick,dotted] ($(blk11.north west)+(-0.2,0.2)$)  rectangle ($(blk13.south east)+(0.2,-0.2)$);
%\draw[thick,dotted] ($(blk21.north west)+(-0.2,0.2)$)  rectangle ($(blk23.south east)+(0.2,-0.2)$);
%\draw[thick,dotted] ($(blk31.north west)+(-0.2,0.2)$)  rectangle ($(blk32.south east)+(0.2,-0.2)$);
%\draw[thick,dotted] ($(blk41.north west)+(-0.2,0.2)$)  rectangle ($(blk41.south east)+(0.2,-0.2)$);

\end{tikzpicture}
}
\caption{\small Partitioning matrix $A$ into five submatrices and assigning two uncoded and one coded task to each of the five workers. The coded submatrix assigned to $W_i$ is denoted as $\bfC_i$.}
\label{coded_bottom_prop}
\end{figure} 
\end{example}

Similar schemes can be arrived at for the matrix-matrix case. 
%\subsection{Coded at the Bottom Approach for Distributed Matrix-matrix Multiplication}
We assume that the uncoded storage fraction for $\bfA$ is $\gamma_{Au} = \frac{a_u}{a_2}$ and the coded storage fraction is $\gamma_{Ac} = \frac{a_c}{a_2}$, so that the total storage fraction is $\gamma_{A} = \frac{a_1}{a_2}$. Each worker also receives $\gamma_{B} = \frac{b_1}{b_2}$ fraction of the uncoded columns of matrix $B$. 

%It should be noted that all of $n \gamma_{Au}$, $n \gamma_{Ac}$, $n \gamma_{B}$ and $n \gamma_{A} \gamma_B$, need to be integers, and thus, $n$ needs to be an integer multiple of $\left( a_2 \times b_2 \right)$.

\begin{algorithm}[h]
	\caption{Cyclic coded at the bottom scheme for distributed matrix-matrix multiplication}
	\label{Alg:Cyclic_Partially_Uncoded_matmat}
   \SetKwInOut{Input}{Input}
   \SetKwInOut{Output}{Output}
   \Input{Matrices $\bfA$ and $\bfB$, $n$-number of workers. Storage fractions $\gamma_{Au} = \frac{a_u}{a_2}$ and $\gamma_{Ac} = \frac{a_c}{a_2}$, so that $\gamma_{A} = \frac{a_1}{a_2}$ and $\gamma_{B} = \frac{b_1}{b_2}$.}
   Set $\Delta_A = a_2$, $\Delta_B = m b_2$, $m = \frac{n}{\left(a_2 \times b_2 \right)}$. Partition $\bfA$ and $\bfB$ into $\Delta_A$ and $\Delta_B$ block-columns, respectively\;
   \For{$i\gets 0$ \KwTo $n - 1$}{
   Define $T = \left\lbrace i, i+1, \dots, i + a_u - 1 \right\rbrace$ (mod $\Delta_A$)\;
   Assign all $\bfA_{m}$'s sequentially from top to bottom to worker node $i$, where $m \in T$\;
   Assign $a_c$ different random linear combinations of $\bfA_m$'s for $m \notin T$\;
   %Assign $\bfA_{i}, \bfA_{i+1}, \dots, \bfA_{i + a_u - 1}$ from top to bottom (subscripts reduced modulo $\Delta_A$) to worker $i$\;
   %Assign $a_c$ coded submatrices from $\bfA$ using the rows of $R_{Ai}$ \; 
   $j \gets \floor{\frac{i}{a_2}}$ and assign $\bfB_{j}, \bfB_{j+1}, \dots, \bfB_{j + m b_1 - 1}$ from top to bottom (subscripts reduced modulo $\Delta_B$) to worker node $i$\;
   }
   \Output{Cyclic coded at the bottom scheme for matrix-matrix multiplication.}
\end{algorithm}

\begin{theorem}
\label{str:coded_bottom_matmat}
The recovery threshold for the matrix-matrix multiplication scheme Alg. \ref{Alg:Cyclic_Partially_Uncoded_matmat} is given by, $\tau = n - m a_2 b_1 + \kappa_{min}$, where $\kappa_{min}$ is the minimum positive integer for $\kappa$ satisfying the inequality
\begin{align*}
\ceil[\bigg]{\frac{\kappa}{m b_1}}  + \kappa a_c  \geq a_2 - a_u + 1 .
\end{align*}
\end{theorem}

\begin{proof}
The detailed proof is discussed in Appendix \ref{App:coded_matmat}.
\end{proof}

\begin{example}
\begin{figure*}[t]
\centering
\captionsetup{justification=centering}
\resizebox{0.99\linewidth}{!}{
\begin{tikzpicture}[auto, thick, node distance=2cm, >=triangle 45]
\draw
    node [sum1, minimum size = 1cm, fill=green!30] (blk1) {$W_0$}
    node [sum1, minimum size = 1cm, fill=blue!30,right = 0.3 cm of blk1] (blk2) {$W_1$}
    node [sum1, minimum size = 1cm, fill=green!30,right = 0.3 cm of blk2] (blk3) {$W_2$}
    node [sum1, minimum size = 1cm, fill=blue!30,right = 0.3 cm of blk3] (blk4) {$W_3$}
    node [sum1, minimum size = 1cm, fill=green!30,right = 0.3 cm of blk4] (blk5) {$W_4$}
	node [sum1, minimum size = 1cm, fill=blue!30,right = 0.3 cm of blk5] (blk6) {$W_5$}
	node [sum1, minimum size = 1cm, fill=green!30,right = 0.3 cm of blk6] (blk7) {$W_6$}
	node [sum1, minimum size = 1cm, fill=blue!30,right = 0.3 cm of blk7] (blk8) {$W_7$}
	node [sum1, minimum size = 1cm, fill=green!30,right = 0.3 cm of blk8] (blk9) {$W_8$}
	node [sum1, minimum size = 1cm, fill=blue!30,right = 0.3 cm of blk9] (blk10) {$W_9$}
	node [sum1, minimum size = 1cm, fill=green!30,right = 0.3 cm of blk10] (blk11) {$W_{10}$}
	node [sum1, minimum size = 1cm, fill=blue!30,right = 0.3 cm of blk11] (blk12) {$W_{11}$}

    node [block, minimum width = 0.9 cm, fill=green!30,below = 0.4 cm of blk1] (blk011) {$\bfA_0$}
    node [block, minimum width = 0.9 cm, fill=green!30,below = 0.0005 cm of blk011] (blk012) {$\bfC_{0}$}
    
    node [block, minimum width = 0.9 cm,fill=blue!30,below = 0.4 cm of blk2] (blk21) {$\bfA_1$}
    node [block, minimum width = 0.9 cm, fill=blue!30,below = 0.0005 cm of blk21] (blk22) {$\bfC_{1}$}
    
    node [block, minimum width = 0.9 cm,fill=green!30,below = 0.4 cm of blk3] (blk31) {$\bfA_2$}
    node [block, minimum width = 0.9 cm, fill=green!30,below = 0.0005 cm of blk31] (blk32) {$\bfC_{2}$}
    
    node [block, minimum width = 0.9 cm,fill=blue!30,below = 0.4 cm of blk4] (blk41) {$\bfA_0$}
    node [block, minimum width = 0.9 cm, fill=blue!30,below = 0.0005 cm of blk41] (blk42) {$\bfC_{3}$}
    
    node [block, minimum width = 0.9 cm,fill=green!30,below = 0.4 cm of blk5] (blk51) {$\bfA_1$}
    node [block, minimum width = 0.9 cm, fill=green!30,below = 0.0005 cm of blk51] (blk52) {$\bfC_{4}$}
    
    node [block, minimum width = 0.9 cm,fill=blue!30,below = 0.4 cm of blk6] (blk61) {$\bfA_2$}
    node [block, minimum width = 0.9 cm, fill=blue!30,below = 0.0005 cm of blk61] (blk62) {$\bfC_{5}$}
    
    node [block, minimum width = 0.9 cm,fill=green!30,below = 0.4 cm of blk7] (blk71) {$\bfA_0$}
    node [block, minimum width = 0.9 cm, fill=green!30,below = 0.0005 cm of blk71] (blk72) {$\bfC_{6}$}
    
    node [block, minimum width = 0.9 cm,fill=blue!30,below = 0.4 cm of blk8] (blk81) {$\bfA_1$}
    node [block, minimum width = 0.9 cm, fill=blue!30,below = 0.0005 cm of blk81] (blk82) {$\bfC_{7}$}
    
    node [block, minimum width = 0.9 cm,fill=green!30,below = 0.4 cm of blk9] (blk91) {$\bfA_2$}
    node [block, minimum width = 0.9 cm, fill=green!30,below = 0.0005 cm of blk91] (blk92) {$\bfC_{8}$}
    
    node [block, minimum width = 0.9 cm,fill=blue!30,below = 0.4 cm of blk10] (blk101) {$\bfA_0$}
    node [block, minimum width = 0.9 cm,fill=blue!30,below = 0.0005 cm of blk101] (blk102) {$\bfC_{9}$}
    
    node [block, minimum width = 0.9 cm,fill=green!30,below = 0.4 cm of blk11] (blk111) {$\bfA_1$}
    node [block, minimum width = 0.9 cm, fill=green!30,below = 0.0005 cm of blk111] (blk112) {$\bfC_{10}$}
    
    node [block, minimum width = 0.9 cm,fill=blue!30,below = 0.4 cm of blk12] (blk121) {$\bfA_2$}
    node [block, minimum width = 0.9 cm, fill=blue!30,below = 0.0005 cm of blk121] (blk122) {$\bfC_{11}$}
    
    node [block, minimum width = 0.9 cm,fill=green!30,below = 0.4 cm of blk012] (blk0111) {$\bfB_0$}
    node [block, minimum width = 0.9 cm,fill=green!30,below = 0.0005 cm of blk0111] (blk0112) {$\bfB_1$}
    node [block, minimum width = 0.9 cm,fill=green!30,below = 0.0005 cm of blk0112] (blk0113) {$\bfB_2$}
       
    node [block, minimum width = 0.9 cm,fill=blue!30,below = 0.4 cm of blk22] (blk211) {$\bfB_0$}
    node [block, minimum width = 0.9 cm,fill=blue!30,below = 0.0005 cm of blk211] (blk212) {$\bfB_1$}
    node [block, minimum width = 0.9 cm,fill=blue!30,below = 0.0005 cm of blk212] (blk0213) {$\bfB_2$}
    
    node [block, minimum width = 0.9 cm,fill=green!30,below = 0.4 cm of blk32] (blk0111) {$\bfB_0$}
    node [block, minimum width = 0.9 cm,fill=green!30,below = 0.0005 cm of blk0111] (blk0112) {$\bfB_1$}
    node [block, minimum width = 0.9 cm,fill=green!30,below = 0.0005 cm of blk0112] (blk0113) {$\bfB_2$}
    
    node [block, minimum width = 0.9 cm,fill=blue!30,below = 0.4 cm of blk42] (blk211) {$\bfB_1$}
    node [block, minimum width = 0.9 cm,fill=blue!30,below = 0.0005 cm of blk211] (blk212) {$\bfB_2$}
    node [block, minimum width = 0.9 cm,fill=blue!30,below = 0.0005 cm of blk212] (blk0213) {$\bfB_3$}
    
    node [block, minimum width = 0.9 cm,fill=green!30,below = 0.4 cm of blk52] (blk0111) {$\bfB_1$}
    node [block, minimum width = 0.9 cm,fill=green!30,below = 0.0005 cm of blk0111] (blk0112) {$\bfB_2$}
    node [block, minimum width = 0.9 cm,fill=green!30,below = 0.0005 cm of blk0112] (blk0113) {$\bfB_3$}
    
    node [block, minimum width = 0.9 cm,fill=blue!30,below = 0.4 cm of blk62] (blk211) {$\bfB_1$}
    node [block, minimum width = 0.9 cm,fill=blue!30,below = 0.0005 cm of blk211] (blk212) {$\bfB_2$}
    node [block, minimum width = 0.9 cm,fill=blue!30,below = 0.0005 cm of blk212] (blk0213) {$\bfB_3$}
    
    node [block, minimum width = 0.9 cm,fill=green!30,below = 0.4 cm of blk72] (blk0111) {$\bfB_2$}
    node [block, minimum width = 0.9 cm,fill=green!30,below = 0.0005 cm of blk0111] (blk0112) {$\bfB_3$}
    node [block, minimum width = 0.9 cm,fill=green!30,below = 0.0005 cm of blk0112] (blk0113) {$\bfB_0$}
    
    node [block, minimum width = 0.9 cm,fill=blue!30,below = 0.4 cm of blk82] (blk211) {$\bfB_2$}
    node [block, minimum width = 0.9 cm,fill=blue!30,below = 0.0005 cm of blk211] (blk212) {$\bfB_3$}
    node [block, minimum width = 0.9 cm,fill=blue!30,below = 0.0005 cm of blk212] (blk0213) {$\bfB_0$}
    
    node [block, minimum width = 0.9 cm,fill=green!30,below = 0.4 cm of blk92] (blk0111) {$\bfB_2$}
    node [block, minimum width = 0.9 cm,fill=green!30,below = 0.0005 cm of blk0111] (blk0112) {$\bfB_3$}
    node [block, minimum width = 0.9 cm,fill=green!30,below = 0.0005 cm of blk0112] (blk0113) {$\bfB_0$}
    
    node [block, minimum width = 0.9 cm,fill=blue!30,below = 0.4 cm of blk102] (blk211) {$\bfB_3$}
    node [block, minimum width = 0.9 cm,fill=blue!30,below = 0.0005 cm of blk211] (blk212) {$\bfB_0$}
    node [block, minimum width = 0.9 cm,fill=blue!30,below = 0.0005 cm of blk212] (blk0213) {$\bfB_1$}
    
    node [block, minimum width = 0.9 cm,fill=green!30,below = 0.4 cm of blk112] (blk0111) {$\bfB_3$}
    node [block, minimum width = 0.9 cm,fill=green!30,below = 0.0005 cm of blk0111] (blk0112) {$\bfB_0$}
    node [block, minimum width = 0.9 cm,fill=green!30,below = 0.0005 cm of blk0112] (blk0113) {$\bfB_1$}
    
    node [block, minimum width = 0.9 cm,fill=blue!30,below = 0.4 cm of blk122] (blk211) {$\bfB_3$}
    node [block, minimum width = 0.9 cm,fill=blue!30,below = 0.0005 cm of blk211] (blk212) {$\bfB_0$}
    node [block, minimum width = 0.9 cm,fill=blue!30,below = 0.0005 cm of blk212] (blk0213) {$\bfB_1$}

    ;

\draw[->](blk1) -- node{} (blk011);
\draw[->](blk2) -- node{} (blk21);
\draw[->](blk3) -- node{} (blk31);
\draw[->](blk4) -- node{} (blk41);
\draw[->](blk5) -- node{} (blk51);
\draw[->](blk6) -- node{} (blk61);
\draw[->](blk7) -- node{} (blk71);
\draw[->](blk8) -- node{} (blk81);
\draw[->](blk9) -- node{} (blk91);
\draw[->](blk10) -- node{} (blk101);
\draw[->](blk11) -- node{} (blk111);
\draw[->](blk12) -- node{} (blk121);

\end{tikzpicture}
}
\centering
\caption{\small Coded matrix-matrix multiplication with $n = 12$ with $\gamma_{Au} = \frac{1}{3}$, $\gamma_{Ac} = \frac{1}{3}$ and $\gamma_{B} = \frac{3}{4}$ where $\Delta_A = 3$ and $\Delta_B = 4$. The coded submatrix for $\bfA$ assigned to $W_i$ is denoted as $\bfC_i$.}

\label{coded_bottom_matmat_prop}
\end{figure*}
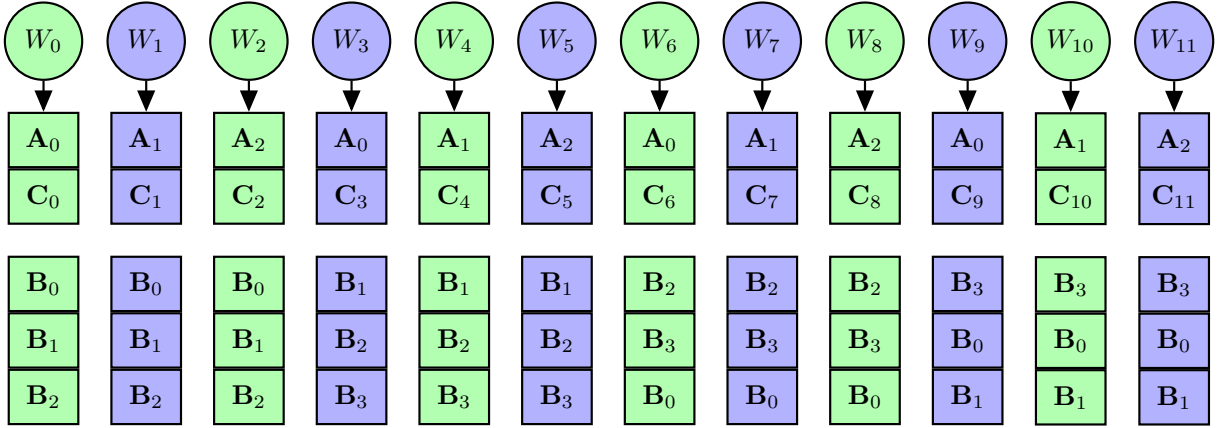 

We consider the scenario as before, where $\gamma_A = \frac{2}{3}$ and $\gamma_B = \frac{3}{4}$, and $n = 12$, so $m = \frac{12}{3 \times 4}= 1$. According to Alg. \ref{Alg:Cyclic_Partially_Uncoded_matmat}, we set $\ell_A = 2$, $\Delta_A = 3 $ and $\ell_B =3$, $\Delta_B = 4$. So, we need to recover $\Delta = \Delta_A \Delta_B = 12$ block products. Figs. \ref{uncoded_matmat_prop} and \ref{coded_bottom_matmat_prop} show the job assignments to the workers for the uncoded case and the proposed coded scheme, respectively. For the coded scheme, we assume $\gamma_{Au} = \frac{1}{3}$ and $\gamma_{Ac} = \frac{1}{3}$, and on the other hand, for the uncoded scheme, we have $\gamma_{Ac} = 0$, so $a_c = 0$. 

Now for the uncoded case, according to Theorem \ref{beta_matmatstrQ}, the recovery threshold is $\tau = n - (c \ell - \beta) = 12 - (1 \times \ell_A \ell_B - 1) = 7$. On the other hand, according to Theorem \ref{str:coded_bottom_matmat}, the recovery threshold for the coded case is, $\tau = n - m a_2 b_1 + \kappa = 12 - 3 \times 3 + 2 = 5$ since the minimum positive integer $\kappa$ that satisfies $\ceil[\big]{\frac{\kappa}{3}} + \kappa \geq 3$ is $2$. 
\end{example}

We expect that the benefits of having densely coded block-columns at the bottom should extend for the case of general $\beta > 1$ and the $Q/\Delta$ analysis should be possible to perform for the matrix-matrix case. However, this appears to be more challenging and will be investigated as part of future work.

\section{Sparsely Coded Straggler (SCS) Optimal Matrix Computations}
\label{sec:optimal}
%In this section we design a scheme for disributed matrix computation where we do not need to apply any constraint on the number of workers ($n$). This scheme performs the same as \cite{yu2017polynomial} and \cite{8919859} in terms of straggler resilience. However, \cite{yu2017polynomial} and \cite{8919859} do not deal with the partial stragglers, whereas this scheme can efficiently utilize the partial calculations done by the slower workers. 

In this section, we develop schemes for distributed matrix computations which perform optimally in terms of straggler resilience. For example, in matrix-matrix multiplication case, if the storage fractions of each worker node are $\gamma_A = 1/k_A$ and $\gamma_B = 1/k_B$ then it can be shown the lowest possible threshold is $k_A k_B$ \cite{yu2017polynomial}. Similarly, for the matrix-vector multiplication case the optimal threshold is $k_A$. Prior work has also demonstrated schemes that achieve these thresholds. In what follows, we present schemes that are similar in spirit to our constructions in Section \ref{sec:beta_level} which are suitable for sparse matrices while continuing to enjoy the optimal threshold $k_A k_B$. Moreover, unlike the previously available dense coded approaches, our proposed sparsely coded straggler (SCS) optimal scheme can utilize the partial computations of the slow workers and can provide significantly small $Q/\Delta$.

\subsection{Matrix-vector Multiplication} 
%We know that for fixed storage constraint $\gamma_A = \frac{1}{k_A}$ with arbitrary integer $k_A$ , the optimal threshold $\tau=k_A$ \cite{yu2017polynomial}. And in the best possible scheme, it should be possible to recover all $\Delta$ necessary submatrix vector products from any $\Delta$ products computed by the worker nodes. Our proposed matrix-vector multiplication scheme in Alg. \ref{Alg:Optimal_Matvec} meets both of these two bounds.

In our proposed scheme in Alg. \ref{Alg:Optimal_Matvec}, we set $\Delta = \textrm{LCM}(n, k_A)$ and assign the uncoded jobs in such a way that all the workers are assigned the uncoded jobs in an equal manner and the replication factor of the uncoded symbols over all $n$ workers is, $r_u = 1$. Thus each of the workers is assigned $\Delta/n$ uncoded jobs and the rest $\ell_c = \frac{\Delta}{k_A} - \frac{\Delta}{n}$ jobs are assigned using a random linear encoding matrix, $\calR$ of size $n \ell_c \times \Delta$. Since any $(\Delta - \lambda) \times (\Delta - \lambda)$ submatrix of $\calR$ is full rank with probability $1$, the master node can decode all the unknowns if it receives any $\lambda$ uncoded symbols and any $\Delta - \lambda$ coded symbols from all the workers. Thus we can say that $Q = \Delta$, and since each worker stores $\Delta/k_A$ block-columns, we have the recovery threshold, $\tau = \frac{\Delta}{\Delta/k_A} = k_A$. 

\begin{algorithm}[t]
	\caption{SCS Optimal scheme for matrix-vector multiplication}
	\label{Alg:Optimal_Matvec}
   \SetKwInOut{Input}{Input}
   \SetKwInOut{Output}{Output}
   \Input{Matrix $\bfA$ and vector $\bfx$, $n$-number of worker nodes, storage fraction $\gamma_A = \frac{1}{k_A}$.}
   Set $\Delta = \textrm{LCM}(n, k_A)$. Partition $\bfA$ into $\Delta$ block-columns $\bfA_0, \bfA_1, \dots, \bfA_{\Delta-1}$\;
   Number of coded submatrices of $\bfA$ in each worker node, $\ell_c = \frac{\Delta}{k_A} - \frac{\Delta}{n}$\;
   \For{$i\gets 0$ \KwTo $n-1$}{
   $u \gets i \times \frac{\Delta}{n}$\; 
   Define $T = \left\lbrace u, u+1, \dots, u + \frac{\Delta}{n} - 1 \right\rbrace$ (mod $\Delta$)\;
   Assign all $\bfA_{m}$'s sequentially from top to bottom to worker node $i$, where $m \in T$\;
   Assign $\ell_c$ different random linear combinations of $\bfA_m$'s for $m \notin T$\;
   }
   \Output{$\langle n, \gamma_A \rangle$ SCS optimal-scheme for matrix-vector multiplication with optimal $Q/\Delta$.}
\end{algorithm}

\begin{example}
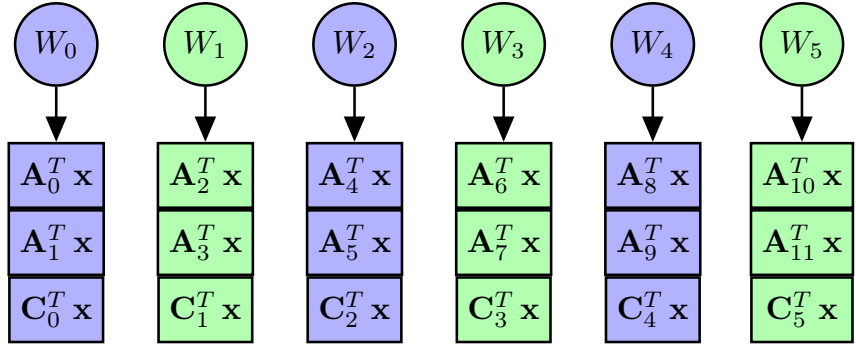
\begin{figure}[t]
\centering
\captionsetup{justification=centering}
\resizebox{0.7\linewidth}{!}{
\begin{tikzpicture}[auto, thick, node distance=2cm, >=triangle 45]

\draw
	node at (0,0)[right=-3mm]{}
	node [sum, fill=blue!30] (blk1){$W_0$}
    node [sum, fill=green!30,right = 0.7cm of blk1] (blk2) {$W_1$}
    node [sum, fill=blue!30,right = 0.7cm of blk2] (blk3) {$W_2$}
    node [sum, fill=green!30,right = 0.7cm of blk3] (blk4) {$W_3$}
    node [sum, fill=blue!30,right = 0.7cm of blk4] (blk5) {$W_4$}
    node [sum, fill=green!30,right = 0.7 cm of blk5] (blk6) {$W_5$}

    node [block, fill=blue!30,below = 0.6 cm of blk1] (blk11) {$\bfA_{0}^T \, \bfx$}
    node [block, fill=blue!30,below = 0.0005 cm of blk11] (blk12) {$\bfA_{1}^T\, \bfx$}
    node [block, fill=blue!30,below = 0.0005 cm of blk12] (blk13) {$\bfC_{0}^T\, \bfx$}

    node [block, fill=green!30,below = 0.6 cm of blk2] (blk21) {$\bfA_{2}^T\, \bfx$}
    node [block, fill=green!30,below = 0.0005 cm of blk21] (blk22) {$\bfA_{3}^T\, \bfx$}
    node [block, fill=green!30,below = 0.0005 cm of blk22] (blk23) {$\bfC_{1}^T\, \bfx$}

    node [block, fill=blue!30,below = 0.6 cm of blk3] (blk31) {$\bfA_{4}^T\, \bfx$}
    node [block, fill=blue!30,below = 0.0005 cm of blk31] (blk32) {$\bfA_{5}^T\, \bfx$}
    node [block, fill=blue!30,below = 0.0005 cm of blk32] (blk33) {$\bfC_{2}^T\, \bfx$}

    node [block, fill=green!30,below = 0.6 cm of blk4] (blk41) {$\bfA_{6}^T\, \bfx$}
    node [block, fill=green!30,below = 0.0005 cm of blk41] (blk42) {$\bfA_{7}^T\, \bfx$}
    node [block, fill=green!30,below = 0.0005 cm of blk42] (blk43) {$\bfC_{3}^T\, \bfx$}

    node [block, fill=blue!30,below = 0.6 cm of blk5] (blk51) {$\bfA_{8}^T\, \bfx$}
    node [block, fill=blue!30,below = 0.0005 cm of blk51] (blk52) {$\bfA_{9}^T\, \bfx$}
    node [block, fill=blue!30,below = 0.0005 cm of blk52] (blk53) {$\bfC_{4}^T\, \bfx$}
    
    node [block, fill=green!30,below = 0.6 cm of blk6] (blk61) {$\bfA_{10}^T\, \bfx$}
    node [block, fill=green!30,below = 0.0005 cm of blk61] (blk62) {$\bfA_{11}^T\, \bfx$}
    node [block, fill=green!30, minimum width = 1.1 cm, below = 0.0005 cm of blk62] (blk63) {$\bfC_{5}^T\, \bfx$}
    ; 
\draw[->](blk1) -- node{} (blk11);
\draw[->](blk2) -- node{} (blk21);
\draw[->](blk3) -- node{} (blk31);
\draw[->](blk4) -- node{} (blk41);
\draw[->](blk5) -- node{} (blk51);
\draw[->](blk6) -- node{} (blk61);

%\draw[thick,dotted] ($(blk11.north west)+(-0.2,0.2)$)  rectangle ($(blk13.south east)+(0.2,-0.2)$);
%\draw[thick,dotted] ($(blk21.north west)+(-0.2,0.2)$)  rectangle ($(blk23.south east)+(0.2,-0.2)$);
%\draw[thick,dotted] ($(blk31.north west)+(-0.2,0.2)$)  rectangle ($(blk32.south east)+(0.2,-0.2)$);
%\draw[thick,dotted] ($(blk41.north west)+(-0.2,0.2)$)  rectangle ($(blk41.south east)+(0.2,-0.2)$);

\end{tikzpicture}
}
\caption{\small Partitioning matrix $A$ into $\Delta = 12$ submatrices and assigning to $n = 6$ workers each of which has been assigned two uncoded and one coded task to be resilient to $s = 2$ stragglers. The coded submatrix assigned to $W_i$ is denoted as $\bfC_i$.}
\label{opt_matvec}
\end{figure} 
We consider an example in Fig. \ref{opt_matvec} with $n = 6$ and $\gamma = \frac{1}{4}$, so $k_A = 4$. We set $\Delta = \textrm{LCM} (6, 4) = 12$, and $\ell_c = \frac{12}{4} - \frac{12}{6} = 1$. Thus, we assign two uncoded jobs and one coded job to each worker where the coded job assignment would be incorporated using a random matrix $\calR$ of size $6 \times 12$. In this case, $Q = 12$, thus $Q / \Delta = 1$, and $\tau = 4$.
\end{example}

\begin{remark}
On the surface Fig. \ref{opt_matvec} may appear equivalent to a systematic version of the RKRP coded scheme \cite{8919859} with the same number of matrix partitions. However, there is a significant difference that the idea in the RKRP coded scheme is to assign the systematic versions to some workers and the coded versions to other workers, whereas we assign the jobs in a symmetric fashion so that every worker receives same number of uncoded and same number of coded jobs. If the input matrices are sparse, then the parity workers in the RKRP coded scheme will be significantly slower than the systematic workers.
\end{remark}

\subsection{Matrix-matrix Multiplication} 

\begin{algorithm}[t]
	\caption{SCS Optimal scheme for distributed matrix-matrix multiplication}
	\label{Alg:Optimal_Matmat}
   \SetKwInOut{Input}{Input}
   \SetKwInOut{Output}{Output}
   \Input{Matrices $\bfA$ and $\bfB$, $n$-number of worker nodes, storage fraction $\gamma_A = \frac{1}{k_A}$ and $\gamma_B = \frac{1}{k_B}$. So, $s = n - k_A k_B$.}
   Set $\Delta_A = \textrm{LCM}(n, k_A)$ and $\Delta_B = k_B$\; 
   Partition $\bfA$ and $\bfB$ into $\Delta_A$ and $\Delta_B$ block-columns, and $\Delta = \Delta_A \Delta_B$\;
   Number of coded submatrices of $\bfA$ in each worker node, $\ell_c = \frac{\Delta_A}{k_A} - \frac{\Delta}{n}$\;
   \For{$i\gets 0$ \KwTo $n-1$}{
   $u \gets i \times \frac{\Delta_A}{n}$\; 
   Define $T = \left\lbrace u, u+1, \dots, u + \frac{\Delta}{n} - 1 \right\rbrace$ (modulo $\Delta_A$)\;
   Assign all $\bfA_{m}$'s sequentially from top to bottom to worker node $i$, where $m \in T$\;
   Assign $\ell_c$ different random linear combinations of $\bfA_m$'s for $m \notin T$\;
   Assign a single random linear combination of all block-columns of $\bfB$\;
   
   }
   \Output{$\langle n, \gamma_A, \gamma_B \rangle$ SCS optimal-scheme for distributed matrix-matrix multiplication.}
\end{algorithm}
%\aditya{Do we really need to state this as a Lemma. We can just include it inline an save space.}
We propose a matrix-matrix multiplication scheme in Alg. \ref{Alg:Optimal_Matmat} with storage fractions $\gamma_A = 1/k_A$ and $\gamma_B = 1/k_B$ and recovery threshold $k_A k_B$. Furthermore, $Q/\Delta = 1 + (k_B-1) \ell_c/\Delta$, where $\ell_c$ is the number of coded-coded matrix-matrix products assigned to each worker node.

\begin{theorem}
\label{thm:opt_straggler_resilience}
Alg. \ref{Alg:Optimal_Matmat} proposes a distributed matrix-matrix multiplication scheme being resilient to $s = n - k_A k_B$ stragglers.
\end{theorem}
\begin{proof}
%Proof for straggler resilience
According to this scheme, we know that every worker is assigned $\frac{\Delta_A}{k_A}$ block-columns (uncoded and coded) from $\bfA$ and one coded block-column from $\bfB$, which indicates that we can obtain, in total, $\frac{\Delta_A}{k_A}$ products from each of the workers. Thus from any $k_A k_B$ workers, the master node can obtain $\frac{\Delta_A}{k_A} \times k_A k_B = \Delta_A k_B = \Delta_A \Delta_B = \Delta$ products. A simple counting argument applied to Alg. \ref{Alg:Optimal_Matmat} shows that any uncoded block-column of $\bfA$ appears exactly $k_B$ times over all $n$ workers.

In what follows we show that each of these block products corresponds to a linearly independent equation where the variables are $\bfA_i^T \bfB_j$ for $i=0, 1, \dots, \Delta_A-1, j = 0, 1, \dots, \Delta_B-1$. Let $e_i$ denote the $i$-th unit vector of length $\Delta_A$, $i = 0, \dots, \Delta_A -1$. It follows that the product $(\sum_{i=0}^{\Delta_A-1} u_i \bfA_i)^T (\sum_{j=0}^{\Delta_B-1} v_j \bfB_j)$ corresponds to the vector $\sum_{i=0}^{\Delta_A-1} u_i (e_i \otimes v$), where $v$ is the vector $[v_0 ~ v_1 ~ \dots ~v_{\Delta_B-1}]^T$ ({\it cf.} discussion around \eqref{eq:kron_prod_matmat}).% and the symbol $\otimes$ denotes the Kronecker product.

Now, suppose that we consider a subset of $k = k_A k_B$ workers indexed by the set $\calI = \{i_0, i_1, \dots, i_{k-1}\}$. Within this worker node set, let $\calJ_i$ denote the index set of the worker nodes where $\bfA_i$ appears uncoded. The random encoding vectors for $\bfA$ and $\bfB$ in worker $W_\ell$ are denoted by $u^{(\ell,j)}$ (of length $\Delta_A$) for $j = 0, 1, \dots, \ell_c-1$ and $v^{(\ell)}$ (of length $\Delta_B$) respectively.

It follows that the products involving the uncoded block-column $\bfA_i$ can be expressed as
\begin{align*}
e_i \otimes v^{(\ell)} \text{~for~} \ell \in \calJ_i.
\end{align*}
%where $v^{(\ell)} = [v_0^\ell~ v_1^\ell~ \dots ~ v_{\Delta_B -1}^\ell]^T$ represents the random linear vector chosen by worker $W_\ell$ for encoding $\bfB$.

Our first observation is that the collection of vectors $\{e_i \otimes v^{(\ell)} \}$ for $\ell \in \calJ_i, i = 0, \dots, \Delta_A -1$ is linearly independent. This follows because any linear combination of these vectors can equivalently be expressed as
\begin{align*}
    \sum_{i = 0}^{\Delta_A -1} \left( e_i \otimes \sum_{\ell \in \calJ_i} \alpha^{(i)}_{\ell} v^{(\ell)} \right)
\end{align*}
where $\alpha^{(i)}_{\ell}$'s are the linear combination coefficients and each term in the above sum needs to be forced to zero.
Note that $|\calJ_i| \leq k_B$. Therefore, the vectors $v^{(\ell)}$ for $\ell \in \calJ_i$ are linearly independent with probability 1, since $v^{(\ell)}$ has length $\Delta_B = k_B$. Thus, there is no setting of $\alpha^{(i)}_{\ell}$'s for which the above sum can be forced to the zero vector.

%Let $u^{(\ell,i)}$ for $i = 0, \dots, \ell_c -1$ represent the $i$-th random linear vector chosen by worker $W_\ell$ for encoding $\bfA$. 

The product of the coded $\bfA$ and $\bfB$ matrices can be represented by $u^{(\ell,j)} \otimes v^{(\ell)}$ for $j = 0, 1, \dots, \ell_c -1$ and $\ell \in \calI$. We will now show that the overall collection of vectors that we obtain is linearly independent with probability 1. To see this suppose that there exist coefficients $\alpha^{(i)}_{\ell}$'s and $\kappa^{(j)}_\ell$'s not all zero such that
\begin{align*}
    \sum_{i = 0}^{\Delta_A -1} e_i \otimes  \sum_{\ell \in \calJ_i} \alpha^{(i)}_{\ell} v^{(\ell)} 
    = \sum_{\ell \in \calI} \sum_{j=0}^{\ell_c -1} \kappa^{(j)}_\ell u^{(\ell,j)} \otimes v^{(\ell)} \nonumber = \sum_{\ell \in \calI} \sum_{j=0}^{\ell_c -1} \kappa^{(j)}_\ell \sum_{j_1=0}^{\Delta_A-1} u^{(\ell,j)}_{j_1} e_{j_1} \otimes v^{(\ell)}.
    %&= \sum_{\ell \in \calI} \sum_{i=0}^{\ell_c -1} \kappa^{(i)}_\ell \sum_{j_1 =0}^{\Delta_A -1} \sum_{j_2 = 0}^{\Delta_B -1} u^{(\ell,i)}_{j_1} v^\ell_{j_2} e_{j_1} \otimes f_{j_2} \label{eq:opt_mat_1}
\end{align*} 
%for some $\alpha^{(i)}_\ell$'s and $\kappa^{(j)}_\ell$'s. 
It can be observed that this decouples into finding solutions for 
\begin{align}
    e_i \otimes \sum_{\ell \in \calJ_i} \alpha^{(i)}_{\ell} v^{(\ell)} = e_i \otimes \sum_{\ell \in \calI} \sum_{j=0}^{\ell_c -1} \kappa^{(j)}_\ell  u^{(\ell,j)}_{i} v^{(\ell)} \label{eq:opt_mat_2}
\end{align}
where the $\alpha^{(i)}_\ell$ values on the LHS can be chosen freely given the RHS. For a given choice of the $\kappa^{(j)}_\ell$'s the above equation can definitely be satisfied if $|\calJ_i|=k_B$. If we $|\calJ_i| < k_B$ then this may not be true depending on the values of the $\kappa^{(j)}_\ell$'s.

The $n-k_A k_B$ stragglers together contain $(n-k_A k_B)\Delta_A k_B/n$ uncoded block-columns of $\bfA$. It is not too hard to see that not all $\bfA_i$'s that appear within the stragglers appear $k_B$ times within the stragglers (see Appendix \ref{aux_argument}). Thus, the number of $\bfA_i$'s with $|\calJ_i| < k_B$ is $\geq (n-k_A k_B)\Delta_A/n + 1$.

In the argument below we only consider the $\alpha^{(i)}_\ell$'s corresponding to these uncoded block-columns and suppose that there is an assignment of $\alpha^{(i)}_\ell$'s that satisfy (\ref{eq:opt_mat_2}). In this case the problem of finding the corresponding $\kappa^{(j)}_\ell$'s is equivalent to solving a block system of equations described below.

Let $\bfA_\delta$ be an uncoded block-column that appears less than $k_B$ times in $\calI$. The block row corresponding to it ({\it cf.} (\ref{eq:opt_mat_2})) is given by $\tilde{\bfV} \odot \tilde{\bfU}$ where
\begin{align*}
\tilde{\bfV}  =    &\begin{bmatrix}
    \bovermat{$\ell_c$}{v^{(i_0)}~ \dots~ v^{(i_0)}} | ~\dots~| \bovermat{$\ell_c$}{v^{(i_{k-1})}~ \dots~ v^{(i_{k-1})}}
    \end{bmatrix}, \;\; \textrm{and} \\ 
    \tilde{\bfU} = & [u^{(i_0,0)}_\delta~\dots~ u^{(i_0,\ell_c-1)}_\delta| ~\dots~| u^{(i_{k-1},0)}_\delta~\dots~ u^{(i_{k-1},\ell_c-1)}_\delta] \label{eq:khatri-rao}
\end{align*}
%\aditya{put overbrace in the equation that indicates that a given v vector is repeated lc times}
where $\odot$ represents the Khatri-Rao product that corresponds to column-wise Kronecker products.

Appendix \ref{app:argumentKR} shows that the concatenation of block rows in $\tilde{\bfV} \odot \tilde{\bfU}$ corresponding to the different $\bfA_\delta$'s is such that any $\ell_c k_A k_B \times \ell_c k_A k_B$ matrix is full rank with probability-1. This implies that from the first $\ell_c k_A$ block rows we can decode all the $\kappa^{(i)}_\ell$'s.

On the other hand the equations in (\ref{eq:opt_mat_2}) need to be satisfied for at least $(n-k_A k_B)\Delta_A/n + 1$ different $\bfA_i$'s based on the argument above. However
\begin{align*}
(n-k_A k_B)\frac{\Delta_A}{n} = \Delta_A - \frac{k_A k_B \Delta_A}{n} = k_A \left( \frac{\Delta_A}{k_A} - \frac{\Delta}{n} \right) = \ell_c k_A \; \; \; \; \textrm{and thus,} \; \;     (n-k_A k_B)\Delta_A/n + 1 > \ell_c k_A;
\end{align*}
\begin{comment}
However, 
\begin{align*}
    & (n-k_A k_B)\Delta_A/n + 1 > \ell_c k_A \\
    \iff & (n-k_A k_B)\Delta_A k_B/n + k_B > \ell_c k_A k_B \\
    \iff & \Delta_A k_B + k_B > k_A k_B (\Delta_A k_B/n + \ell_c) \\
    \iff & k_B > 0.
\end{align*}
The last equivalence above follows since $\Delta_A k_B/n + \ell_c = \Delta_A/k_A$.
\end{comment} 
This implies that there is at least one equation that need to be satisfied with a fixed choice of the $\kappa^{(i)}_\ell$'s. But this probability is zero since each of the remaining equations involve random $u^{(\ell,i)}_\delta$ values that have not appeared in the first $\ell_c k_A$ block rows.
\end{proof}

\begin{theorem}
\label{thm:optimalQ}
Alg. \ref{Alg:Optimal_Matmat} proposes a distributed matrix-matrix multiplication scheme with $Q = \Delta + (k_B - 1) \ell_c $.
\end{theorem}
\begin{proof}
As in the proof of the previous result, we let $u^{(\ell,j)}$ for $j = 0, \dots, \ell_c -1$ denote the $j$-th random encoding vector for $\bfA$ in worker $W_\ell$ and $v^{(\ell)}$ the corresponding random encoding vector for $\bfB$. We will demonstrate that the system of equations that corresponding to decoding the $\bfA_i^T \bfB_j$'s is nonsingular with probability 1. Let $e_i$ denote the $i$-th unit vector of length $\Delta_A$. For a given $\bfA_i$, suppose that it appears uncoded in $\calJ_i$ worker nodes where $|\calJ_i| \leq k_B$ we obtain certain equations from the uncoded part which correspond to $e_i \otimes v^{(\ell)}$ for $\ell \in \calJ_i$. If $|\calJ_i| < k_B$ then it needs to use the coded-coded products for decoding the unknowns corresponding to $\bfA_i$.

The block system of equations under consideration corresponds to a $\Delta_A k_B \times \Delta_A k_B$ square matrix with random entries. For $\bfA_i$ such that $|\calJ_i| = k_B$ the matrix consists of a $k_B \times k_B$ block on the diagonal with $k_B$ distinct vectors $v^{(\ell)}$. This block is nonsingular with probability-1 owing to the random choice of the $v^{(\ell)}$'s.

For the other $\bfA_i$'s where $|\calJ_i| < k_B$ we will demonstrate a setting of the  $u^{(\ell,j)}$'s such that the entire matrix is a block diagonal matrix with $k_B \times k_B$ blocks of distinct $v^{(\ell)}$ vectors. This demonstrates that there exists a choice of random coefficients for which the system of equations is nonsingular. Following this the result holds with probability-1 when the choice is made at random.

Towards this end, suppose that the pattern of obtained products is such that we get $\Delta- \lambda$ uncoded-coded  products and $\lambda + (k_B-1)\ell_c$ coded-coded products. %Suppose that we need to decode block products $\bfA_i^T \bfB_j$ for $i \in \calI$ using the coded block products. 
Without loss of generality we assume that we need to decode the products that involve $\bfA_0, \bfA_1, \dots, \bfA_{\delta-1}$ using the coded-coded products.
%\anindya{Need to define $\calI$} W. l. o. g. we suppose that these are $\bfA_0, \bfA_1, \dots, \bfA_{\delta-1}$. 
Furthermore we suppose that $\bfA_i$ appears $k_B-\eta_i$ times within the uncoded-coded products, so that $\eta_0 + \eta_1 + \dots + \eta_{\delta-1} = \lambda$. 

Under this setting, there are at least $(k_B-1)\ell_c + \lambda - (k_B-\eta_0) \ell_c = (\eta_0 -1)\ell_c + \lambda$ coded-coded products that can be obtained from worker nodes that do not contain an uncoded copy of $\bfA_0$. Furthermore, these are spread out in at least $\eta_0$ distinct worker nodes. Next, we pick $\eta_0$ encoding vectors for $\bfA$ from the $\eta_0$ distinct workers and set them all to $e_0$. With this setting we obtain a $k_B \times k_B$ block (corresponding to decoding $\bfA_0^T \bfB_j, j = 0, \dots, \Delta_B-1$) that consists of distinct $v^{(\ell)}$ vectors that are nonsingular with probability $1$.

At this point we are left with $(k_B-1)\ell_c + \lambda - \eta_0$ coded-coded products. The argument can be repeated for $\bfA_1$ since there are at least $(\eta_1 -1)\ell_c + \lambda - \eta_0$ coded-coded products that can be obtained from workers where $\bfA_1$ does not appear, which in turn correspond to at least $\eta_1$ distinct workers. In this case we will set the $\eta_1$ encoding vectors to $e_1$. The process can be continued in this way until the coded-coded products are assigned to each of $\bfA_0, \bfA_1, \dots, \bfA_{\delta-1}$.

At the end of the process we can claim that we have a block diagonal matrix where each block is a $k_B \times k_B$ square matrix with distinct $v^{(\ell)}$ vectors. Thus each block and consequently the entire system of equations is nonsingular. 

Finally, as there exists a choice of random values that makes the system of equations nonsingular, it continues to be nonsingular with probability $1$ under a random choice.

\end{proof}

To summarize, Theorems \ref{thm:opt_straggler_resilience} and \ref{thm:optimalQ} demonstrate that our proposed scheme has the optimal threshold $k_A k_B$ and 
\begin{align*}
\frac{Q}{\Delta} &= 1 + \frac{(k_B - 1)\ell_c}{\Delta} = 1 + \frac{(k_B - 1)\left( \frac{\Delta_A}{k_A} - \frac{\Delta}{n}\right)}{\Delta} \\
&=  1 + \frac{\Delta (k_B - 1)\left( \frac{1}{k_A k_B} - \frac{1}{n}\right)}{\Delta} = 1 + \frac{(k_B - 1) s}{n k_A k_B} \approx 1 + \frac{s}{n k_A};
\end{align*} if $k_B$ is significantly larger than $1$. Moreover in the practical cases, we usually have $s << n k_A$, thus in this SCS optimal scheme, we have $Q/\Delta \approx 1$.

%\aditya{CHECK}. \aditya{Somewhere need to compare to Random Khatri-Rao approach} \anindya{Could you please clarify ? We have comparison in numerical experiment section.}

\begin{example}
\begin{figure*}[t]
\centering
\captionsetup{justification=centering}
\resizebox{0.85\linewidth}{!}{
\begin{tikzpicture}[auto, thick, node distance=2cm, >=triangle 45]
\draw
    node [sum1, minimum size = 1.2cm, fill=green!30] (blk1) {$W_0$}
    node [sum1, minimum size = 1.2cm, fill=blue!30,right = 1.7 cm of blk1] (blk2) {$W_1$}
    node [sum1, minimum size = 1.2cm, fill=green!30,right = 1.7 cm of blk2] (blk3) {$W_2$}
    node [sum1, minimum size = 1.2cm, fill=blue!30,right = 1.7 cm of blk3] (blk4) {$W_3$}
    node [sum1, minimum size = 1.2cm, fill=green!30,right = 1.7 cm of blk4] (blk5) {$W_4$}

    node [block, minimum width = 2.1cm, fill=green!30,below = 0.4 cm of blk1] (blk11) {$\bfA_0$}
    node [block, minimum width = 2.1cm, fill=green!30,below = 0.0005 cm of blk11] (blk12) {$\bfA_1$}
    node [block, minimum width = 2.1cm, fill=green!30,below = 0.0005 cm of blk12] (blk13) {$\bfA_2$}
    node [block, minimum width = 2.1cm, fill=green!30,below = 0.0005 cm of blk13] (blk14) {$\bfA_3$}
    node [block,  minimum width = 2.1cm, fill=green!30,below = 0.0005 cm of blk14] (blk19) {$\sum\limits_{i=0}^9 \bfR_A^{(0,i)} \bfA_i$}
    
    node [block,  minimum width = 2.1cm, fill=orange!50,below = 0.2 cm of blk19] (blk15) {$\sum\limits_{i=0}^1 \bfR_B^{(0,i)} \bfB_i$}

    node [block, minimum width = 2.1cm, fill=blue!30,below = 0.4 cm of blk2] (blk21) {$\bfA_2$}
    node [block, minimum width = 2.1cm, fill=blue!30,below = 0.0005 cm of blk21] (blk22) {$\bfA_3$}
    node [block, minimum width = 2.1cm, fill=blue!30,below = 0.0005 cm of blk22] (blk23) {$\bfA_4$}
    node [block, minimum width = 2.1cm, fill=blue!30,below = 0.0005 cm of blk23] (blk24) {$\bfA_5$}
    node [block,  minimum width = 2.1cm, fill=blue!30,below = 0.0005 cm of blk24] (blk29) {$\sum\limits_{i=0}^9 \bfR_A^{(1,i)} \bfA_i$}
    node [block,  minimum width = 2.1cm, fill=mycolor2!30,below = 0.2 cm of blk29] (blk25) {$\sum\limits_{i=0}^1 \bfR_B^{(1,i)} \bfB_i$}

    node [block, minimum width = 2.1cm, fill=green!30,below = 0.4 cm of blk3] (blk31) {$\bfA_4$}
    node [block, minimum width = 2.1cm, fill=green!30,below = 0.0005 cm of blk31] (blk32) {$\bfA_5$}
    node [block, minimum width = 2.1cm, fill=green!30,below = 0.0005 cm of blk32] (blk33) {$\bfA_6$}
    node [block, minimum width = 2.1cm, fill=green!30,below = 0.0005 cm of blk33] (blk34) {$\bfA_7$}
    node [block, minimum width = 2.1cm, fill=green!30,below = 0.0005 cm of blk34] (blk39) {$\sum\limits_{i=0}^9 \bfR_A^{(2,i)} \bfA_i$}
    node [block,  minimum width = 2.1cm,fill=orange!50,below = 0.2 cm of blk39] (blk35) {$\sum\limits_{i=0}^1 \bfR_B^{(2,i)} \bfB_i$}

    node [block, minimum width = 2.1cm, fill=blue!30,below = 0.4 cm of blk4] (blk41) {$\bfA_6$}
    node [block, minimum width = 2.1cm, fill=blue!30,below = 0.0005 cm of blk41] (blk42) {$\bfA_7$}
    node [block, minimum width = 2.1cm, fill=blue!30,below = 0.0005 cm of blk42] (blk43) {$\bfA_8$}
    node [block, minimum width = 2.1cm, fill=blue!30,below = 0.0005 cm of blk43] (blk44) {$\bfA_9$}
    node [block, minimum width = 2.1cm, fill=blue!30,below = 0.0005 cm of blk44] (blk49) {$\sum\limits_{i=0}^9 \bfR_A^{(3,i)} \bfA_i$}
    node [block,  minimum width = 2.1cm, fill=mycolor2!30,below = 0.2 cm of blk49] (blk45) {$\sum\limits_{i=0}^1 \bfR_B^{(3,i)} \bfB_i$}

    node [block, minimum width = 2.1cm, fill=green!30,below = 0.4 cm of blk5] (blk51) {$\bfA_8$}
    node [block, minimum width = 2.1cm, fill=green!30,below = 0.0005 cm of blk51] (blk52) {$\bfA_9$}
    node [block, minimum width = 2.1cm, fill=green!30,below = 0.0005 cm of blk52] (blk53) {$\bfA_0$}
    node [block, minimum width = 2.1cm, fill=green!30,below = 0.0005 cm of blk53] (blk54) {$\bfA_1$}
    node [block,  minimum width = 2.1cm,fill=green!30,below = 0.0005 cm of blk54] (blk59) {$\sum\limits_{i=0}^9 \bfR_A^{(4,i)} \bfA_i$}  
    node [block,  minimum width = 2.1cm,fill=orange!30,below = 0.2 cm of blk59] (blk55) {$\sum\limits_{i=0}^1 \bfR_B^{(4,i)} \bfB_i$}

    ;

\draw[->](blk1) -- node{} (blk11);
\draw[->](blk2) -- node{} (blk21);
\draw[->](blk3) -- node{} (blk31);
\draw[->](blk4) -- node{} (blk41);
\draw[->](blk5) -- node{} (blk51);

\end{tikzpicture}
}
\centering
\caption{\small Matrix-matrix multiplication with $n = 5$ and $s = 1$ with $\gamma_A = \gamma_B = \frac{1}{2}$. Here $\bfR_A$ and $\bfR_B$ are random matrices whose superscripts indicate their corresponding rows and columns.}

\label{opt_matmat}
\end{figure*}
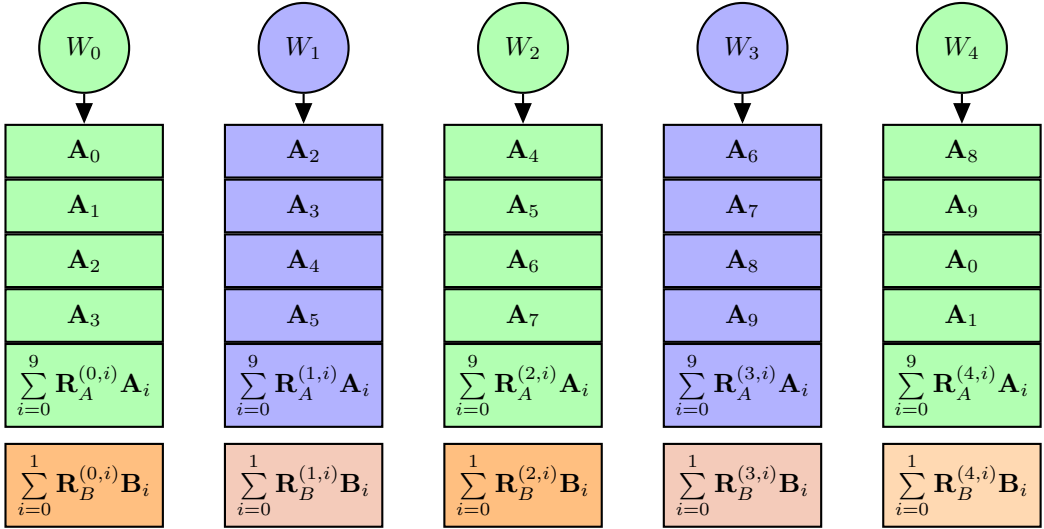 
We consider an example in Fig. \ref{opt_matmat} with $n = 5$ and $k_A = k_B = 2$, so the system is resilient $s = 5 - 4 = 1$ straggler. We set $\Delta_A = \textrm{LCM} (n, k_A) = 10$ and $\Delta_B = k_B = 2$, and in this example, $Q = 21$, thus $Q/\Delta = 1.05$.
\end{example}

\begin{comment}
\section{Numerical Experiments and Comparisons}
\label{sec:numerical_exp}
\begin{itemize}
\item Comparison of required time among uncoded, mixed and fully coded approaches.
\item Comparison of required time between our optimal and any other approach.
\item Any comments on decoding time.
\end{itemize}
\end{comment}

\section{Numerical Experiments and Comparisons}
\label{sec:numerical_exp}
In this section, we discuss the results of the numerical experiments for our proposed approaches and compare them with other available methods. First we compare all the approaches in terms of number of stragglers that a scheme can be resilient to, and in terms of $Q$ values. Next we compare the approaches in terms of the worker computation time and numerical stability during the decoding process. Software code for recreating these experiments can be found at \cite{anindyacode2}.

\subsection{Number of stragglers and Q value}
%In this experiment, we compare different approaches for distributed matrix computations in terms of number of stragglers and $Q$ values. First, 
Table \ref{tabmatvec} shows the comparison for matrix-vector multiplication for $n = 30$ workers, each of which can store $\gamma_A = \frac{1}{10}$ fraction of matrix $\bfA$. For the convolutional code approach, we assume $s = 15$ so that $n - s = 15 > \frac{1}{\gamma_A} = 10$ which satisfies the required condition in \cite{das2019random}. And for the coded at bottom approach, we assume $\gamma_u = \frac{1}{15}$ and $\gamma_c = \frac{1}{30}$, so that $\gamma = \gamma_u + \gamma_c$. Similarly, Table \ref{tabmatmat} shows the comparison for different approaches for matrix-matrix multiplication for $n = 18$ workers, each of which can store $\gamma_A = \frac{1}{3}$ and $\gamma_B = \frac{1}{3}$ fraction of matrices $\bfA$ and $\bfB$ respectively. Here we assume $k_A = k_B = 4 > \frac{1}{\gamma_A} = \frac{1}{\gamma_B} = 3$ for the approach in \cite{das2019random}. % and we assume $\gamma_{Au} = \frac{1}{6}$ and $\gamma_{Ac} = \frac{1}{6}$ for the coded at the bottom approach.

\begin{table*}[t]
\captionsetup{justification=centering}
\caption{{\small Comparison of number of stragglers, $Q$ values, worker computation time (in $ms$) and worst case condition number$(\kappa_{worst})$ for matrix-vector multiplication for  $n = 30$ and $\gamma_A = \frac{1}{10}$ (*for convolutional code, we assumed $s = 15$).}}
\label{tabmatvec}
\begin{center}
\begin{small}
\begin{sc}
\begin{tabular}{c c c c c c}
\hline
\toprule
\multirow{2}{*}{Methods} & \multirow{2}{*}{Stragglers} & $\; \; \;$\multirow{2}{*}{$\frac{Q}{\Delta}$ value}$\; \; \;$ & \multicolumn{2}{c}{$\; \; \; $ Worker Computation Time $\; \; \;$} & $\; \; \;$\multirow{2}{*}{$\kappa_{worst}$}$\; \; \;$ \\ \cline{4-5}
 & & & Sparsity $98\%$ & Sparsity $95\%$ &  \\
 \midrule
Polynomial Code  \cite{yu2017polynomial} & $20$ & $N/A$ &  $62.8$ & $87.1$  & $5.99 \times 10^9$\\
Ortho-Poly Code \cite{8849468} & $20$ & $N/A$ & $62.3$ & $86.4$  & $4.34 \times 10^{11}$\\
RKRP Code\cite{8919859} & $20$ & $N/A$ &  $62.9$ & $86.8$  & $ 5.44 \times 10^8$\\
Convolutional Code* \cite{das2019random} & $15$ & $N/A$ &   $63.1$ & $87.7$  & $6.24 \times 10^4$\\
Uncoded \cite{c3les}   & $2$ & $85/30$ &  $19.1$ & $32.9$  & $1.7321$\\
Uncoded (proposed)   & $2$ & $84/30$ &  $19.2$ & $33.1$  & $1.7321$\\
$\beta$-level coding ($\beta = 2$) & $4$ & $81/30$ &  $25.3$ & $40.2$  & $242.89$\\
$\beta$-level coding ($\beta = 3$) & $6$ & $79/30$ &  $29.2$ & $47.9$  & $1.53 \times 10^3$\\
Coded at Bottom   & $15$ & $58/30$ & $24.1$ & $37.8$ & $1.41 \times 10^3$\\
%Optimal Scheme   & $20$ & $1$ & $35.6$ & $54.2$ & $1.36 \times 10^9$\\
\bottomrule
\end{tabular}
\end{sc}
\end{small}
\end{center}
\end{table*}%

\begin{table*}[t]
\captionsetup{justification=centering}
\caption{{\small Comparison of number of stragglers, $Q$ values, worker computation time (in seconds) and worst case condition number $(\kappa_{worst})$ for matrix-matrix multiplication for $n = 18$ and $\gamma_A = \gamma_B = \frac{1}{3}$ (*for convolutional code, $k_A = k_B = 4$).}}
\label{tabmatmat}
\begin{center}
\begin{small}
\begin{sc}
\begin{tabular}{c c c c c c}
\hline
\toprule
\multirow{2}{*}{Methods} & \multirow{2}{*}{ Stragglers} & \multirow{2}{*}{$\frac{Q}{\Delta}$ value} & \multicolumn{2}{c}{Worker Computation Time} & \multirow{2}{*}{$\; \; \; \kappa_{worst} \; \; \; $} \\ \cline{4-5}
 & & & Sparsity $98\%$ & Sparsity $95\%$ &   \\
 \midrule
Polynomial Code  \cite{yu2017polynomial} & $9$ & $N/A$ & $2.58$ & $10.16$ & $7.33 \times 10^{6}$\\
Ortho-Poly Code \cite{8849468} & $9$ & $N/A$   & $2.51$ & $10.08$ & $1.33 \times 10^{7}$\\
RKRP Code\cite{8919859} & $9$ & $N/A$   & $2.63$ & $10.23$ & $2.15 \times 10^5$\\
Convolutional Code* \cite{das2019random} & $2$ & $N/A$ & $2.44$  & $10.19$  & $1.82 \times 10^3$\\
%Uncoded (setting $\Delta = n$)   & $3$ & $135/36$ & $0.69$ & $1.96$ & $2.00$\\
Uncoded (proposed)   & $1$ & $17/9$ & $0.69$ & $1.96$ & $1.41$\\
$\beta$-level coding ($\beta_A = \beta_B = 2$) & $4$ & $16/9$ & $1.02$  & $3.68$ & $8.89 \times 10^3$\\
%Coded at Bottom   & $8$ & $117/36$ & $1.25$ & $4.88$ & $4.39 \times 10^3$\\
%Optimal Scheme   & $27$ & $7/6$ & $10.13$ & $10.25$ & $2.65 \times 10^{11}$\\
\bottomrule
\end{tabular}
\end{sc}
\end{small}
\end{center}
\end{table*}%
In case of both matrix-vector and matrix-matrix multiplications, we know that the dense coded approaches \cite{yu2017polynomial}, \cite{8919859}, \cite{das2019random} and \cite{8849468} are MDS but they do not consider the partial computations of the slower workers. On the other hand, our proposed approaches are able to utilize the partial computations of the stragglers for both matrix-vector and matrix-matrix multiplications. We can see that the $\beta$-level coding approaches, with $\beta = 2$ or $3$, have smaller $Q/\Delta$ values than the uncoded approaches, one of which is introduced in \cite{c3les} and the other is a special case of our proposed $\beta$-level coding where $\beta = 1$. We emphasize that a larger value of $\beta$ or a larger value of $\gamma_c$ will provide smaller values of $Q/\Delta$ for our proposed $\beta$-level coding approach and the coded-at the bottom scheme, respectively. It should be noted that the approach in \cite{das2019random} requires the condition $n - s > \frac{1}{\gamma}$ to be full-filled to be resilient to $s$ stragglers, so as mentioned in Tables \ref{tabmatvec} and \ref{tabmatmat}, this convolutional code-based approach is resilient to less number of stragglers than the other dense coded approaches.
%Furthermore from both Tables \ref{tabmatvec} and \ref{tabmatmat} we can see that our proposed optimal scheme provides the best $Q/\Delta$ among all the available methods.

\subsection{Worker Computation Time}
%Now we choose a real time example and 
We compare the computation time required by the workers in case of different approaches by experiments performed on an Amazon Web Services (AWS) cluster where we choose a {\tt t2.2xlarge} machine as the master node and {\tt t2.small} machines as the worker nodes, which are, in fact, responsible for computing the submatrix products. 

For matrix-vector multiplication, We choose a matrix $\bfA$ of size $40,000 \times 17,640$ and a vector $\bfx$ of length $40,000$, and the job is to compute $\bfA^T \bfx$ in a distributed fashion. We assume that the matrix $\bfA$ is sparse, which indicates that the most of the entries of $\bfA$ are $zero$. For example, the sparsity of $\bfA$ can be $98\%$ (or $95\%$), which indicates that randomly chosen $2\%$ (or $5\%$) entries of matrix $\bfA$ are non-zero. We consider the same scenario where we have $n = 30$ workers, each of which can store $\gamma_A = \frac{1}{10}$ fraction of matrix $\bfA$. The comparison among different approaches for different sparsity values is shown in Table \ref{tabmatvec}. Next a similar experiment is carried out for matrix-matrix multiplication where both $\bfA$ and $\bfB$ are sparse and of sizes $12000 \times 13680$ and $12000 \times 10260$, respectively, and the corresponding results are shown in Table \ref{tabmatmat}.

From the experimental results shown in Tables \ref{tabmatvec} and \ref{tabmatmat}, we can see that the workers require much more time to complete their assigned jobs in case of the dense coded approaches (\cite{yu2017polynomial}, \cite{8919859}, \cite{das2019random} and \cite{8849468}) than our proposed approaches. The reason is that the dense coded approaches cannot preserve the sparsity of the matrices $\bfA$ or $\bfB$, so the corresponding coded submatrices are quite dense even if $A$ and $B$ are sparse. On the other hand, our proposed approaches can preserve the sparsity in the submatrices, and can complete the jobs $3 \sim 4$ times faster than the available approaches. It should be noted that a smaller value of $\beta$ or a smaller value of $\gamma_c$ will lead to less worker computation time for our proposed $\beta$-level coding approach and the coded at the bottom scheme, respectively. 

We note here that while there is a significant difference between the required time of the dense coded approaches and our proposed approaches, this difference can be much higher. For example, in Table \ref{tabmatvec}, we can see that the polynomial code approach is around $3 \sim 4$ times slower than the uncoded approach, but the gap according to the theoretical analysis should be as large as $10$ times, since $\gamma_A = 1/10$. The reason underlying the smaller gap is the use of two different commands in {\tt Python} to compute products between the matrix and the vector. Since the proposed uncoded or the $\beta$-level coding approaches can preserve the sparsity up to certain level, we have leveraged the sparse matrix-multiplication commands in these cases, whereas for the dense coded approaches which cannot preserve the sparsity, regular matrix-multiplication command provided better results. A more optimized sparse matrix-multiplication scheme could result in bigger multiplicative gaps between these approaches. Furthermore, the difference of the required time would be certainly higher and more significant if the matrix sizes were higher (for example, in millions). However, owing to the memory limitations of the machines that we are using (in this case, {\tt t2.small}), we cannot conduct experiments with such large matrices.

\subsection{Numerical Stability}
Now we do another experiment to compare the numerical stability of different schemes. We know that for decoding a system of equations, errors in the input can get amplified by the condition number (ratio of maximum and minimum singular values) of the associated decoding matrix; hence, a low condition number is critical \cite{das2019random,ramamoorthy2019numerically}. For example, let us consider the polynomial codes \cite{yu2017polynomial} for matrix vector multiplication, where each of $n$ workers can store $\gamma = \frac{1}{k}$ fraction of matrix $\bfA$. Now partitioning $\bfA$ into $\Delta = k$ submatrices lead to $\Delta$ unknowns, $\bfA^T_0 \bfx, \bfA^T_1 \bfx, \dots, \bfA^T_{\Delta - 1} \bfx$. Now in order to assign the coded jobs to $n$ workers, we need to choose a polynomial of degree $k-1$ and $n$ evaluation points, thus the coding matrix is of size $n \times k$. Since the recovery threshold here is $\tau = k$, we are interested in all choices of $k \times k$ submatrices of that $n \times k$ coding matrix. It can be shown that the system will be numerically more stable in the worst case if the evaluation points are chosen uniformly spaced in $[-1, 1]$, rather than choosing the integers $1, 2, \dots, n$ \cite{TangKR19}. In other words, choosing interpolation points uniformly spaced in $[-1, 1]$ will lead to a smaller worst case condition number $(\kappa_{worst})$.

In this experiment we compare the condition numbers for different approaches in case of the worst choice of full stragglers. Tables \ref{tabmatvec} and \ref{tabmatmat} show the comparison of worst case condition numbers $(\kappa_{worst})$ for matrix-vector and matrix-matrix multiplication, respectively, for the previously chosen scenario. We can see that the dense coded approaches (\cite{yu2017polynomial}, \cite{8919859} and \cite{8849468}) have a very high worst case condition number, thus suffer from numerical instability which leads to erroneous results. On the other hand, our proposed $\beta$-level coding approach has a much smaller worst case condition number. The reason is that even in the worst case, the decoding of some $\beta$ unknowns depends on a $\beta \times \beta$ system matrix whose entries are randomly chosen. Thus a smaller $\beta$ leads to a smaller $\kappa_{worst}$, for example, we can see that the uncoded case (same as the case with $\beta = 1$) is the scheme having the smallest $\kappa_{worst}$.

\subsection{Comparison with the Proposed SCS Optimal Scheme}
In this experiment, we compare the dense coded approaches with our proposed SCS optimal coding scheme in terms of $Q$ values and worker computation time. First we do the comparison for matrix-vector multiplication where we choose a square sparse matrix $\bfA$ of size $27,720 \times 27,720$, and a vector $\bfx$ of length $27,720$. The job is to compute $\bfA^T \bfx$ in a distributed system of $n = 18$ workers, each of which can store $\gamma_A = \frac{1}{15}$ fraction of matrix $\bfA$. We consider two different choices of matrix $\bfA$. In the first case, $\bfA$ is a band matrix \cite{atkinson2008introduction} where the entries are non-zero along the principal diagonal and in $1000$ other $k$-diagonals just above and below the principal diagonal. In the second case, the entries are non-zero along the principal diagonal and in $2000$ other randomly chosen $k$-diagonals. The comparison is shown in Table \ref{tabmatvec2} where we can see that the proposed SCS optimal scheme requires less time from the worker nodes in comparison to the other dense coded approaches, which in fact, cannot leverage the sparsity of matrix $\bfA$. 

Next to show an example for distributed matrix-matrix multiplication, we choose two random sparse matrices $\bfA$ and $\bfB$ of sizes $12000 \times 15000$ and $12000 \times 13500$, where randomly chosen any $2\%$ and $5\%$ entries are non-zero. We consider a distributed system having $n = 24$ workers, each of which can store $\gamma_A = \frac{1}{4}$ fraction of matrix $\bfA$ and $\gamma_B = \frac{1}{5}$ fraction of matrix $\bfB$. The comparison is shown in Table \ref{tabmatmat2} which further confirms the superiority of the proposed SCS optimal scheme in terms of workers' computation speeds. The major reason behind the enhancement of the speed in the SCS optimal scheme lies in its ability to leverage the sparsity of the matrices up to certain level, whereas the approaches in \cite{yu2017polynomial} or \cite{8849468} use the dense linear combinations of the submatrices which destroy the sparsity. The approaches in \cite{das2019random} and \cite{8919859} consider some parity worker nodes where all the assigned submatrices are dense, which leads to high worker computation time for those workers. On the other hand, in the proposed SCS optimal scheme the submatrices, obtained from dense linear combinations, are assigned uniformly within the workers. This removes the asymmetry between the worker node computation times.
%, which does not let the worker computation time go high.

\begin{table*}[t]
\captionsetup{justification=centering}
\caption{{\small Comparison of $Q$ values, worker computation time (in $ms$) and worst case condition number $(\kappa_{worst})$ for matrix-vector multiplication for $n = 18, \gamma_A = \frac{1}{15}$ (*for convolutional code, we assume $\gamma_A = \frac{1}{10}$).}}
\label{tabmatvec2}
\begin{center}
\begin{small}
\begin{sc}
\begin{tabular}{c c c c c c}
\hline
\toprule
\multirow{2}{*}{Methods} & \multirow{2}{*}{No of Stragglers}& \multirow{2}{*}{$\frac{Q}{\Delta}$ value} & \multicolumn{2}{c}{Worker Computation Time} & \multirow{2}{*}{$\; \; \; \kappa_{worst} \; \; \; $} \\ \cline{4-5}
 & & & $\;\;\;$ Band $\;\;\;$ & $\;\;\;$ Random $\;\;\;$ &   \\
 \midrule
Polynomial Code  \cite{yu2017polynomial} & $3$ &  $N/A$ & $29.7$ & $30.2$ & $4.03 \times 10^7$ \\
Ortho-Poly Code \cite{8849468}  & $3$ & $N/A$   & $30.1$ & $29.8$ & $2.13 \times 10^4$\\
RKRP Code\cite{8919859} & $3$ & $N/A$   & $29.3$ & $30.0$ & $6.35 \times 10^3$ \\
Convolutional Code* \cite{das2019random}  & $3$ & $N/A$ & $35.2$  & $34.7$  & $1.21 \times 10^3$\\
SCS Optimal Scheme   & $3$ & $1$ & $14.8$ & $20.3$ & $6.81 \times 10^4$\\
\bottomrule
\end{tabular}
\end{sc}
\end{small}
\end{center}
\end{table*}%

\begin{table*}[t]
\captionsetup{justification=centering}
\caption{{\small Comparison of $Q$ values, worker computation time (in seconds) and worst case condition number $(\kappa_{worst})$ for matrix-matrix multiplication for $n = 24, \gamma_A = \frac{1}{4}$ and $\gamma_B = \frac{1}{5}$ (*for convolutional code, we assume $\gamma_A = \frac{2}{5}$ and $\gamma_B = \frac{1}{3}$). The values in the parentheses for the SCS optimal scheme shows the time required for uncoded and coded portions, respectively.}}
\label{tabmatmat2}
\begin{center}
\begin{small}
\begin{sc}
\begin{tabular}{c c c c c c}
\hline
\toprule
\multirow{2}{*}{Methods} & \multirow{2}{*}{No of Stragglers}& \multirow{2}{*}{$\frac{Q}{\Delta}$ value} & \multicolumn{2}{c}{Worker Computation Time} & \multirow{2}{*}{$\; \; \; \kappa_{worst} \; \; \; $} \\ \cline{4-5}
 & & & Sparsity $98\%$ & Sparsity $95\%$ &   \\
 \midrule
Polynomial Code  \cite{yu2017polynomial} & $4$ &  $N/A$ & $3.11$ & $8.29$ & $2.40 \times 10^{10}$\\
Ortho-Poly Code \cite{8849468}  & $4$ & $N/A$   & $3.08$ & $8.16$ & $1.96 \times 10^{6}$\\
RKRP Code\cite{8919859} & $4$ & $N/A$   & $3.15$ & $8.22$ & $2.83 \times 10^5$\\
Convolutional Code* \cite{das2019random}  & $4$ & $N/A$ & $5.16$  & $10.92$  & $2.65 \times 10^4$\\
\multirow{2}{*}{SCS Optimal Scheme}   & \multirow{2}{*}{$4$} & \multirow{2}{*}{$7/6$} & $1.93$ & $4.76$ & \multirow{2}{*}{$4.93 \times 10^6$}\\
 & &  & $(0.91 + 1.02)$  & $(3.71+1.05)$  & \\

\bottomrule
\end{tabular}
\end{sc}
\end{small}
\end{center}
\end{table*}%

Now, similar to the most of the dense coded approaches \cite{yu2017polynomial}, \cite{8849468}, \cite{8919859}, our proposed SCS optimal scheme is also resilient to $s = n - k_A k_B $ stragglers, where $\gamma_A = \frac{1}{k_A}$ and $\gamma_B = \frac{1}{k_B}$. We point out that we did not compare with the approach in \cite{wang2018coded} since their approach does not respect the storage constraints for the matrices at each worker node and only has a high-probability guarantee on the recovery threshold. Similarly we did not compare with \cite{hasirciouglu2020bivariate} which assumes heterogeneous workers, but we note that this approach provides with a value of $Q/\Delta$ to be $11/10$ for the example shown in Table \ref{tabmatmat2}. 
Now, in the dense coded approaches, we can decode all $\Delta = k_A k_B$ unknowns from any $k_A k_B$ submatrix block products, and in that sense we have $\frac{Q}{\Delta} = 1$. But it does not necessarily mean that those scheme can utilize the partial computations done by the slower workers, since in those cases the master requires $k_A k_B$ workers to finish their jobs, and discard the computations done by others.

However, one can still use those approaches to utilize the partial computations, by partitioning the matrices into more submatrices. We can consider the an example of $n = 10$ workers with $\gamma_A = \gamma_B = 1/3$ and $98\%$ sparse matrices $\bfA$ and $\bfB$, both having size $12,000 \times 12,000$. Now we can partition matrix $\bfA$ into $\Delta_A = 3$ or $\Delta_A = 9$ submatrices for the dense coded approaches. We can see the comparison of $\kappa_{worst}$ and worker computation time in Table \ref{tab:numcond} for these two values of $\Delta_A$. In case of $\Delta_A = 9$, we will require polynomials of higher degrees (for \cite{yu2017polynomial} or \cite{8849468}) or more random coefficients (for \cite{8919859}) than in the case of $\Delta_A = 3$. It leads to a very high condition number ($\approx 10^{13}$) which will make the whole system numerically unstable. Besides, a larger $\Delta_A$ would make the submatrices even denser, which will lead to higher worker computation time for the workers. The case is similar for the work in \cite{hasirciouglu2020bivariate} which uses larger $\Delta_A$ and $\Delta_B$ to utilize the partial computations. On the other hand, in the proposed SCS optimal scheme, uncoded submatrices are placed at the top, and coded submatrices are placed at the bottom. Moreover the coded jobs are allocated uniformly among all the workers which does not let the worker computation time go high for any particular worker. %Thus our proposed optimal scheme can not only preserve the sparsity of the matrices to reduce the worker computation time but also utilize the partial computations of the slower workers while keeping the system numerically stable. 

\begin{comment}
\begin{table*}[t]
\captionsetup{justification=centering}
\caption{{\small Comparison of the Condition Numbers for Matrix-matrix Multiplication for $n = 10, \gamma_A = \gamma_B = \frac{1}{3}$, so $s = 1$.}}
\label{tab:numcond}
\begin{center}
\begin{small}
\begin{sc}
\begin{tabular}{c c c c c}
\hline
\toprule
\multirow{2}{*}{Methods} &  $\;\;\;$ \multicolumn{2}{c}{W/o Partial Computations} $\;\;\;$ & $\;\;\;$ \multicolumn{2}{c}{W/ Partial Computations} $\;\;\;$ \\ \cline{2-5}
 &  $\;\;\; \;\;\; \;\;\;\;\;\;\frac{Q}{\Delta}$ value & $\kappa_{worst} $&  $\;\;\;\;\;\;\;\;\;\;\;\;\frac{Q}{\Delta}$ value & $\kappa_{worst} $\\ 
 \midrule
Polynomial Code  \cite{yu2017polynomial}  & $\;\;\; \;\;\; \;\;\;\;\;\;$ $N/A$ & $8.8 \times 10^3$ & $\;\;\; \;\;\; \;\;\;\;\;\;$$1$ & $1.86 \times 10^{13}$\\
Ortho-Poly Code \cite{8849468}   & $\;\;\; \;\;\; \;\;\;\;\;\;$$N/A$   & $16.66$ & $\;\;\; \;\;\; \;\;\;\;\;\;$$1$ & $4.33 \times 10^{5}$\\
Random KR Code\cite{8919859} & $\;\;\; \;\;\; \;\;\;\;\;\;$$N/A$   & $11.96$ & $\;\;\; \;\;\; \;\;\;\;\;\;$$1$ & $1.16 \times 10^4$\\
Proposed Optimal Scheme    & $\;\;\; \;\;\; \;\;\;\;\;\;$$-$ & $-$ & $\;\;\; \;\;\; \;\;\;\;\;\;$$1.02$ & $2.15 \times 10^3$\\
\bottomrule
\end{tabular}
\end{sc}
\end{small}
\end{center}
\end{table*}
\end{comment}

\begin{table*}[t]
\captionsetup{justification=centering}
\caption{{\small Comparison of the $Q/\Delta$ values, worker computation time (in seconds) and worst case condition numbers for matrix-matrix multiplication for $n = 10, \gamma_A = \gamma_B = \frac{1}{3}$, so $s = 1$.}}
\label{tab:numcond}
\begin{center}
\begin{small}
\begin{sc}
\begin{tabular}{c c c c c c c c c}
\hline
\toprule
\multirow{2}{*}{Methods} & &  \multicolumn{3}{c}{W/o Partial Computations}  & &  \multicolumn{3}{c}{W/ Partial Computations} \\ \cline{3-5} \cline{7-9}
 & &$\frac{Q}{\Delta}$ value & $\kappa_{worst}$ & Worker time & & $\frac{Q}{\Delta}$ value & $\kappa_{worst}$ & Worker time \\ 
 \midrule
Poly Code  \cite{yu2017polynomial} &  & N/A & $8.8 \times 10^3$ & $2.46$ &  & $1$ & $1.86 \times 10^{13}$ & $7.09$\\
Ortho-Poly \cite{8849468}  & & N/A   & $16.66$ & $2.49$ &  & $1$ & $4.33 \times 10^{5}$ & $7.06$\\
RKRP Code\cite{8919859} & & N/A & $11.96$ & $2.41$ &  & $1$ & $1.16 \times 10^4$ & $7.14$\\
SCS Optimal Scheme &   & - & - & - & & $1.02$ & $2.15 \times 10^3$ & $2.04$\\
\bottomrule
\end{tabular}
\end{sc}
\end{small}
\end{center}
\end{table*}

\section{Conclusions and Future Work}
\label{sec:conclusion}
In this work we have presented several coded matrix computation schemes that (i) leverage partial computations by stragglers and (ii) impose constraints on the extent to which coding is allowed in the solution. 

The second feature is especially valuable in the practical case of computations with sparse matrices and provides significant reductions in worker node computation time and better numerical stability as compared to the previous schemes. Prior work has demonstrated schemes with optimal recovery threshold in certain cases. We present schemes that match the optimal threshold while enjoying lower worker node computation times and improved numerical stability. Exhaustive numerical experiments corroborate our findings.

There are several opportunities for future work. We have demonstrated that carefully chosen different parallel classes provide improved recovery thresholds and $Q/\Delta$ metrics for the matrix-vector problem. We expect that this should help even in the case of matrix multiplication. Schemes that apply for a larger range of storage fractions are also of interest. In this work we defined the value of $Q$ as the worst case number of symbols that allows for recovering the intended result. Analysis and constructions for the random case may be of interest.

\section{Acknowledgments}
The authors acknowledge interesting conversations with Dr. Li Tang and his participation in \cite{c3les}.
\appendix

\subsection{Properties of $\beta$-level Coding when $\beta > c$}
\label{app:diffc}
In Section \ref{sec:beta_level}, we have discussed $\beta$-level coding for distributed matrix computations when $c \geq \beta$ and here we prove the properties of $\beta$-level coding when $\beta > c$. The difference is that the constraint $c \geq \beta$ ensures that we will have at least $\beta$ worker groups, whereas it is not the case when $\beta > c$.

\subsubsection{Matrix-vector Multiplication}
Suppose that we have $n = c a_2$ workers, each of which can store $\gamma = \frac{a_1}{a_2}$ fraction of matrix $\bfA$. To incorporate $\beta$-level coding, matrix $\bfA$ is partitioned into $\Delta = \beta a_2$ block-columns, and thus each worker will be assigned $\ell = \Delta \gamma = \beta a_1$ jobs. It should be noted that we have $\frac{n}{\Delta/\beta} = c$ worker groups among the workers because of the cyclic fashion of job assignments.

%where $\Delta$ needs to be a multiple of $\beta$. We assume $\Delta = \rho \beta a_2$ ($\rho$ divides $c$), and we need to find $\rho$ such that we can optimize $s$ and $Q/\Delta$. 

\begin{lemma}
\label{lem:QDeltaarb}
If we use a single parallel class in Alg. \ref{Alg:betalevelcoding}, then the number of stragglers will be $s = c \ell - \beta$ and we will have
\begin{align*}
    Q = c \left[ \frac{\Delta}{\beta}\ell - \frac{\ell(\ell+1)}{2} \right] + c \sum\limits_{i=0}^{c_1 - 1} (\ell - i) + c_2 (\ell - c_1) + 1 
\end{align*} where $c_1 = \floor{\frac{\beta - 1}{c}}$ and $c_2 = \beta - 1 - c c_1$.
\end{lemma}
\begin{proof}
The straggler resilience follows similar to the proof of Theorem \ref{thm:beta_matvecstrQ} by counting the  number of occurrences of the meta-symbols.
%Here we have $\frac{n}{\Delta/\beta} = c$ worker groups. If we use a single parallel class $\calP$ for assigning the jobs to the workers, we can say that in any worker group, because of the cyclic fashion, there are $\ell$ meta-symbols corresponding to any element of $\calP$. Thus from $c$ worker groups, we can have, $c \ell$ such meta-symbols corresponding to any element, but the master requires only $\beta$ symbols to recover the corresponding unknowns, thus we can say that the scheme will be resilient to $c \ell - \beta = c \beta a_1 - \beta$ stragglers.

For the $Q$ analysis, assume that there exists a meta-symbol $\star$ that appears at most $\beta-1$ times among the acquired $Q$ symbols where $Q$ is defined in the theorem statement. We have $c$ worker groups and in each group, $\star$ appears in positions $0, 1, 2, \dots, \ell - 1$. 

Now we know that we can process $\alpha_0 = \left[ \frac{\Delta}{\beta}\ell - \frac{\ell(\ell+1)}{2} \right]$ meta-symbols from each of the worker groups without processing $\star$. Any additional processing will necessarily process $\star$. Suppose we choose any particular worker, where the position index of $\star$ is $i$. In that case, we can acquire at most $\ell - 1 - i$ more symbols from that particular worker without any more appearances of $\star$.
%Thus with only one appearance of $*$, the maximum number of symbols that can be acquired from the workers are given by the vector
Thus, the maximum number of meta-symbols that can be processed for each additional appearance of $\star$ can be expressed by the following vector.
%\begin{align*}
%    \bfz =\left( \undermat{c}{\ell - 1, \ell - 1, \dots, \ell - 1}, \undermat{c}{\ell - 2, \dots, \ell - 2}, \dots, \undermat{c}{0, 0, \dots, 0}\right) .\\
%\end{align*} 
\begin{align*}
    \bfz =\left( \undermat{c}{\ell, \ell, \dots, \ell}, \undermat{c}{\ell - 1, \dots, \ell - 1}, \dots, \undermat{c}{1, 1, \dots, 1}\right) .\\
\end{align*}
Here $\bfz$ is a non-increasing sequence, so in order to obtain the maximum number of symbols where the meta-symbol $\star$ appears at most $\beta - 1$ times, we need to acquire symbols sequentially as mentioned in $\bfz$. Let $c_1 = \floor{\frac{\beta - 1}{c}}$ and $c_2 = \beta - 1 - c c_1$. Thus we can choose the first $c c_1 + c_2 = \beta - 1$ workers (as mentioned in $\bfz$) so that we can have $Q^{'}$ symbols where $\star$ appears exactly $\beta - 1$ times, so 
\begin{align*}
 Q^{'} = c \alpha_0 + c \sum\limits_{i=0}^{c_1 - 1} (\ell - i) + c_2 (\ell - c_1);
\end{align*} which indicates that $Q = Q^{'} + 1$ symbols ensures that $\star$ will appear at least $\beta$ times. This leads to a contradiction and concludes the proof.
\end{proof}

\subsubsection{Matrix-matrix Multiplication}
The argument is almost the same for the matrix-matrix case with appropriate definitions for $\ell$ and $\beta$.
Specifically, recall that $n = c \times a_2 b_2$, and $\Delta_A = \beta_A a_2$ and $\Delta_B = \beta_B b_2$. Thus, we have $\frac{n}{\Delta/\beta} = c$ worker groups, where $\Delta = \Delta_A \Delta_B$ and $\beta = \beta_A \beta_B$. In each worker, we assign $\ell_A = \Delta_A \gamma_A$ and $\ell_B = \Delta_B \gamma_B$ coded submatrices of $\bfA$ and $\bfB$, respectively and set $\ell = \ell_A \ell_B$. Following this, we can obtain the number of stragglers as $s = c \ell - \beta$ and 
\begin{align*}
    Q = c \left[ \frac{\Delta}{\beta}\ell - \frac{\ell(\ell+1)}{2} \right] + c \sum\limits_{i=0}^{c_1 - 1} (\ell - i) + c_2 (\ell - c_1) + 1.
\end{align*}

%\subsection{Proof of Lemma \ref{lem:gdec}}
%\label{app:G_dec}
%\aditya{this subsection can be removed}
%Note that $\bfG_{dec}$ specifies a system of equations in $\Delta$ unknowns. We need to argue that this system is invertible.
%In the argument below, suppose that the random linear coefficients of each meta-symbol are indeterminates.

%We argue that there exists a matching in $\bfG_{dec}$ where all the unknowns are matched. Towards this end, note that any subset $S$ of the unknowns has at least $2|S| - 1$ outgoing edges. Since each meta-symbol has degree-2, this implies that the neighborhood of $S$ is at least of size $\lceil |S| - 1/2 \rceil = |S|$. Thus, by the Hall's marriage theorem \cite{marshall1986combinatorial}, there exists a matching where each unknown is matched to a corresponding meta-symbol. This, implies that the system of equations specified by $\bfG_{dec}$ is such that there is an assignment of values to the indeterminate linear combination coefficients such that the system of equations is invertible, i.e., the determinant of the corresponding matrix is not identically zero. This in turn implies that the system of equations is invertible with probability $1$ when the coefficients are chosen at random.

%\subsection{Khatri-Rao argument}

\subsection{Proof of Theorem \ref{thm:beta2_matvecstrQ}}
\label{App:diff_beta_2}
\begin{proof}
{\bf Straggler Resilience:} To prove the straggler resilience, we note that if there are at $2\ell-1$ stragglers it is evident that $\bfG_{dec}$ formed by the remaining meta-symbols is such that each unknown has degree at least one. Let $X_i$ and $X_j$ denote the subset of worker nodes where unknowns $\bfA_i^T \bfx$ and $\bfA_j^T \bfx$ appear within a meta-symbol, so that $|X_i|=|X_j| = 2\ell$. Furthermore, $|X_i \cap X_j| \leq 2 \ell -2$. To see this we note that if $\{i,j\}$ appear together w.l.o.g. in $\calG_0$ then $|X_i \cap X_j| = \ell + \ell-2$ as this implies that they appear together in exactly $\ell-2$ workers in $\calG_1$ (since $\ell \leq \Delta/2-2$). On the other hand if $i$ and $j$ do not appear together in either $\calG_0$ or $\calG_1$ then they appear together in the workers of each group at most $\ell-1$ times, so the claim holds. Thus, 
\begin{align*}
|X_i \cup X_j| &= |X_i| + |X_j| - |X_i \cap X_j|\\
&\geq 2\ell + 2.
\end{align*}
Now suppose by way of contradiction that we have two unknowns $\bfA_i^T \bfx$ and $\bfA_j^T \bfx$ (where $i < j$) both of which appear exactly once across the remaining $n- 2\ell + 1$ workers.  
The preceding argument shows that if $2\ell-1$ workers are stragglers then unknowns $\bfA_i^T \bfx$ or $\bfA_j^T \bfx$ or both appear in at least three nodes, i.e., at least one of them appears at least twice. This contradicts our original assumption. By Lemma \ref{lem:allbeta} the decoding is successful.

{\bf Value of $Q$:} Note that $\alpha_0 = \frac{\Delta}{2}\ell - \frac{\ell(\ell+1)}{2}$ denotes the maximum number of meta-symbols that can be processed within a group such that a specific meta-symbol is not processed ({\it cf.} Lemma \ref{lem:cyclicQ}). This implies that at most $2 \alpha_0$ meta-symbols can be processed without processing any specific unknown. Let $\rho_0$ and $\rho_1$ denote the number of meta-symbols processed in the two groups $\calG_0$ and $\calG_1$ where we assume w.l.o.g. that $\rho_0 \geq \rho_1$. 
\begin{itemize}
    \item Case 1: Suppose that $\rho_0 \geq \alpha_0 + \ell + 1$. Lemma \ref{lem:cyclicQ} implies that each meta-symbol $\in \calP_0$  is processed at least twice in $\calG_0$ . Then by Lemma \ref{lem:allbeta}, the decoding is successful.
    \item Case 2: If $\alpha_0 + 2 \leq \rho_0 \leq \alpha_0 + \ell$, we claim that at most one meta-symbol in $\calG_0$ is processed once. The other meta-symbols are processed at least twice. To see this, consider two meta-symbols $(2i,2i+1)$ and $(2j,2j+1)$ in $\calG_0$ such that $j > i$ such that $(2i,2i+1)$ is processed only once. If $j - i \geq 2$ then there are at least two workers in $\calG_0$ where the meta-symbol $(2j,2j+1)$ appears but $(2i,2i+1)$ does not. Therefore, if at least $\alpha_0+1$ meta-symbols are processed in $\calG_0$, then $(2j,2j+1)$ appears at least twice. On the other hand if $j=i+1$ then there is only one worker where $(2j,2j+1)$ appears but $(2i,2i+1)$ does not. Thus, if $\alpha_0+1$ meta-symbols are processed then we have processed $(2j,2j+1)$ at least once. The $\alpha_0 +2$-th meta-symbol cannot be $(2i,2i+1)$ since by assumption it is processed only once, thus it has to be $(2j,2j+1)$ (since $j=i+1$).
    
    Now, we argue either unknown $\bfA_{2i}^T \bfx$ or $\bfA_{2i+1}^T \bfx$ appear within the meta-symbols in $\calG_1$. Towards this end, we note that there are exactly two workers in $\calG_1$ where $\bfA_{2i+1}^T \bfx$ appears but $\bfA_{2i}^T \bfx$ does not. Therefore, at most $\alpha_0 - (2\ell-1)$ meta-symbols can be processed in $\calG_1$ while avoiding both the unknowns $\bfA_{2i}^T \bfx$ and $\bfA_{2i+1}^T \bfx$.
    
    This implies that the total number of meta-symbols that can be processed such that at least two unknowns appear only once in $G_{dec}$ is at most $2\alpha_0 -\ell + 1 < Q$.
    
    \item Case 3: If $\rho_0 = \alpha_0 + 1$, then we can have two meta-symbols $(2i,2i+1)$ and $(2i+2,2i+3)$ that appear exactly once in $\calG_0$. It can be verified that none of the unknowns $\bfA_{2i}^T \bfx, \dots, \bfA_{2i+3}^T \bfx$ appear together in a meta-symbol in $\calG_1$ since $\Delta \geq 8$. Thus, if we process $\rho_1 = \alpha_0$ symbols in  $\calG_1$, then we can avoid at most one unknown from the set $\{\bfA_{2i}^T \bfx, \dots, \bfA_{2i+3}^T \bfx\}$. It follows that at most one unknown appears once in $G_{dec}$ and by Lemma \ref{lem:allbeta}, the decoding is successful.
\end{itemize}

\end{proof}

\subsection{Proof of Theorem \ref{thm:Q_partially_coded_end}}

\label{App:coded_bot}

\begin{proof}
We need to show that for any pattern of $Q$ symbols the master node can decode $\bfA^T \bfx$. Towards this end, from Theorem \ref{thm:beta_matvecstrQ} (setting $\beta = 1$ and $\ell = \ell_u = r_u$), we know that any pattern of $Q$ uncoded symbols allows the recovery of all $\Delta$ unknowns. In other words for any computation state vector $ \bfw(t) =[w_0(t)~w_2(t)~ \dots~w_{n-1}(t)]$ such that $w_i(t) \leq \ell_u$ and $\sum_{i=0}^{n-1} w_i(t) \geq Q$, the master node can decode. Now, consider a vector $\bfw'(t)$ such that (w.l.o.g.) $w'_0(t), \dots, w'_{\alpha-1}(t) \geq \ell_u + 1$ and $w'_{\alpha}(t), \dots, w'_{n-1}(t) \leq \ell_u$ and $\sum_{i=0}^{n-1} w'_i(t) \geq Q$, i.e., the first $\alpha$ worker nodes process coded blocks whereas the others do not. It is not too hard to determine a different vector $\tilde{\bfw}(t)$ with the following properties.
\begin{align*}
\tilde{w}_i(t) &= \begin{cases}
\ell_u & ~1 \leq i \leq \alpha,\\
w'_i(t) + \beta_i & \alpha+1 \leq i \leq n,
\end{cases}
\end{align*}
where $\beta_i$'s are positive integers such that $w'_i(t) + \beta_i \leq \ell_u$ and $\sum_{i=0}^{n-1} \tilde{w}_i(t) = Q$. Thus, $\tilde{\bfw}(t)$ corresponds to a pattern of $Q$ uncoded blocks that recovers $\Delta$ distinct blocks.

Now, we compare the vectors $\bfw'(t)$ and $\tilde{\bfw}(t)$. Let the uncoded symbols in $\bfw'(t)$ be denoted by the set $\calA$. Then the set of uncoded symbols in $\tilde{\bfw}(t)$ can be expressed as $\calA \cup \calB$ where the set $\calB$ results from the transformation above. It is evident that for computation state vector $\bfw'(t)$ the master node has $\sum\limits_{i=0}^{\alpha-1} \left( w'_i(t) - \ell_u \right)$ equations with $\Delta - |\calA|$ variables. Now,
\begin{align*}
\sum\limits_{i=0}^{\alpha-1} \left( w'_i(t) - \ell_u \right) \geq |\calB|
&\geq |\calB \setminus \calA| = \Delta - |\calA|.
\end{align*}
In particular, this establishes that we have at least as many equations as variables. Since any square submatrix of a random matrix is invertible with probability 1, we have the required result.

%----
Next, we establish the straggler resilience of our scheme. Consider worker nodes $0 \leq i_1 < i_2 < \dots < i_k \leq n-1$; each of these worker nodes has $\ell_u$ uncoded symbols. Consider the case that $i_t - i_{t-1} < \ell_u$ for $t= 2, 3, \dots, k$. We claim that these worker nodes contain at least $\min (\ell_u+k-1, \Delta)$ distinct uncoded symbols. To see this we proceed inductively. Let $X_{i_j}$ denote the symbols in worker $i_j$. If $k=2$, then $|X_{i_1} \cup X_{i_2}| = |X_{i_1}| + |X_{i_2}| - |X_{i_1} \cap X_{i_2}| \geq 2\ell_u - (\ell_u -1) = \ell_u + 1$. We assume the inductive hypothesis, i.e, $|X_{i_1} \cup \dots \cup X_{i_{k-1}}| \geq \min(\ell_u + k-2, \Delta)$. 

Now consider $|X_{i_1} \cup \dots \cup X_{i_{k-1}} \cup X_{i_k}|$. It can be observed that if $i_{k-1} + \ell_u - 1 - \Delta < i_1$ then there exists at least one symbol in $X_{i_k}$ that does not exist in $X_{i_1} \cup \dots \cup X_{i_{k-1}}$. Thus, in this case $|X_{i_1} \cup \dots \cup X_{i_{k-1}} \cup X_{i_k}| \geq \ell_u + k -1$.

On the other hand if $i_{k-1} + \ell_u - 1 - \Delta \geq i_1$ then $X_{i_k} \subseteq X_{i_{k-1}} \cup X_{i_1}$. Let $\delta$ be the smallest integer such that $i_\delta + \ell_u -1 - \Delta \geq i_1$ but $i_{\delta-1} + \ell_u -1 - \Delta < i_1$. In this case, we have $|X_{i_1} \cup X_{i_2} \dots \cup X_{i_{\delta-1}}| = \ell_u + \sum_{x=2}^{\delta-1} (i_x - i_{x-1})$. Furthermore, $X_{i_\delta}$ contributes another $\Delta + i_1 - (i_{\delta-1} + \ell_u - 1) - 1$ symbols so that 
\begin{align*}
|X_{i_1} \cup \dots \cup X_{i_{\delta-1}} \cup X_{i_\delta}| = \ell_u + \sum_{x=2}^{\delta-1} (i_x - i_{x-1}) + \Delta + i_1 - (i_{\delta-1} + \ell_u - 1) - 1 = \Delta.
\end{align*}
Thus, there is nothing to prove in this case.

On the other hand, suppose that $1 \leq \alpha \leq k$ is the least value such that $i_{\alpha} - i_{\alpha -1} \geq \ell_u$.  In this case, we know from the above claim that $|X_{i_1} \cup \dots \cup X_{i_{\alpha-1}}| \geq \ell_u + \alpha - 2$. It follows that $X_{i_{\alpha}}, \dots X_{i_{k}}$ each contribute at least one new symbol, namely $i_{\alpha}, \dots, i_{k}$. Therefore $|X_{i_1} \cup \dots \cup X_{i_{k}}| \geq \ell_u + \alpha - 2 + k - \alpha + 1 = \ell_u + k -1$.

Thus, if we think about choosing $k$ workers, then we need to ensure that
\begin{align*}
\ell_u + (k - 1) + k (\ell - \ell_u) \geq \Delta
\end{align*} which further implies
\begin{align*}
k \geq \frac{n - \ell_u + 1}{\ell - \ell_u + 1} = \frac{n - n \gamma_u + 1}{n \gamma - n \gamma_u + 1}
\end{align*} as $n = \Delta$. So, if the system is resilient to $s$ stragglers then
\begin{align*}
s \leq \floor[\bigg] {n - \frac{n - n \gamma_u + 1}{n \gamma - n \gamma_u + 1}} = \floor[\bigg]{\frac{n^2 \gamma_c + n \gamma_u - 1}{n \gamma_c + 1}}.
\end{align*} It should be noted that setting $\gamma_c = 0$ leads to the uncoded case which is resilient to $(n \gamma - 1)$ workers (same as setting $\beta = 1$ in Theorem \ref{thm:beta_matvecstrQ}).
\end{proof}

\subsection{Proof of Theorem \ref{str:coded_bottom_matmat}}
\label{App:coded_matmat}
\begin{proof}

To prove the theorem by contradiction, we assume that there exists an unknown $\bfA_i^T \bfB_j$, which cannot be decoded from a particular set of $\tau$ workers where $\tau$ is defined in the theorem statement. We consider the set $\calB_j = \{\bfA_0^T \bfB_j, \bfA_1^T \bfB_j, \dots, \bfA_i^T \bfB_j, \dots, \bfA_{\Delta_A-1}^T \bfB_j \}$, i.e, the set of all unknowns corresponding to $\bfB_j$, for $j = 0, 1, \dots, \Delta_B - 1$, thus $|\calB_j| = \Delta_A = a_2$. It should be noted that the equations consisting of the unknowns of $\calB_j$ are disjoint with the equations consisting of the unknowns of $\calB_m$, ($j \neq m$) since the assigned submatrices from $\bfB$ are uncoded. 

Let $\calS_j$ denote the set of workers where $\bfB_j$ does not appear in the assignments and $\calT_j$ denote the set of workers where it appears. According to the scheme in  Alg. \ref{Alg:Cyclic_Partially_Uncoded_matmat}, there are $m a_2 b_1$ workers each of which has an uncoded copy of $\bfB_j$. Thus, $|\calS_j| = n - m a_2 b_1$. %  since each worker stores $\ell_B = m b_1$ uncoded submatrices from $\bfB$, and thus $|\calS_j| = n - m a_2 b_1$.

Next, partition the workers of $\calT_j$ into $\ell_B = m b_1$ worker groups, within each of which, all $a_2$ uncoded block-columns of $\bfA$ appear in a cyclic fashion. From the proof of Theorem \ref{thm:Q_partially_coded_end}, we know that any $k$ workers within a group will provide $\min(a_u + k - 1, a_2)$ uncoded symbols corresponding to $\calB_i$. Now we have $\ell_B$ such worker groups which indicates that we have $\ell_B$ workers of $\calT_j$ which have the same uncoded job assignments. Thus, from any $\kappa$ workers of $\calT_j$, we will obtain $\min(a_u + \ceil{\frac{\kappa}{m b_1}} - 1, a_2)$ uncoded symbols, and $\kappa a_c$ coded symbols. So in order to be able to decode the elements of $\calB_j$, we need to find the minimum positive integer for $\kappa$ (which is denoted as $\kappa_{min}$) such that
\begin{align*}
a_u + \ceil[\bigg]{\frac{\kappa}{m b_1}} - 1 + \kappa a_c  \geq a_2 .
\end{align*} It indicates that any $\kappa_{min}$ workers of $\calT_j$ are enough to recover all the elements of $\calB_j$ including $\bfA_i^T \bfB_j$. But $\tau - |\calS_j| = \kappa_{min}$, which leads to a contradiction and hence concludes the proof.

\end{proof}

\subsection{Concatenation of block rows in $\tilde{\bfV} \odot \tilde{\bfU}$}
\label{app:argumentKR}
Let $\bfU$ denote a $\ell_c k_A \times \ell_c k_A k_B$ matrix whose $\delta$-th row is given by $[u^{(i_0,0)}_\delta ~\dots~u^{(i_0,\ell_c-1)}_\delta| ~\dots~ |u^{(i_{k-1},0)}_\delta ~\dots~u^{(i_{k-1},\ell_c-1)}_\delta]$ where we recall that each entry of $\bfU$ is chosen i.i.d. at random from a continuous distribution and $k= k_A k_B$. The matrix $\bfU$ can be written as \begin{align*}
    \bfU &= [\bfU_0~|~ \bfU_1~|~ \dots ~|~\bfU_{k-1}]
\end{align*}
where each $\bfU_j$ is of dimension $\ell_c k_A \times \ell_c$. We wish to show that %\aditya{check notation, should be superscript for v}
\begin{align*}
    \begin{bmatrix}
    \bfU_0 \otimes v^{(i_0)} \;|\; \bfU_1 \otimes v^{(i_1)} \; | \;  ~\dots~ | \; \bfU_{k-1} \otimes v^{(i_{k-1})}
    \end{bmatrix}
\end{align*}
is full-rank with probability $1$. %\anindya{Do we need the Hadamard product with $U$ ?}

Note that the vectors $v^{(i_\ell)}$'s are also chosen at random and any collection of $k_B$ such vectors is full rank with probability $1$. In the argument below we show a specific choice of $\bfU$ that yields a full-rank matrix. This implies that the matrix continues to be full-rank under the random choice. Towards this end, we pick the first $\ell_c$ rows of $\bfU$ to be
\begin{align*}
\begin{bmatrix}
\bfI_{\ell_c} & \dots & \bfI_{\ell_c} & \mathbf{0} & \dots & \mathbf{0}
\end{bmatrix}, 
\end{align*}
i.e., the first $k_B$ block-columns are identity matrices. It can be seen that these result in $\ell_c k_B$ linearly independent rows. The next block row of $\bfU$ is a $k_B$ block-column shifted version of the first block row, i.e., it is
\begin{align*}
\begin{bmatrix}
\mathbf{0} & \dots & \mathbf{0} & \bfI_{\ell_c} & \dots & \bfI_{\ell_c} & \mathbf{0} & \dots & \mathbf{0}
\end{bmatrix}
\end{align*}
This yields another $\ell_c k_B$ linearly independent block rows. This process can be repeated $k_A$ times to provide the required result.

\subsection{Number of $\bfA_i$'s that appear less than $k_B$ times within the stragglers}
%\subsection{Auxiliary proof}
\label{aux_argument}
In the setting of Theorem \ref{Alg:Optimal_Matmat}, suppose that we have $n-k_A k_B$ stragglers that together contain $(n-k_A k_B)\Delta_A k_B/n$ uncoded block-columns of $\bfA$. We want to show that not all $\bfA_i$'s appear $k_B$ times within the stragglers. To see this consider a bipartite graph that specifies the placement of the uncoded block-columns of $\bfA$. It contains vertices denoting the $\bfA_i$'s and the worker nodes. An edge connects $\bfA_i$ and $W_j$ if $\bfA_i$ appears in $W_j$. Thus each $\bfA_i$ has degree $k_B$.
It can be seen that this graph is connected as any two neighboring workers $W_j$ and $W_{j+1}$ (indices reduced modulo-$n$) have block-columns in common.
Suppose that the stragglers are such that each $\bfA_i$ that appears within the stragglers also appears $k_B$ times within the stragglers. This implies that the subgraph induced by the stragglers is such that it disconnected from the remaining workers. This is a contradiction.

%\input{approach_for_sparse_matvec}

\ifCLASSOPTIONcaptionsoff
  \newpage
\fi



%\bibliographystyle{IEEEtran}

\bibliographystyle{IEEEtran}
\bibliography{citations}


%\end{IEEEbiography}


%\newpage



%\vfill


\end{document}